\pdfoutput=1
\documentclass[a4paper,USenglish]{lipics-v2021}

\hideLIPIcs
\nolinenumbers 

\usepackage{amsmath}
\usepackage{amsfonts}
\usepackage{amsthm}
\usepackage{subcaption}
\usepackage{graphicx}
\usepackage{xcolor}
\usepackage{algorithm}
\usepackage[noend]{algpseudocode}
\usepackage{comment}
\usepackage[title]{appendix}
\usepackage{tikz-cd}
\usetikzlibrary{patterns}

\newcommand{\ex}{\mathbb{E}}

\newcommand{\mathbbNzero}{\mathbb{N}_{\geq 0}}

\title{Foraging in Particle Systems via Self-Induced Phase Changes}
\date{}

\author{Shunhao Oh}{School of Computer Science, Georgia Institute of Technology}{ohoh@gatech.edu}{}{}

\author{Dana Randall}{School of Computer Science, Georgia Institute of Technology}{randall@cc.gatech.edu}{}{}

\author{Andr\'ea W.\ Richa}{School of Computing and Augmented Intelligence, Arizona State University}{aricha@asu.edu}{}{}

\authorrunning{S. Oh, D. Randall and A. W. Richa} 

\Copyright{Shunhao Oh, Dana Randall and Andr\'ea W.\ Richa} 



\relatedversion{} 

\EventEditors{John Q. Open and Joan R. Access}
\EventNoEds{2}
\EventLongTitle{42nd Conference on Very Important Topics (CVIT 2016)}
\EventShortTitle{CVIT 2016}
\EventAcronym{CVIT}
\EventYear{2016}
\EventDate{December 24--27, 2016}
\EventLocation{Little Whinging, United Kingdom}
\EventLogo{}
\SeriesVolume{42}
\ArticleNo{23}
\begin{document}
\maketitle

\begin{abstract}
We study {\em the foraging problem}, where rudimentary particles with very limited computational, communication and movement capabilities reside on the triangular lattice and
a  ``food'' particle may be placed at any point, removed, or shifted  at arbitrary times,  possibly adversarially. We would like the particles  to consistently self-organize to find and tightly gather around the food particle using only local interactions.  
If a food particle remains in a position for a long enough time, we want the particles to enter a {\em gather phase}, in which they collectively form a single large component with low perimeter and compress around the food particle.
On the other hand, if no food particle has existed recently, we want the particles to undergo a self-induced phase change and switch to a {\em search phase} in which they distribute themselves randomly throughout the lattice region to search for food.
Unlike previous approaches, we want this process to repeat forever - the challenge is to share information locally and autonomously so that eventually most particles learn about and correctly react to the current phase,  so  the algorithm can withstand repeated waves of phase changes 
that may interfere with each other. In our model, the particles themselves decide the appropriate phase based on learning the environment and, like a true phase change,  we want microscopic changes to a single food particle to trigger these macroscopic, system-wide transitions.
The algorithms in this paper are the first to leverage {\em self-induced phase changes as an algorithmic tool}.

In this paper, we present two rigorous solutions that allow the system to iteratively recover from a change in food status.  First, we give an algorithm 
for foraging  that 
is fully stochastic and builds upon the compression algorithm of Cannon et al.\ (PODC'16), motivated by a phase change that occurs in the fixed magnetization Ising model from statistical physics.
When the food is present, particles incrementally  enter a gather phase where they attach and compress around the food; when the food disappears or moves, the particles transition to a search phase where they relax their ferromagnetic attraction 
and instead move according to a simple exclusion process that causes the particles to disperse and explore the domain in search of new food. Second, we present a highly structured alternative algorithm that gathers 
by incrementally building a minimal spiral tightly wrapped around the food particle until the food is depleted. 
A key component of our algorithms is a careful token passing mechanism that ensures that the rate at which the compressed cluster around the food dissipates outpaces the rate at which the cluster may continue to grow, ensuring that the broadcast wave
of the ``dispersion token'' will always outpace that of a compression wave. 




\end{abstract}

\section{Introduction} 
Collective behavior of interacting agents is a fundamental, nearly ubiquitous phenomenon across fields,  reliably producing rich and complex coordination.  Examples at the micro- and nano-scales include coordinating cells (including our own immune system or self-repairing tissue and bacterial colonies), 
micro-scale swarm robotics, and interacting particle systems in physics; at the macro scale it can represent flocks of birds,  coordination of drones,  and societal dynamics such as segregation \cite{Schelling1971}.
Common properties of these disparate systems is that they 1) respond to simple environmental conditions and 2) many undergo {\it phase changes} as parameters of the systems are slowly modified,  allowing collectives to gracefully toggle between two often dramatically different macroscopic states.
Insights from statistical physics, discrete probability, and computer science already have been useful for rigorously describing limiting behavior of large systems undergoing such phase changes in a surprising variety of contexts,  including chemistry \cite{Miracle2011}, biology \cite{alberti},  economics \cite{BMR},  swarm robotics \cite{Li2021-bobbots},  and
physics \cite{lubetzky2013}.
Here, we seek to understand how to design distributed algorithms for collective coordination, such as foraging, by algorithmically inducing phase changes in interacting systems, where the desired macrostate is driven by some small, local environmental signal such as the presence or absence and location of a dynamically changing food source.

Self-organizing collective behaviors are found throughout nature,  e.g.,  shoals of fish aggregating to intimidate predators~\cite{Magurran1990}, fire ants forming rafts to survive floods~\cite{Mlot2011}, and bacteria forming biofilms to share nutrients when metabolically stressed~\cite{Liu2015,Prindle2015}. Inspired by such systems, researchers in swarm robotics and active matter have used many approaches towards enabling ensembles of simple, independent units to cooperatively accomplish complex tasks~\cite{Sahin2005}.
Current approaches present inherent tradeoffs, especially as individual robots become smaller and have limited functional capabilities~\cite{Xie2019} or approach the thermodynamic limits of computing~\cite{Wolpert2019} and power~\cite{Dario1992}. 

{\it Programmable matter} was coined by Toffoli and Margolus  in 1991 to realize a physical computing medium composed of simple,  homogeneous entities that is scalable,  versatile, instantly reconfigurable, safe to handle, and robust to failures~\cite{Toffoli1991}. There are formidable challenges to realizing programmable matter computationally and many researchers in distributed computing and swarm and modular robotics have investigated how such small, simply programmed entities can coordinate to solve complex tasks and exhibit useful emergent collective behavior.  A more ambitious goal is to design programmable matter systems of {\em self-actuated} individuals that autonomously respond to their environment.

To better understand how to model such collective behaviors, the work in~\cite{Cannon2016,Cannon2019,Arroyo2018,Li2021-bobbots}, among others,  uses  {\it self-organizing particle systems} (or SOPS) that allows us to define a formal distributed algorithm and rigorously quantify long-term behavior. {\em Particles} in a SOPS exist on the vertices (or nodes) of a lattice, with at most one particle per vertex, and move between vertices along lattice edges. Each particle is anonymous (unlabeled), interacts only with particles on adjacent lattice vertices, has limited (normally assumed to be constant-size) memory, and does not have access to any global information such as a coordinate system or the total number of particles. 
Recently, particle systems that leverage phase transitions in basic physical systems have proven fruitful for driving basic programmable matter 
systems to accomplish simple tasks, but these tend not to be self-induced or adaptable depending on the current need.  
As an example, a problem of interest is
{\it dynamic free aggregation}, where particles gather together without preference for a specific aggregation site (see Section 3.2.1 of~\cite{Bayindir2016}), and {\it dispersion}, its inverse, have been widely studied, but without much effort toward understanding the underlying distributed computational models or the formal algorithmic underpinnings of interacting particles. Many studies in self-actuated systems take inspiration from emergent behavior in social insects, but either lack rigorous mathematical foundations explaining the generality and limitations as sizes scale (see, e.g., \cite{Garnier2005,Garnier2009,Correll2011}) or rely on long-range signaling, such as microphones or line-of-sight sensors \cite{Soysal2005,Fates2010,Fates2011,Ozdemir2018}.

Cannon et al$.$ \cite{Cannon2016} designed a distributed Markov chain algorithm for dynamic free aggregation in the connected setting, referred to as {\em compression}, where a connected SOPS wants to cluster in a configuration of small perimeter while maintaining connectivity at all times. 
Building on a variant of the Ising model from statistical physics, the authors rigorously proved that by adding a ferromagnetic attraction $\lambda$ between particles suffices to stochastically lead the particles in a SOPS to
$\alpha$-{\em compressed} configurations, where the constant $\alpha>1$ determines an upper bound on the ratio of the configuration perimeter by the minimum possible system perimeter.
The Markov chain is defined so that each configuration $\sigma$ appears with probability $\pi(\sigma) = \lambda^{|E(\sigma)|}/Z$ at stationarity, where $|E(\sigma)|$ is the number of edges in $\sigma$ and $Z$ is the normalizing constant.
It is rigorously shown that the SOPS will reach an $\alpha$-compressed configuration at stationarity, for some constant $\alpha>1$, if the attraction force $\lambda$ is strong enough and not if the attraction is sufficiently small.  Moreover, it is also shown that when the attraction forces are  small, the configurations will nearly maximize their perimeter and achieve {\it expansion}. Additional algorithms based on this approach, which also exhibit similar phase changes, found applications to the {\em separation} problem~\cite{Cannon2019}, where colored particles would like to separate into clusters of same color particles, {\em aggregation}~\cite{Li2021-bobbots} where the system would like to compress but no longer needs to remain connected, shortcut bridging~\cite{Arroyo2018}, locomotion~\cite{Savoie2018}, alignment~\cite{kedia2022} and transport~\cite{Li2021-bobbots}. 

We are interested in whether such  phase changes (or bifurcations)  can be self-modulated, whereby a simple parameter governing the system behavior can self-regulate according to environmental cues that are communicated through the collective to induce desirable system-wide behaviors. 
%
Our goal here is to self-regulate a system-wide adjustment in the bias parameters when one or more particles notice the presence or depletion of food to induce the appropriate global coordination to
provably transition the collective between macro-modes when required.  

\vskip.1in
\underbar{\bf The Foraging Problem:} \ 
In the {\em foraging} problem, ``ants'' (i.e., particles) may initially be 
searching for ``food'' (i.e., any resource in the environment, e.g., an energy source); once a food source is found, the ants that have learned about the food source start informing other ants, and these transition to a {\em compression state}
that leads them to start to gather around the food source to consume the food, i.e., to incrementally enter a {\em gather phase}; however, once the source is depleted, the ants closer to the depleted source start a broadcast wave, gradually informing other ants that they should restart the {\em search phase} again by individually switching to an internal {\em dispersion state}.
The foraging problem is very general and has several fundamental application domains, including search-and-rescue operations in swarms of nano- or micro-robots; health applications (e.g., a collective of nano-sensors that could search for, identify, and gather around a foreign body to isolate or consume it, then resume searching, etc.); and  finding and consuming/deactivating hazards in a nuclear reactor or a minefield.

It is appealing to  use the   { compression/expansion algorithm} of~\cite{Cannon2016}, where a phase change  corresponds to a {\em system-wide change of the bias parameter $\lambda$}:  Small $\lambda$, namely $\lambda< 2.17$, provably corresponds to {\it expansion} mode, which is desirable to search for food,  while  large $\lambda$, namely $\lambda>2+\sqrt{2}$, corresponds to {\em compression (or gather)} mode, desirable when food has been discovered.  (A similar phase change occurs when the particles are allowed to disconnect, as explained in Section~\ref{sec:model}, but the notion of compression becomes more complicated.)
Our goal is to perform a system-wide adjustment in the bias parameters when one or more particles notice the presence or depletion of food to induce the appropriate global coordination to
provably transition the collective between macro-modes when required.  Informally, one can imagine individual particles adjusting their $\lambda$ parameter to be high when they are fed, encouraging compression around the food, and making $\lambda$ small when they are hungry, promoting the search for more food. 


\vskip.1in
\underbar{\bf Our Results and Techniques:} \ 
We present the first rigorous local distributed algorithms for solving the foraging problem:
{\it Adaptive $\alpha$-Compression,}  
based on the stochastic compression algorithm of~\cite{Cannon2016}, 
and  {\it 
 Adaptive Spiraling},
which takes a more structured, incremental approach to compression. 
In both algorithms, there are two main states a particle can be in at any point in time, {\it dispersion} or {\it compression}, which, at the macro-level, should induce the collective to enter the {\em search} or {\em gather} phases respectively.
Particles in the dispersion state move around in a process akin to a simple exclusion process, where they perform a random walk while avoiding two particles occupying the same site.
Particles enter a compression state when food is found and this information is propagated in the system, consequently resulting in the system gathering and compressing around the food.

 Adaptive $\alpha$-Compression  
 uses a variant of the compression algorithm of~\cite{Cannon2016} to function in the case where particles are to compress around a single fixed point (the food particle). 
While the compression algorithm largely stays the same, the proofs --- in particular, the one showing {\em ergodicity of the underlying Markov Chain} --- are non-trivial and do not follow directly from~\cite{Cannon2016}.
Using this modified compression algorithm, we allow the particles to stochastically reassemble into configurations with low perimeter around the food. We prove the following theorem (which follows from Theorems~\ref{theorem:compressionbasedwithfood} and \ref{theorem:compressionbasednofood}):
\begin{theorem}
Starting from any valid configuration, in the presence of a single food particle that remains static for a sufficient amount of time, our Adaptive $\alpha$-Compression algorithm 
will converge to an $\alpha$-compressed configuration, for any $\alpha>1$, connected to the food particle at stationarity with high probability.
Conversely, 
if there are no food particles in the system for a sufficient amount of time,
all particles will converge to the dispersion state. 
\end{theorem}




The second, more structured, algorithm for compressing around the food has particles incrementally attach to the end of a spiral
that is growing predictably and tightly wrapped around the food source.   By marking which particle was the last to join, we can ensure that a minimum perimeter cluster centered at the food source
grows incrementally as more particles attach.  
This idea is simple in principle, although the correct local conditions must be carefully chosen to ensure that the spiral is the only possible structure that forms, and that all other structures that may remain as a food particle is moved around adversarially will dissipate in a timely manner. While this careful spiral construction has the added advantage of leading to a self-stabilizing algorithm --- in the sense that it will converge to the desired configuration starting from an arbitrary (potentially corrupted) system configuration --- its highly structured, sequential nature slows progress and the method may break down under small perturbations in the environment, such as obstacles and situations where the food location and longevity are dynamic.
%

 We prove the following theorem (which follows from Theorems~\ref{theorem:spiralwithfood} and \ref{theorem:spiralnofood}):

\begin{theorem}
From any starting configuration, in the presence of a single food particle that remains static for a sufficient amount of time,
the Adaptive Spiraling Algorithm will reach a configuration where all particles form a single spiral tightly centered at the food particle.
Conversely, 
if no food  particles are in the system for a sufficient amount of time, all particles will converge to the dispersion state. 
\end{theorem}

The algorithms in this paper are, to the best of our knowledge, the first to leverage {\em self-induced phase changes as an algorithmic tool.}
We note that the compression algorithm of Cannon et al. \cite{Cannon2016} was not shown rigorously to converge in polynomial time, although they give strong experimental evidence that this is the case.  All other parts of the $\alpha$-Compression Algorithm provably run in expected polynomial time, as does  Adaptive Spiraling.

For both approaches, the challenge is to share information locally and autonomously so that eventually most particles enter the correct state and the system exhibits the appropriate phase behavior.
%
We rely on token passing for the system to be able to collectively transition between (multiple, possibly overlapping) gather and search phases.  In order to ensure that, our token passing scheme needs to be carefully engineered so that when the food particle moves or vanishes, the {\em rate at which the compressed cluster around the food dissipates (via particles returning to the dispersion state) outpaces the rate at which the cluster may continue to grow (via particles joining the cluster in a compression state)}, ensuring that the broadcast wave of the ``dispersion token'' will always outpace that of a compression wave. We address this formally and in detail in Sections~\ref{sec:adaptivecompression} and \ref{sec:spiraling} .



\section{Model and Preliminaries}
\label{sec:model}
We consider an abstraction of a self-organizing particle system (SOPS) on the triangular lattice. 
More specifically, particles occupy vetices in $\Lambda = \Lambda_N$, which is a $\sqrt{N} \times \sqrt{N}$  piece of the triangular lattice with periodic boundary conditions, where each vertex is occupied by at most one particle.
We assume that particles have constant-size memory, and have no global orientation or any other global information beyond a common chirality. Particles communicate by sending tokens to its nearest neighbors (i.e., with particles occupying adjacent vertices in the lattice), where a {\em token} is a constant-size piece of information.
 We are given a singular,  stationary food source (which we also refer to as a food particle) that may be placed at any point in space, removed, or shifted around at arbitrary times,  possibly adversarially.

Individual particles are activated according to their own Poisson clocks, possibly with different rates, and perform instantaneous actions upon activation. 
It will be convenient to refer a vertex in $\Lambda$ according to hypothetical global coordinates $(x,y)$, but note that this is just for ease of exposition, since the particles are not aware of any such global coordinate system. In this convention, a vertex $(x,y) \in \Lambda$ has edges to the vertices corresponding to $(x-1,y-1)$, $(x-1,y)$, $(x-1,y)$, $(x+1,y)$, $(x,y+1)$ ,$(x+1,y+1)$, with the arithmetic taken modulo $\sqrt{N}$.
Particles are aware of their own and their neighbors' current states and when a particle is activated, it may do a bounded amount of computation, send at most one token (not necessarily identical) to each of its neighbors, and choose one of its six neighbors in the lattice to see if it is unoccupied and move there.
Note that this model can be viewed as a high level abstraction of the (canonical) {\em Amoebot model}~\cite{Daymude2021-canonicalamoebot} under a {\em sequential scheduler}, where at most one particle would be active at any point in time. One should be able to port the model and algorithms presented in this paper to the Amoebot model; however a formal description on how this should be done is beyond the scope of this paper.


The compression algorithm of Cannon et al.~\cite{Cannon2016}, first defined on the infinite lattice, also works in the finite setting, which is how we present it.  We are given a Markov chain $(\mathcal{M}_{\rm COM}, \Omega_{\rm COM}, \pi_{\rm COM}),$ where $\Omega_{\rm COM}$ is the set of simply connected configurations with $n$ particles, $\mathcal{M}_{\rm COM}$ is the transition matrix giving the probabilities of moves, and $\pi_{\rm COM}$ is the stationary distribution over $\Omega_{\rm COM}$.  
A closely related algorithm $(\mathcal{M}_{\rm AGG}, \Omega_{\rm AGG}, \pi_{\rm AGG})$ relaxes the connectivity requirement 
and is defined on a finite graph so that the particles can find each other, with $\Omega_{\rm AGG}$ corresponding to all assignments of particles to distinct vertices~\cite{Li2021-bobbots}, $\mathcal{M}_{\rm AGG}$ the transition probabilities that no longer need to maintain connectivity, and $\pi_{\rm AGG}$ the corresponding stationary distribution.
 In both cases, for any allowable configurations $\sigma$ and $\tau$ differing by the move of a single particle $p$ along one lattice edge, the transition probability is 
 proportional to $\min(1, \lambda^{n'-n}),$ where $\lambda > 0$ is a bias parameter that is an input to the algorithm, $n$ is the number of neighbors of $p$ in~$\sigma$ and $n'$ is the number of neighbors of $p$ in $\tau$.
These probabilities  are defined so that the Markov chain $\mathcal{M}_{\rm COM}$ 
converges to the Boltzmann distribution $\pi_{\rm COM}$ such that $\pi_{\rm COM}(\sigma)$ is proportional to $\lambda^{E(\sigma)},$ where $\sigma \in \Omega_{\rm COM}$ and $E(\sigma)$ is the number of nearest neighbor pairs in configuration~$\sigma$.  Likewise, the Markov chain $\mathcal{M}_{\rm AGG}$ converges to $\pi_{\rm AGG}$ such that $\pi_{\rm AGG}(\sigma)$ is also proportional to $\lambda^{E(\sigma)},$ for any $\sigma \in \Omega_{\rm AGG}$.
Both algorithms are known to exhibit phase changes as the parameter $\lambda$ is varied, leading  to aggregation (compression) when $\lambda$ is large and dispersion (expansion) when $\lambda$ is small. 


Recall that a configuration is {\it $\alpha$-compressed} if  the perimeter (measured by the length of the closed walk around its boundary edges) is at most $\alpha\, p_{\text{min}}(n)$, for some constant $\alpha > 1,$  where $p_{\text{min}}(n)$ denotes the minimum possible perimeter of a connected system with $n$ particles.  When the configurations can be disconnected, we need a more subtle definition of compression.  Namely, a configuration is  {\it $(\beta,\delta)-aggregated$} if there exists a subset $R$ of vertices such that:
    \begin{enumerate}
        \item At most {$\beta \sqrt{N}$} edges have exactly one endpoint in $R$;
        \item The density of particles in $R$ is at least $1-\delta$; and
        \item The density of particles not in $R$ is at most $\delta$.
    \end{enumerate}
We will make use of of the following theorems, which were shown in the non-adaptive settings where there is only one goal and this remains unchanged.

\begin{theorem} \label{thm:compression} {\it Compression/Expansion}\ \cite{Cannon2016}: \  
Let configuration $\sigma$ be drawn from the stationary distribution of $\mathcal{M}_{\rm COM}$ on a bounded, compact region of the triangular lattice, when the number of particles $N$ is sufficiently large.
If $\lambda > 3.24$, then with high probability\footnote{ We use {\em with high probability} in this paper to denote ``with all but an exponentially small probability.''} there exists a constant $\alpha > 0$  such that $\sigma$ will be $\alpha$-compressed.
However, when $\lambda < 2.17$, the configuration $\sigma$ will be expanded with high probability.
\end{theorem}

\begin{theorem} \label{thm:aggregation} {\it Aggregation/Dispersion}\ \cite{Li2021-bobbots}: \  
Let configuration $\sigma$ be drawn from the stationary distribution of $\mathcal{M}_{\rm AGG}$ on a bounded, compact region of the triangular lattice, when the number of particles~$N$ is sufficiently large.
If $\lambda > 5.66$, then with high probability there exist constants $\beta > 0$ and $0 < \delta < 1/2$ such that $\sigma$ will be $(\beta, \delta)$-aggregated.
However, when $0.98 < \lambda < 1.02$, the configuration $\sigma$ will be dispersed with high probability.
\end{theorem}

\section{Adaptive $\alpha$-Compression}
\label{sec:adaptivecompression}
We now show how to build upon the compression and aggregation algorithms so that the particles collectively gather when there is food and disperse to search for food when there is none or they are not aware of any using only local movement and communication.  We would like for the particles to encode their current knowledge of the world with a small number of states so that the majority will collectively perform the appropriate gather or disperse behavior, iteratively self-regulating to adaptively transition between these phases as the food location changes or disappears.  

For convenience, the algorithm we describe here is a fairly straightforward hybrid of both  compression and aggregation algorithms where we carefully toggle between the high $\lambda$ (compression) regime for connected configurations (Theorem~\ref{thm:compression}) and the low $\lambda$ (dispersion) regime for disconnected configurations (Theorem~\ref{thm:aggregation}).  This combined approach allows us to compress into configurations that are simply connected, which helps the token passing protocol, and then to disconnect when searching for food, which allows the particles to provably explore the domain more efficiently.  Interestingly Theorems~\ref{thm:compression} and~\ref{thm:aggregation} still hold in this hybrid setting.  We note that the same results do hold in the more straightforward setting of compression/expansion where we always keep the configurations simply connected, but then the search for food is less efficient.

The Adaptive $\alpha$-Compression Algorithm works by integrating two main mechanisms: a state-based mechanism which each particle uses to determine  whether it is appropriate to disperse or compress and  a stochastic compression/dispersion algorithm, that implements the Markov chain allowing a simply-connected set of particles to compress or disperse.  

We start by defining the state-based mechanism which includes six states: {\em four compression states} $\mathscr{C}=\{C, C_G, C_F, C_{GF}\}$, a {\em dispersion state} $D$, and a {\em dispersion token state}
($DT$). The four compression states allow a particle to encode two additional bits of information: labels
$C_G$ or $C_{GF}$ indicate that the particle currently holds on to a {\em growth token} (which is used to grow the compressed cluster in a controlled way), while $C_F$ or $C_{GF}$ indicate that the particle believes it is currently next to a food particle.  Particles labeled $C$ have neither a growth token nor awareness of food.

Particles in the dispersion state execute simple random walks, disallowing moves into locations that are currently occupied.
The compression states coordinate to implement the gather phase; when activated, these particles execute the compression algorithm from \cite{Cannon2016} by picking a random direction, and if moving in the chosen movement satisfies certain properties (outlined in Section~\ref{section:statemovement}), the move is executed with a probability chosen to converge to a Gibbs measure proportional to $\lambda^E(\sigma)$, where configuration $\sigma$ has $E(\sigma)$ nearest neighbor pairs on the lattice. 
These actions are used to form a low-perimeter cluster around the food particle.
The dispersion token state is only used as an intermediate state to propagate the dispersion token, which we will discuss later.
We do not define any movement for these particles as they switch to the dispersion state on activation.

\vskip.1in
\underbar{\textbf{Growth Token:}}
Our aim is for all particles in the dispersion state to eventually join the cluster comprised by the set of compression state particles connected to the food particle. A particle in the dispersion state can switch to a compression state if it makes direct contact with the food particle. To grow the cluster of particles beyond the initial layer around the food, we introduce a mechanism of {\em growth token}. New growth tokens are constantly generated by particles neighboring the food particle, which are then passed randomly between particles of the cluster. 
If a particle in the dispersion state comes into contact with a particle in the cluster holding a growth token, the growth token is consumed, and the dispersion state particle joins the cluster (by switching to a compression state).
This system of growth tokens intentionally limits the rate of growth of the cluster of particles, which is key to ensuring that when the food particle is removed, the cluster disperses more quickly than it grows.



\vskip.1in
 \underbar{\textbf{Dispersion Token:}}
The {\em dispersion token} allows for a mechanism to return a cluster of particles to the dispersion state after a food particle has been removed. When a particle neighboring a food particle detects that the food particle has vanished, it generates a dispersion token, which spreads quickly throughout the cluster, returning all particles of the cluster to the dispersion state. This is done through the dispersion token state, which on activation switches the activated particle to the dispersion state and all neighbors to the dispersion token state.

\subsection{Particle Actions}
\label{section:specificimplementationdetails}
When a particle is activated, it can move and change its internal state. We divide this process into two steps, which we call the \emph{state change step} and the \emph{particle movement step}. These two steps are executed one after the other.
In each of these steps, we break down the particle's behavior by the state it is currently in.
There are two types of constant-size tokens that may be exchanged between two adjacent particles in the algorithm (which we call {\em growth} and {\em dispersion tokens}). For ease of explanation, and since we are under the assumption of a sequential scheduler, we present our token passing scheme by assuming that a particle $u$ who wants to send a token to a neighboring particle $v$ can do so by writing the token directly into the memory of $v$.

\vskip.1in
\noindent \underbar{\textbf{Step 1 - State Change}}
\label{section:statebehavior} 
We describe the behavior of each of the states on activation of a particle $u$. Let $p < \frac{1}{6}$ be a fixed positive constant that will determine the probability of a particle adjacent to the food particle to switch to a compression state, or of generating a new growth token when eligible, as we explain below. 

%
\begin{itemize}
\item \textbf{Dispersion State ($D$):}
If next to a food particle, switch $u$ to state $C_F$ with probability $p$.
If $u$ is not next to a food particle, and
there is a neighboring compression state particle with a growth token (in state $C_G$ or $C_{GF}$), then with probability $p$, consume that growth token (switching the neighboring particle to state $C$ or $C_F$) and switch $u$ to state $C$.

\item \textbf{Compression State ($C,C_G,C_{GF},C_F$):}
If $u$ has the food bit set (state $C_{GF}$ or $C_F$) but is not next to the food particle, does not have the food bit set but is next to the food particle, or is next to the food particle but shares a compression state neighbor with the food particle without the food bit set, switch $u$ to state $D$ and switch all compression state neighbors of $u$ to state $DT$.

Otherwise, if $u$ has the growth token bit set (state $C_G$ or $C_{GF}$), pick a random direction. If the neighbor $v$ in that direction is in a compression state without the growth bit set, flip the growth bit on both $u$ and $v$ (effectively passing the growth token from $u$ to $v$).

If $u$ does not have the growth token bit set but is next to the food particle, generate a growth token by flipping the corresponding bit with probability $p$.

\item \textbf{Dispersion Token State ($DT$):}
On activation, switch to state $D$ and switch all compression state neighbors of $u$ to the state $DT$.
\end{itemize}
%

\vskip.1in
\noindent \underbar{\textbf{Step 2 - Particle Movement}}
\label{section:statemovement}
This step is applied after the state change step.
\begin{itemize}
\item \textbf{Dispersion State ($D$):}
Executes a simple random walk, by picking a direction at random, and moving in that direction if and only if the immediate neighboring site in that direction is unoccupied.

\item \textbf{Compression State ($C,C_G,C_{GF},C_F$):}
Executes a compression movement. It first picks a direction at random to move in. If this move is a valid compression move (Definition~\ref{defn:validcompressionmoves}), we will make the move with the probability given in Definition~\ref{defn:compressionmoveprobability}.

After this, regardless of whether the move is made, we set the food bit to $1$ if the particle is now adjacent to food, and $0$ otherwise.

\item \textbf{Dispersion Token State ($DT$):}
A particle in the dispersion token state will not have the chance to move, as on activation it would have switched to the dispersion state.
\end{itemize}
We use the term \emph{cluster particles} to refer to the food particle or particles in a compression state or the dispersion token state. We restrict the movements of the compression state particles to keep the cluster particles connected.
Suppose a compression state particle in location $\ell$ wants to move to an empty adjacent location $\ell'$. Denote by $N(\ell)$ and $N(\ell')$ the sets of lattice sites neighboring the positions $\ell$ and $\ell'$ respectively. Also, $N(\ell \cup \ell')$ is defined to be $N(\ell)\cup N(\ell') \setminus \{\ell, \ell'\}$. Finally, let $\mathbb{S} := N(\ell) \cap N(\ell')$ denote the set of sites adjacent to both $\ell$ and $\ell'$ (thus $|\mathbb{S}| \in \{0,1,2\}$).

\begin{definition}[Valid Compression Moves]
\label{defn:validcompressionmoves}
Consider the following two properties:

\noindent
\underbar{Property 1:} $|\mathbb{S}| \geq 1$ and every cluster particle in $N(\ell \cup \ell')$ is connected to a cluster particle in $\mathbb{S}$ through $N(\ell \cup \ell')$.

\noindent
\underbar{Property 2:} $|\mathbb{S}| = 0$, $\ell$ and $\ell'$ each have at least one neighbor, all cluster particles in $N(\ell) \setminus \{\ell'\}$ are connected by paths within this set, and all cluster particles in $N(\ell') \setminus \{\ell\}$ are connected by paths within this set.

\noindent We say the move from $\ell$ to $\ell'$ is a \emph{valid compression move} if it satisfies both properties, and $N(\ell)$ contains less than $5$ cluster particles.
\end{definition}

\begin{definition}[Local Disconnection]
Consider a cluster particle currently in position $\ell$ trying to move to an adjacent location. We say the movement causes local disconnection if the set of cluster particles in $N(\ell) \cup \{\ell\}$ were connected through cluster particles in the same set, but no longer will be after the movement.
\end{definition}
We can see that valid compression moves do not cause local disconnection.
\begin{definition}[Transition probabilities]
\label{defn:compressionmoveprobability}
Fix $\lambda > 1$.
Even when a movement is valid, we only make the move with the probability $\min\{1,\lambda^{e(\sigma')-e(\sigma)}\}$,
where $\sigma$ and $\sigma'$ are the configurations before and after the movement is made, and $e(\cdot)$ represents the number of edges between cluster particles in a configuration.
It is important to note that even though $e(\cdot)$ is a global property, the difference $e(\sigma')-e(\sigma)$ can be computed locally by only counting the number of neighbors the particle has before and after its movement.
\end{definition}

These movement conditions and probabilities are based on the compression algorithm in \cite{Cannon2016}, giving the algorithm its name. With a far-from-trivial modification of the analysis in~\cite{Cannon2016} to account for the immobile food particle, we show in Appendix~\ref{apx:irreducibilityproof} that this allows the cluster of particles to form a configuration of low-perimeter.

\vskip.1in
\underbar{\textbf{The state invariant:}} \ 
For a fixed configuration, we consider a graph with the vertex set consisting of the set of all particles in a compression state or the dispersion token state. Vertices are adjacent in the graph if their corresponding particles are adjacent in the lattice. In this section, a \emph{component} refers to a connected component of this graph. Notably, food particles are not considered part of any connected component (See Figure~\ref{fig:key_invariant}).

This allows us to define the following state invariant, that we know holds from an initial configuration with every particle in the dispersion state. In Appendix~\ref{apx:proofofkeyinvariant} we show Lemma~\ref{lemma:keyinvariant}, which states that this key state invariant always holds when running the Adaptive $\alpha$-Compression Algorithm, despite the potential adversarial movement of the food particle.

\begin{definition}[State Invariant]
\label{defn:keyinvariant}
We say a component of a configuration satisfies the \emph{state invariant} if it contains at least one particle in the states $C_{GF}$, $C_F$ or $DT$. A configuration satisfies the state invariant if every component in the configuration satisfies the state invariant.
\end{definition}

\begin{figure}
\centering
\begin{center}
\begin{tikzpicture}[x=0.55cm,y=0.55cm]
\draw[lightgray] (5.19615,-9) -- (8.66025,-7);
\draw[lightgray] (0,-7) -- (8.66025,-2);
\draw[lightgray] (0,-7) -- (3.4641,-9);
\draw[lightgray] (3.4641,0) -- (8.66025,-3);
\draw[lightgray] (3.4641,0) -- (3.4641,-9);
\draw[lightgray] (7.79423,-0.5) -- (7.79423,-8.5);
\draw[lightgray] (0,-4) -- (6.9282,0);
\draw[lightgray] (0,-4) -- (7.79423,-8.5);
\draw[lightgray] (5.19615,0) -- (8.66025,-2);
\draw[lightgray] (5.19615,0) -- (5.19615,-9);
\draw[lightgray] (0,-1) -- (1.73205,0);
\draw[lightgray] (0,-1) -- (8.66025,-6);
\draw[lightgray] (0,-1) -- (0,-8);
\draw[lightgray] (6.9282,-9) -- (8.66025,-8);
\draw[lightgray] (0,-5) -- (7.79423,-0.5);
\draw[lightgray] (0,-5) -- (6.9282,-9);
\draw[lightgray] (0.866025,-0.5) -- (8.66025,-5);
\draw[lightgray] (0.866025,-0.5) -- (0.866025,-8.5);
\draw[lightgray] (0,-8) -- (8.66025,-3);
\draw[lightgray] (0,-8) -- (1.73205,-9);
\draw[lightgray] (0,-2) -- (3.4641,0);
\draw[lightgray] (0,-2) -- (8.66025,-7);
\draw[lightgray] (2.59808,-0.5) -- (2.59808,-8.5);
\draw[lightgray] (6.9282,0) -- (8.66025,-1);
\draw[lightgray] (6.9282,0) -- (6.9282,-9);
\draw[lightgray] (4.33013,-0.5) -- (4.33013,-8.5);
\draw[lightgray] (6.06218,-0.5) -- (6.06218,-8.5);
\draw[lightgray] (0,-6) -- (8.66025,-1);
\draw[lightgray] (0,-6) -- (5.19615,-9);
\draw[lightgray] (0,-3) -- (5.19615,0);
\draw[lightgray] (0,-3) -- (8.66025,-8);
\draw[lightgray] (3.4641,-9) -- (8.66025,-6);
\draw[lightgray] (1.73205,-9) -- (8.66025,-5);
\draw[lightgray] (8.66025,-1) -- (8.66025,-8);
\draw[lightgray] (1.73205,0) -- (8.66025,-4);
\draw[lightgray] (1.73205,0) -- (1.73205,-9);
\draw[lightgray] (0.866025,-8.5) -- (8.66025,-4);
\draw[black, line width=0.4mm, fill=white] (0.866025,-3.5) circle (0.288);
\draw[black, line width=0.4mm, fill=white] (0.866025,-5.5) circle (0.288);
\draw[black, line width=0.4mm, fill=white] (1.73205,-2) circle (0.288);
\draw[black, line width=0.4mm, fill=white] (1.73205,-3) circle (0.288);
\draw[black, line width=0.4mm, fill=white] (1.73205,-5) circle (0.288);
\draw[black, line width=0.4mm, fill=white] (1.73205,-6) circle (0.288);
\draw[black, line width=0.4mm, fill=white] (2.59808,-1.5) circle (0.288);
\draw[black, line width=0.4mm, fill=white] (2.59808,-2.5) circle (0.288);
\draw[black, line width=0.4mm, fill=white] (2.59808,-3.5) circle (0.288);
\draw[black, line width=0.4mm, fill=white] (2.59808,-5.5) circle (0.288);
\draw[black, line width=0.4mm, fill=white] (2.59808,-6.5) circle (0.288);
\draw[black, line width=0.4mm, fill=white] (2.59808,-7.5) circle (0.288);
\draw[black, line width=0.4mm, fill=white] (3.4641,-1) circle (0.288);
\node[align=left] at (3.4641,-1) {\fontsize{4.5}{4}\selectfont $DT$};
\draw[black, line width=0.4mm, fill=white] (3.4641,-2) circle (0.288);
\node[align=left] at (3.4641,-2) {\fontsize{4.5}{4}\selectfont $DT$};
\draw[black, line width=0.4mm, fill=white] (3.4641,-3) circle (0.288);
\draw[black, line width=0.4mm, fill=white] (3.4641,-6) circle (0.288);
\node[align=left] at (3.4641,-6) {\fontsize{4.5}{4}\selectfont $C_F$};
\draw[black, line width=0.4mm, fill=white] (3.4641,-7) circle (0.288);
\draw[black, line width=0.4mm, fill=white] (3.4641,-8) circle (0.288);
\draw[black, line width=0.4mm, fill=white] (4.33013,-4.5) circle (0.288);
\draw[black, line width=0.4mm, fill=white] (4.33013,-5.5) circle (0.336);
\draw[black, line width=0.32mm] (4.33013,-5.5) circle (0.24);
\node[align=left] at (4.33013,-5.5) {\tiny $f$};
\draw[black, line width=0.4mm, fill=white] (4.33013,-6.5) circle (0.288);
\node[align=left] at (4.33013,-6.5) {\fontsize{4.5}{4}\selectfont $C_F$};
\draw[black, line width=0.4mm, fill=white] (5.19615,-3) circle (0.288);
\draw[black, line width=0.4mm, fill=white] (5.19615,-4) circle (0.288);
\draw[black, line width=0.4mm, fill=white] (5.19615,-5) circle (0.288);
\node[align=left] at (5.19615,-5) {\fontsize{4.5}{4}\selectfont $C_F$};
\draw[black, line width=0.4mm, fill=white] (5.19615,-8) circle (0.288);
\draw[black, line width=0.4mm, fill=white] (6.06218,-3.5) circle (0.288);
\draw[black, line width=0.4mm, fill=white] (6.06218,-4.5) circle (0.288);
\draw[black, line width=0.4mm, fill=white] (6.06218,-6.5) circle (0.288);
\draw[black, line width=0.4mm, fill=white] (6.06218,-7.5) circle (0.288);
\node[align=left] at (6.06218,-7.5) {\fontsize{4.5}{4}\selectfont $C_F$};
\draw[black, line width=0.4mm, fill=white] (6.9282,-3) circle (0.288);
\draw[black, line width=0.4mm, fill=white] (6.9282,-6) circle (0.288);
\draw[black, line width=0.4mm, fill=white] (6.9282,-7) circle (0.288);
\draw[black, line width=0.4mm, fill=white] (7.79423,-5.5) circle (0.288);
\draw[black, line width=0.4mm, fill=white] (7.79423,-6.5) circle (0.288);
\draw[black, line width=0.5mm] (6.32199,-3.35) -- (6.6684,-3.15);
\draw[black, line width=0.5mm] (3.4641,-6.7) -- (3.4641,-6.3);
\draw[black, line width=0.5mm] (3.72391,-6.85) -- (4.07032,-6.65);
\draw[black, line width=0.5mm] (1.12583,-3.35) -- (1.47224,-3.15);
\draw[black, line width=0.5mm] (3.4641,-7.7) -- (3.4641,-7.3);
\draw[black, line width=0.5mm] (4.33013,-6.2) -- (4.33013,-5.8);
\draw[black, line width=0.5mm] (6.06218,-7.2) -- (6.06218,-6.8);
\draw[black, line width=0.5mm] (6.32199,-7.35) -- (6.6684,-7.15);
\draw[black, line width=0.5mm] (2.59808,-7.2) -- (2.59808,-6.8);
\draw[black, line width=0.5mm] (2.85788,-7.35) -- (3.20429,-7.15);
\draw[black, line width=0.5mm] (2.85788,-7.65) -- (3.20429,-7.85);
\draw[black, line width=0.5mm] (1.99186,-5.15) -- (2.33827,-5.35);
\draw[black, line width=0.5mm] (2.59808,-3.2) -- (2.59808,-2.8);
\draw[black, line width=0.5mm] (2.85788,-3.35) -- (3.20429,-3.15);
\draw[black, line width=0.5mm] (4.33013,-5.2) -- (4.33013,-4.8);
\draw[black, line width=0.5mm] (4.58993,-5.35) -- (4.93634,-5.15);
\draw[black, line width=0.5mm] (6.32199,-6.35) -- (6.6684,-6.15);
\draw[black, line width=0.5mm] (6.32199,-6.65) -- (6.6684,-6.85);
\draw[black, line width=0.5mm] (5.45596,-3.15) -- (5.80237,-3.35);
\draw[black, line width=0.5mm] (1.12583,-5.35) -- (1.47224,-5.15);
\draw[black, line width=0.5mm] (1.12583,-5.65) -- (1.47224,-5.85);
\draw[black, line width=0.5mm] (2.59808,-6.2) -- (2.59808,-5.8);
\draw[black, line width=0.5mm] (2.85788,-6.35) -- (3.20429,-6.15);
\draw[black, line width=0.5mm] (2.85788,-6.65) -- (3.20429,-6.85);
\draw[black, line width=0.5mm] (1.99186,-1.85) -- (2.33827,-1.65);
\draw[black, line width=0.5mm] (1.99186,-2.15) -- (2.33827,-2.35);
\draw[black, line width=0.5mm] (7.18801,-5.85) -- (7.53442,-5.65);
\draw[black, line width=0.5mm] (7.18801,-6.15) -- (7.53442,-6.35);
\draw[black, line width=0.5mm] (5.19615,-3.7) -- (5.19615,-3.3);
\draw[black, line width=0.5mm] (5.45596,-3.85) -- (5.80237,-3.65);
\draw[black, line width=0.5mm] (5.45596,-4.15) -- (5.80237,-4.35);
\draw[black, line width=0.5mm] (2.59808,-2.2) -- (2.59808,-1.8);
\draw[black, line width=0.5mm] (2.85788,-2.35) -- (3.20429,-2.15);
\draw[black, line width=0.5mm] (2.85788,-2.65) -- (3.20429,-2.85);
\draw[black, line width=0.5mm] (1.73205,-5.7) -- (1.73205,-5.3);
\draw[black, line width=0.5mm] (1.99186,-5.85) -- (2.33827,-5.65);
\draw[black, line width=0.5mm] (1.99186,-6.15) -- (2.33827,-6.35);
\draw[black, line width=0.5mm] (4.58993,-4.35) -- (4.93634,-4.15);
\draw[black, line width=0.5mm] (4.58993,-4.65) -- (4.93634,-4.85);
\draw[black, line width=0.5mm] (1.73205,-2.7) -- (1.73205,-2.3);
\draw[black, line width=0.5mm] (1.99186,-2.85) -- (2.33827,-2.65);
\draw[black, line width=0.5mm] (1.99186,-3.15) -- (2.33827,-3.35);
\draw[black, line width=0.5mm] (6.9282,-6.7) -- (6.9282,-6.3);
\draw[black, line width=0.5mm] (7.18801,-6.85) -- (7.53442,-6.65);
\draw[black, line width=0.5mm] (3.4641,-1.7) -- (3.4641,-1.3);
\draw[black, line width=0.5mm] (7.79423,-6.2) -- (7.79423,-5.8);
\draw[black, line width=0.5mm] (5.19615,-4.7) -- (5.19615,-4.3);
\draw[black, line width=0.5mm] (5.45596,-4.85) -- (5.80237,-4.65);
\draw[black, line width=0.5mm] (2.85788,-5.65) -- (3.20429,-5.85);
\draw[black, line width=0.5mm] (3.72391,-5.85) -- (4.07032,-5.65);
\draw[black, line width=0.5mm] (3.72391,-6.15) -- (4.07032,-6.35);
\draw[black, line width=0.5mm] (5.45596,-7.85) -- (5.80237,-7.65);
\draw[black, line width=0.5mm] (2.85788,-1.35) -- (3.20429,-1.15);
\draw[black, line width=0.5mm] (2.85788,-1.65) -- (3.20429,-1.85);
\draw[black, line width=0.5mm] (6.06218,-4.2) -- (6.06218,-3.8);
\draw[black, line width=0.5mm] (3.4641,-2.7) -- (3.4641,-2.3);
\end{tikzpicture}
\end{center}
\caption{A configuration with four components that satisfies the state invariant. Particles in the dispersion state are not drawn. The label $f$ refers to the food particle, $DT$ refers to dispersion state particles, and $C_F$ refers to particles with the food bit set.}
\label{fig:key_invariant}
\end{figure}
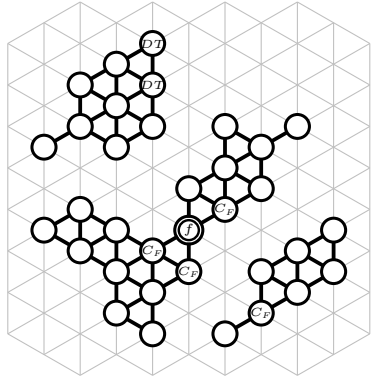

\begin{lemma}
\label{lemma:keyinvariant}
When the Adaptive $\alpha$-Compression Algorithm is run starting from a configuration where every particle is in the dispersion state, the state invariant will hold on every subsequent configuration.
\end{lemma}

\subsection{Gathering and Dispersing}
To verify the behavior of the algorithm in response to environmental changes, we show the following two theorems.
\begin{theorem}
\label{theorem:compressionbasedwithfood}
Starting from any configuration satisfying the state invariant, given that the food particle exists and remains in a fixed position for a sufficient amount of time, we will reach a configuration where all particles are in a compression state in a single cluster around the food particle in a polynomial number of steps in expectation. 

Furthermore, for any $\alpha > 1$, there exists a $\lambda_\alpha > 1$ such that for any $\lambda > \lambda_\alpha$ as used in Definition~\ref{defn:compressionmoveprobability}, after a sufficiently long time has passed, the perimeter of the cluster will be at most $\alpha p_{min}$ with high probability, where $p_{min}$ is the minimum possible perimeter of a cluster of the same size.
\end{theorem}

\begin{theorem}
\label{theorem:compressionbasednofood}
Starting from any configuration satisfying the state invariant, if there is no food particle in the system for a sufficient amount of time, all particles will return to the dispersion state within a polynomial number of steps in expectation.
\end{theorem}

A key component of the proof is the elimination of {\em residual particles}, which are particles that remain from previous attempts to gather and that we wish to return to the dispersion state.

\begin{definition}[Residual Components and Particles]
A \emph{residual component} is a component that satisfies at least one of the following three criteria:
\begin{enumerate}
\item Or it contains a particle in the dispersion token state;
\item It contains a compression state particle with the food bit set that is not adjacent to food;
\item It contains a compression state particle without the food bit set that is adjacent to food.
\end{enumerate}
We call particles belonging to residual components \emph{residual particles}
\end{definition}

Out of the four components displayed in Figure~\ref{fig:key_invariant}, only the component in the bottom left is not a residual component.
When starting from an arbitrary configuration, there are likely to be a large amount of residual particles that need to cleared out. Residual particles are the remnants of partially built clusters as the food particle moves from place to place prior to the start of our analysis.
Residual particles are problematic as they obstruct existing particles, do not contribute to the main cluster and can cause even more residual particles to form.
Our main strategy to show that all residual particles will eventually vanish is to define a {\em potential function that decreases more quickly than it increases}.

\begin{definition}[Potential]
For a configuration $\sigma$, we define its potential $\Phi(\sigma)$ as:
$$\Phi(\sigma) := \Phi_C(\sigma) + \Phi_{DT}(\sigma) + \Phi_T(\sigma)$$
Where $\Phi_C(\sigma)$, $\Phi_{DT}(\sigma)$ and $\Phi_T(\sigma)$ represent the number of compression state particles, the number of dispersion token particles, and the number of compression state particles with growth tokens respectively.
\end{definition}
We use the following lemma to guarantee that the expected number of steps before all residual particles are removed is polynomial. The proof of this lemma is given in Appendix~\ref{apx:proofofgenerallemma}. This same lemma will be used again later in the proof of the Adaptive Spiraling algorithm.
%
%
\begin{lemma}
\label{lemma:generallemmafinal}
Fix integers $n, k \geq 1$, and $0 < \eta < 1$. Consider two random sequences of probabilities $(p_t)_{t \in \mathbbNzero}$ and $(q_t)_{t \in \mathbbNzero}$, with the properties that $\frac{1}{n} \leq p_t \leq 1$ and $0 \leq q_t \leq \eta p_t$, and $p_t + q_t \leq 1$. Now consider a sequence $(X_t)_{t \in \mathbbNzero}$, where $X_0 = k$, and 
\begin{align*}
X_{t+1} \begin{cases}
\leq X_t-1 &\text{with probability $p_t$}\\
= X_t+1 &\text{with probability $q_t$}\\
= X_t &\text{with probability $1-p_t-q_t$}.
\end{cases}.
\end{align*}
Then $\ex\left[\min\{t \geq 0 \mid X_t = 0\}\right] \leq \frac{nk}{1-\eta}$
\end{lemma}
%
%
In this case, $(X_t)_{t \in \mathbb{N}_{\geq 0}}$ 
represents the sequence of potentials, starting from some configuration satisfying the state invariant.
In Appendix~\ref{apx:compressioneliminatingresidual}, we make use of this to show that assuming no change in the food, we reach a configuration with no residual components within a polynomial number of steps in expectation.
In addition, we show that once all residual components are eliminated, no new ones will be generated, and all compression state particles must form a single cluster (defined formally as a connected component when we include both compression state particles and the food particle as vertices) around the food particle.

In the case where no food particle exists, this necessarily implies that all particles will be in the dispersion state, giving us Theorem~\ref{theorem:compressionbasednofood}.
%
In the case where a food particle exists and remains in place, we need to show next that all dispersion state particles eventually switch to the compression state. Lemma~\ref{lemma:allswitchtocompressionstate} implies this as without residual particles, no compression state particle can switch back to the dispersion state.

\begin{lemma}
\label{lemma:allswitchtocompressionstate}
Suppose that the food particle exists and remains in a fixed location, the state invariant holds, and that there are no residual particles in the configuration. Then the number of steps before the next particle switches from the dispersion to some compression state is polynomial in expectation.
\end{lemma}
The proof of this uses bounds on hitting times of random walks on evolving connected graphs~\cite{RandomWalksDynamicGraphs}, and is given in Appendix~\ref{apx:compressionstatelemmaproof}. This gives us the first part of Theorem~\ref{theorem:compressionbasedwithfood}.

We note that the ability to appropriately gather and disperse does not rely on the specific movements behavior defined for the compression state particles. The token passing mechanism functions as long as they only make moves that maintain local connectivity. In this case however, as we desire a low perimeter configuration of the cluster for the second part of Theorem~\ref{theorem:compressionbasedwithfood}, we apply the movement behaviors from the compression algorithm of~\cite{Cannon2016}.

\paragraph*{Achieving $\alpha$-Compression.}
Assuming irreducibility of the Markov chain, the results of \cite{Cannon2016} guarantee that for any $\alpha > 1$, there exists a sufficiently large constant $\lambda$ such that at stationarity, the perimeter of the cluster is at most $\alpha$ times as large as its minimum possible perimeter with high probability.
However, while the Metropolis-Hastings algorithm allows us to obtain the same stationary distribution and hence the same low-perimeter clusters, irreducibility of the Markov chain is not a given, despite using the same set of moves.
This is as the food particle, which we do not have control over, is part of the cluster. This is described in the following theorem:





\begin{lemma}
\label{lemma:irreducible}
We consider configurations of particles in compression states with a single food particle.
Suppose that there is enough space on the lattice for all particles to be laid out on a single line on any position and in any direction, without the line intersecting itself.
If the compression state particles can only make the moves given in Definition~\ref{defn:validcompressionmoves} while the food particle remains fixed in place, there exists a sequence of valid moves to transform any connected configuration of particles into any hole-free configuration of particles.
\end{lemma}
We call a configuration hole-free if there exists no closed path over the particles that encircles an empty site of the lattice.
The proof, which is given in Appendix~\ref{apx:irreducibilityproof}, builds upon the analysis in \cite{Cannon2016}, which describes a sequence of valid compression moves to transform any connected configuration of particles into a single long line. Unfortunately, having a single immobile particle makes this proof significantly more complex.

The main strategy in the proof is to treat the immobile food particle as the ``center'' of the configuration, and consider the lines extending from the center in each of the six possible directions. These lines, which we call ``spines'', divide the lattice into six regions.
The sequence of moves described in \cite{Cannon2016} is then modified to operate within one of these regions, with limited side effects on the two regions counterclockwise from this region. We call this sequence of moves a ``comb'', and show that there is a sequence of comb operations that can be applied to the configuration, repeatedly going round the six regions in counterclockwise order, until the resulting configuration is a single long line.
We then observe that any valid compression move transforming a hole-free configuration to a hole-free configuration is also valid in the reverse direction, giving us the statement of the lemma.

\section{Adaptive Spiraling}
\label{sec:spiraling}
In Adaptive Spiraling, particles attempt to construct a minimum perimeter configuration by incrementally building a spiral tightly wrapped around the food particle. Each particle checks for a specific local structure, and we show that if this local structure is satisfied by every particle, the only possible result is the spiral.
If the food particle moves or vanishes, particles dissipate in a similar manner - those closest to the food particle are the first to notice that the food particle is no longer present, and thus switch to the dispersion state. Particles in outer layers of the spiral eventually also switch to the dispersion state when they see their neighbors on the inner side of the spiral switch to the dispersion state.


\begin{figure}
\begin{minipage}[t]{.55\columnwidth}
  \begin{center}
  \begin{tikzpicture}[x=0.6cm,y=0.6cm]
  \draw[lightgray] (3.4641,-7) -- (6.06218,-5.5);
\draw[lightgray] (2.59808,-0.5) -- (6.06218,-2.5);
\draw[lightgray] (2.59808,-0.5) -- (2.59808,-6.5);
\draw[lightgray] (0,-2) -- (2.59808,-0.5);
\draw[lightgray] (0,-2) -- (6.06218,-5.5);
\draw[lightgray] (0,-7) -- (6.06218,-3.5);
\draw[lightgray] (1.73205,-1) -- (1.73205,-7);
\draw[lightgray] (0,-3) -- (4.33013,-0.5);
\draw[lightgray] (0,-3) -- (6.06218,-6.5);
\draw[lightgray] (5.19615,-7) -- (6.06218,-6.5);
\draw[lightgray] (0,-4) -- (6.06218,-0.5);
\draw[lightgray] (0,-4) -- (5.19615,-7);
\draw[lightgray] (4.33013,-0.5) -- (6.06218,-1.5);
\draw[lightgray] (4.33013,-0.5) -- (4.33013,-6.5);
\draw[lightgray] (3.4641,-1) -- (3.4641,-7);
\draw[lightgray] (0,-5) -- (6.06218,-1.5);
\draw[lightgray] (0,-5) -- (3.4641,-7);
\draw[lightgray] (0.866025,-0.5) -- (6.06218,-3.5);
\draw[lightgray] (0.866025,-0.5) -- (0.866025,-6.5);
\draw[lightgray] (0,-1) -- (0.866025,-0.5);
\draw[lightgray] (0,-1) -- (6.06218,-4.5);
\draw[lightgray] (0,-1) -- (0,-7);
\draw[lightgray] (1.73205,-7) -- (6.06218,-4.5);
\draw[lightgray] (6.06218,-0.5) -- (6.06218,-6.5);
\draw[lightgray] (5.19615,-1) -- (5.19615,-7);
\draw[lightgray] (0,-6) -- (6.06218,-2.5);
\draw[lightgray] (0,-6) -- (1.73205,-7);
\draw[black, line width=0.4mm, fill=white] (0.866025,-2.5) circle (0.288);
\node[align=left] at (0.866025,-2.5) {\scriptsize $6$};
\draw[black,-{Stealth[length=1.6mm,width=2.2mm]},line width=0.5mm] (1.13449,-2.345) -- (1.46358,-2.155);
\draw[black, line width=0.4mm, fill=white] (0.866025,-3.5) circle (0.288);
\node[align=left] at (0.866025,-3.5) {\scriptsize $6$};
\draw[black,-{Stealth[length=1.6mm,width=2.2mm]},line width=0.5mm] (0.866025,-3.19) -- (0.866025,-2.81);
\draw[black, line width=0.4mm, fill=white] (0.866025,-4.5) circle (0.288);
\node[align=left] at (0.866025,-4.5) {\scriptsize $6$};
\draw[black,-{Stealth[length=1.6mm,width=2.2mm]},line width=0.5mm] (0.866025,-4.19) -- (0.866025,-3.81);
\draw[black, line width=0.4mm, fill=white] (0.866025,-5.5) circle (0.288);
\node[align=left] at (0.866025,-5.5) {\scriptsize $6$};
\draw[black,-{Stealth[length=1.6mm,width=2.2mm]},line width=0.5mm] (0.866025,-5.19) -- (0.866025,-4.81);
\draw[black, line width=0.4mm, fill=white] (1.73205,-2) circle (0.288);
\node[align=left] at (1.73205,-2) {\scriptsize $6$};
\draw[black,-{Stealth[length=1.6mm,width=2.2mm]},line width=0.5mm] (2.00052,-1.845) -- (2.32961,-1.655);
\draw[black, line width=0.4mm, fill=white] (1.73205,-3) circle (0.288);
\node[align=left] at (1.73205,-3) {\scriptsize $6$};
\draw[black,-{Stealth[length=1.6mm,width=2.2mm]},line width=0.5mm] (2.00052,-2.845) -- (2.32961,-2.655);
\draw[black, line width=0.4mm, fill=white] (1.73205,-4) circle (0.288);
\node[align=left] at (1.73205,-4) {\scriptsize $6$};
\draw[black,-{Stealth[length=1.6mm,width=2.2mm]},line width=0.5mm] (1.73205,-3.69) -- (1.73205,-3.31);
\draw[black, line width=0.4mm, fill=white] (1.73205,-5) circle (0.288);
\node[align=left] at (1.73205,-5) {\scriptsize $6$};
\draw[black,-{Stealth[length=1.6mm,width=2.2mm]},line width=0.5mm] (1.73205,-4.69) -- (1.73205,-4.31);
\draw[black, line width=0.4mm, fill=white] (1.73205,-6) circle (0.288);
\node[align=left] at (1.73205,-6) {\scriptsize $6$};
\draw[black,-{Stealth[length=1.6mm,width=2.2mm]},line width=0.5mm] (1.46358,-5.845) -- (1.13449,-5.655);
\draw[black, line width=0.4mm, fill=white] (2.59808,-1.5) circle (0.288);
\node[align=left] at (2.59808,-1.5) {\scriptsize $6$};
\draw[black,-{Stealth[length=1.6mm,width=2.2mm]},line width=0.5mm] (2.86654,-1.655) -- (3.19563,-1.845);
\draw[black, line width=0.4mm, fill=white] (2.59808,-2.5) circle (0.288);
\node[align=left] at (2.59808,-2.5) {\scriptsize $6$};
\draw[black,-{Stealth[length=1.6mm,width=2.2mm]},line width=0.5mm] (2.86654,-2.655) -- (3.19563,-2.845);
\draw[black, line width=0.4mm, fill=white] (2.59808,-3.5) circle (0.288);
\node[align=left] at (2.59808,-3.5) {\fontsize{6.5}{4}\selectfont ~$0^*$};
\draw[black,-{Stealth[length=1.6mm,width=2.2mm]},line width=0.5mm] (2.86654,-3.655) -- (3.19563,-3.845);
\draw[black, line width=0.4mm, fill=white] (2.59808,-4.5) circle (0.288);
\node[align=left] at (2.59808,-4.5) {\fontsize{6.5}{4}\selectfont ~$1^*$};
\draw[black,-{Stealth[length=1.6mm,width=2.2mm]},line width=0.5mm] (2.59808,-4.19) -- (2.59808,-3.81);
\draw[black, line width=0.4mm, fill=white] (2.59808,-5.5) circle (0.288);
\node[align=left] at (2.59808,-5.5) {\scriptsize $6$};
\draw[black,-{Stealth[length=1.6mm,width=2.2mm]},line width=0.5mm] (2.32961,-5.345) -- (2.00052,-5.155);
\draw[black, line width=0.4mm, fill=white] (3.4641,-2) circle (0.288);
\node[align=left] at (3.4641,-2) {\scriptsize $6$};
\draw[black,-{Stealth[length=1.6mm,width=2.2mm]},line width=0.5mm] (3.73257,-2.155) -- (4.06166,-2.345);
\draw[black, line width=0.4mm, fill=white] (3.4641,-3) circle (0.288);
\node[align=left] at (3.4641,-3) {\fontsize{6.5}{4}\selectfont ~$5^*$};
\draw[black,-{Stealth[length=1.6mm,width=2.2mm]},line width=0.5mm] (3.73257,-3.155) -- (4.06166,-3.345);
\draw[black, line width=0.4mm, fill=white] (3.4641,-4) circle (0.336);
\draw[black, line width=0.32mm] (3.4641,-4) circle (0.24);
\node[align=left] at (3.4641,-4) {\tiny $f$};
\draw[black, line width=0.4mm, fill=white] (3.4641,-5) circle (0.288);
\node[align=left] at (3.4641,-5) {\fontsize{6.5}{4}\selectfont ~$2^*$};
\draw[black,-{Stealth[length=1.6mm,width=2.2mm]},line width=0.5mm] (3.19563,-4.845) -- (2.86654,-4.655);
\draw[black, line width=0.4mm, fill=white] (3.4641,-6) circle (0.288);
\node[align=left] at (3.4641,-6) {\scriptsize $6$};
\draw[black,-{Stealth[length=1.6mm,width=2.2mm]},line width=0.5mm] (3.19563,-5.845) -- (2.86654,-5.655);
\draw[black, line width=0.4mm, fill=white] (4.33013,-2.5) circle (0.288);
\node[align=left] at (4.33013,-2.5) {\scriptsize $6$};
\draw[black,-{Stealth[length=1.6mm,width=2.2mm]},line width=0.5mm] (4.59859,-2.655) -- (4.92768,-2.845);
\draw[black, line width=0.4mm, fill=white] (4.33013,-3.5) circle (0.288);
\node[align=left] at (4.33013,-3.5) {\fontsize{6.5}{4}\selectfont ~$4^*$};
\draw[black,-{Stealth[length=1.6mm,width=2.2mm]},line width=0.5mm] (4.33013,-3.81) -- (4.33013,-4.19);
\draw[black, line width=0.4mm, fill=white] (4.33013,-4.5) circle (0.288);
\node[align=left] at (4.33013,-4.5) {\fontsize{6.5}{4}\selectfont ~$3^*$};
\draw[black,-{Stealth[length=1.6mm,width=2.2mm]},line width=0.5mm] (4.06166,-4.655) -- (3.73257,-4.845);
\draw[black, line width=0.4mm, fill=white] (4.33013,-5.5) circle (0.288);
\node[align=left] at (4.33013,-5.5) {\scriptsize $6$};
\draw[black,-{Stealth[length=1.6mm,width=2.2mm]},line width=0.5mm] (4.06166,-5.655) -- (3.73257,-5.845);
\draw[black, line width=0.4mm, fill=white] (5.19615,-3) circle (0.288);
\node[align=left] at (5.19615,-3) {\scriptsize $6$};
\draw[black,-{Stealth[length=1.6mm,width=2.2mm]},line width=0.5mm] (5.19615,-3.31) -- (5.19615,-3.69);
\draw[black, line width=0.4mm, fill=white] (5.19615,-4) circle (0.288);
\node[align=left] at (5.19615,-4) {\scriptsize $6$};
\draw[black,-{Stealth[length=1.6mm,width=2.2mm]},line width=0.5mm] (5.19615,-4.31) -- (5.19615,-4.69);
\draw[black, line width=0.4mm, fill=white] (5.19615,-5) circle (0.288);
\node[align=left] at (5.19615,-5) {\scriptsize $6$};
\draw[black,-{Stealth[length=1.6mm,width=2.2mm]},line width=0.5mm] (4.92768,-5.155) -- (4.59859,-5.345);
  \end{tikzpicture}
  \end{center}
  \captionof{figure}{A spiral surrounding the food particle (labeled $f$). Particles are labeled with their states.}
  \label{fig:spiral}
\end{minipage}%
\hfill
\begin{minipage}[t]{.4\columnwidth}
  \begin{center}
  \begin{tikzpicture}[x=0.3cm,y=0.3cm]
  \draw[lightgray] (0,-7) -- (6.9282,-3);
\draw[lightgray] (0,-7) -- (4.33013,-9.5);
\draw[lightgray] (0,-10) -- (6.9282,-6);
\draw[lightgray] (0,-10) -- (1.73205,-11);
\draw[lightgray] (0,-4) -- (6.9282,0);
\draw[lightgray] (0,-4) -- (6.9282,-8);
\draw[lightgray] (0,-4) -- (0,-12);
\draw[lightgray] (2.59808,-2.5) -- (6.9282,-5);
\draw[lightgray] (2.59808,-2.5) -- (2.59808,-10.5);
\draw[lightgray] (4.33013,-1.5) -- (6.9282,-3);
\draw[lightgray] (4.33013,-1.5) -- (4.33013,-9.5);
\draw[lightgray] (0,-11) -- (6.9282,-7);
\draw[lightgray] (0,-11) -- (0.866025,-11.5);
\draw[lightgray] (0,-5) -- (6.9282,-1);
\draw[lightgray] (0,-5) -- (6.06218,-8.5);
\draw[lightgray] (0,-8) -- (6.9282,-4);
\draw[lightgray] (0,-8) -- (3.4641,-10);
\draw[lightgray] (5.19615,-1) -- (6.9282,-2);
\draw[lightgray] (5.19615,-1) -- (5.19615,-9);
\draw[lightgray] (0.866025,-3.5) -- (6.9282,-7);
\draw[lightgray] (0.866025,-3.5) -- (0.866025,-11.5);
\draw[lightgray] (6.9282,0) -- (6.9282,-8);
\draw[lightgray] (0,-12) -- (6.9282,-8);
\draw[lightgray] (1.73205,-3) -- (6.9282,-6);
\draw[lightgray] (1.73205,-3) -- (1.73205,-11);
\draw[lightgray] (3.4641,-2) -- (6.9282,-4);
\draw[lightgray] (3.4641,-2) -- (3.4641,-10);
\draw[lightgray] (6.06218,-0.5) -- (6.9282,-1);
\draw[lightgray] (6.06218,-0.5) -- (6.06218,-8.5);
\draw[lightgray] (0,-6) -- (6.9282,-2);
\draw[lightgray] (0,-6) -- (5.19615,-9);
\draw[lightgray] (0,-9) -- (6.9282,-5);
\draw[lightgray] (0,-9) -- (2.59808,-10.5);
\draw[black, line width=0.4mm, fill=white] (0,-4) circle (0.288);
\draw[black, line width=0.4mm, fill=white] (0,-5) circle (0.288);
\draw[black, line width=0.4mm, fill=white] (0,-6) circle (0.288);
\draw[black, line width=0.4mm, fill=white] (0,-7) circle (0.288);
\draw[black, line width=0.4mm, fill=white] (0,-8) circle (0.288);
\draw[black, line width=0.4mm, fill=white] (0,-11) circle (0.288);
\draw[black, line width=0.4mm, fill=white] (0.866025,-3.5) circle (0.288);
\draw[black, line width=0.4mm, fill=white] (0.866025,-4.5) circle (0.288);
\draw[black, line width=0.4mm, fill=white] (0.866025,-5.5) circle (0.288);
\draw[black, line width=0.4mm, fill=white] (0.866025,-6.5) circle (0.288);
\draw[black, line width=0.4mm, fill=white] (0.866025,-7.5) circle (0.288);
\draw[black, line width=0.4mm, fill=white] (0.866025,-8.5) circle (0.288);
\draw[black, line width=0.4mm, fill=white] (0.866025,-10.5) circle (0.288);
\draw[black, line width=0.4mm, fill=white] (1.73205,-3) circle (0.288);
\draw[black, line width=0.4mm, fill=white] (1.73205,-4) circle (0.288);
\draw[black, line width=0.4mm, fill=white] (1.73205,-5) circle (0.288);
\draw[black, line width=0.4mm, fill=white] (1.73205,-6) circle (0.288);
\draw[black, line width=0.4mm, fill=white] (1.73205,-7) circle (0.288);
\draw[black, line width=0.4mm, fill=white] (1.73205,-8) circle (0.288);
\draw[black, line width=0.4mm, fill=white] (1.73205,-9) circle (0.288);
\draw[black, line width=0.4mm, fill=white] (2.59808,-2.5) circle (0.288);
\draw[black, line width=0.4mm, fill=white] (2.59808,-3.5) circle (0.288);
\draw[black, line width=0.4mm, fill=white] (2.59808,-4.5) circle (0.288);
\draw[black, line width=0.4mm, fill=white] (2.59808,-5.5) circle (0.288);
\draw[black, line width=0.4mm, fill=white] (2.59808,-6.5) circle (0.288);
\draw[black, line width=0.4mm, fill=white] (2.59808,-7.5) circle (0.288);
\draw[black, line width=0.4mm, fill=white] (2.59808,-8.5) circle (0.288);
\draw[black, line width=0.4mm, fill=white] (2.59808,-9.5) circle (0.288);
\draw[black, line width=0.4mm, fill=white] (3.4641,-2) circle (0.288);
\draw[black, line width=0.4mm, fill=white] (3.4641,-3) circle (0.288);
\draw[black, line width=0.4mm, fill=white] (3.4641,-4) circle (0.288);
\draw[black, line width=0.4mm, fill=white] (3.4641,-5) circle (0.288);
\draw[black, line width=0.4mm, fill=white] (3.4641,-6) circle (0.288);
\draw[black, line width=0.4mm, fill=white] (3.4641,-7) circle (0.288);
\draw[black, line width=0.4mm, fill=white] (3.4641,-8) circle (0.288);
\draw[black, line width=0.4mm, fill=white] (3.4641,-9) circle (0.288);
\draw[black, line width=0.4mm, fill=white] (3.4641,-10) circle (0.288);
\draw[black, line width=0.4mm, fill=white] (4.33013,-2.5) circle (0.288);
\draw[black, line width=0.4mm, fill=white] (4.33013,-3.5) circle (0.288);
\draw[black, line width=0.4mm, fill=white] (4.33013,-4.5) circle (0.288);
\draw[black, line width=0.4mm, fill=white] (4.33013,-5.5) circle (0.288);
\draw[black, line width=0.4mm, fill=white] (4.33013,-6.5) circle (0.288);
\draw[black, line width=0.4mm, fill=white] (4.33013,-7.5) circle (0.288);
\draw[black, line width=0.4mm, fill=white] (4.33013,-8.5) circle (0.288);
\draw[black, line width=0.4mm, fill=white] (4.33013,-9.5) circle (0.288);
\draw[black, line width=0.4mm, fill=white] (5.19615,-1) circle (0.288);
\draw[black, line width=0.4mm, fill=white] (5.19615,-3) circle (0.288);
\draw[black, line width=0.4mm, fill=white] (5.19615,-4) circle (0.288);
\draw[black, line width=0.4mm, fill=white] (5.19615,-5) circle (0.288);
\draw[black, line width=0.4mm, fill=white] (5.19615,-6) circle (0.288);
\draw[black, line width=0.4mm, fill=white] (5.19615,-7) circle (0.288);
\draw[black, line width=0.4mm, fill=white] (5.19615,-8) circle (0.288);
\draw[black, line width=0.4mm, fill=white] (5.19615,-9) circle (0.288);
\draw[black, line width=0.4mm, fill=white] (6.06218,-2.5) circle (0.288);
\draw[black, line width=0.4mm, fill=white] (6.06218,-3.5) circle (0.288);
\draw[black, line width=0.4mm, fill=white] (6.06218,-4.5) circle (0.288);
\draw[black, line width=0.4mm, fill=white] (6.06218,-5.5) circle (0.288);
\draw[black, line width=0.4mm, fill=white] (6.06218,-6.5) circle (0.288);
\draw[black, line width=0.4mm, fill=white] (6.06218,-7.5) circle (0.288);
\draw[black, line width=0.4mm, fill=white] (6.06218,-8.5) circle (0.288);
\draw[black, line width=0.4mm, fill=white] (6.9282,-2) circle (0.288);
\draw[black, line width=0.4mm, fill=white] (6.9282,-4) circle (0.288);
\draw[black, line width=0.4mm, fill=white] (6.9282,-5) circle (0.288);
\draw[black, line width=0.4mm, fill=white] (6.9282,-6) circle (0.288);
\draw[black, line width=0.4mm, fill=white] (6.9282,-7) circle (0.288);
\draw[black, line width=0.4mm, fill=white] (6.9282,-8) circle (0.288);
  \end{tikzpicture}
  \end{center}
  \captionof{figure}{A lattice with insufficient space to build a spiral.}
  \label{fig:nospace}
\end{minipage}
\end{figure}

%

Similar to the Adaptive $\alpha$-Compression Algorithm, particles have two main types of states, dispersion and compression. 
The compression states are labeled to ensure that deviations from the desired structure do not happen.
Specifically, we encode sub-states of the compression state, which we label with $0$ to $6$ and $0^*$ to $5^*$, along with one of six possible ``parent'' directions.
States $0$ to $5$ are used for the initial six particles to attach around the food particle in counterclockwise order, while every subsequent particle in the spiral is in state $6$.
Additionally, for the initial six particles, the \emph{verified states} $0^*$ to $5^*$ are used to confirm that all of the six particles are in the correct state and direction before continuing to build the spiral.
The parent directions indicate the next particle in the spiral in the inwards direction.
More details on these states are given in Appendix~\ref{apx:spiralalgorithm}.


A particle in the dispersion state moves as a simple random walk over the underlying lattice, with moves into occupied lattice sites rejected.
A particle in a compression state on the other hand, does not move.
In our model, a particle is able to observe the states and directions of its neighbors. A particle, when activated, uses this information to potentially switch between states, and try to make a move if it is in the dispersion state after.


Algorithm~\ref{alg:spiralalgorithm} describes the behavior of a particle $v$ when it is activated.
The basic idea is that a particle switches to some compression state when it satisfies an ``\emph{attachment property}'' (for some state and direction), and to the dispersion state when it does not.
This ``attachment property'' is a predicate that considers a particle, its local neighborhood, a state~$s$ and direction~$d$. A particle satisfies the attachment property under a state~$s$ and a direction~$d$ if it satisfies the specific local conditions described in Appendix~\ref{apx:spiralprotocols}.
Importantly, every particle in a spiral (Figure~\ref{fig:spiral}) satisfies the attachment property with its current state and direction.





\begin{algorithm}[t]
\caption{Adaptive Spiraling. $\rho \in (0,{1}/{2})$ is a constant.}
\label{alg:spiralalgorithm}
\begin{algorithmic}[1]
\Procedure{Activate}{$v$}
\If{$v$ is a stable particle}
    \State $s \gets$ the base state of $v$.
    \If{$v$ is verifiable}
        \State Set the state of $v$ to $s^*$
    \Else
        \State Set the state of $v$ to $s$.
    \EndIf
\ElsIf{$v$ is an attachable particle for some state $s$ and direction $d$}
    \If{$v$ is currently in the dispersion state}
        \State Switch $v$ to compression state $s$ with parent direction $d$ with probability $\rho$.
    \Else
        \State Switch $v$ to compression state $s$ with parent direction $d$.
    \EndIf
\EndIf
\EndProcedure
\end{algorithmic}
\end{algorithm}

\begin{definition}
We say the \emph{base state} $\overline{s}$ of a particle of state $s$ is $x$ if $s \in \{x,x^*\}$. The following definitions are useful for the description of the algorithm and its proof:
\begin{itemize}
\item \textbf{Stable particle}: A particle in a compression state that satisfies the attachment property with its current parent direction $d$ and current state $s$.

\item \textbf{Unstable particle}: A particle in a compression state that does not satisfy the attachment property with its current parent direction $d$ and current state $s$.

\item \textbf{Attachable particle}: A particle that is unstable or in the dispersion state, that would satisfy the attachment property for some direction $d$ and state $s$.

\item \textbf{Verifiable particle}: Consider a stable particle $v$ of base state $\overline{s} \leq 5$. For $v$ to be stable, it must be one of the six particles circling the food particle. If $\overline{s} = 5$, we say $v$ is verifiable. Otherwise, denote by $u$ the particle one step counterclockwise of $v$ around the food particle. We say $v$ is verifiable if $u$ exists, is of compression state $(\overline{s}+1)^*$, and has $v$ as its parent.
\end{itemize}
\end{definition}
%
The algorithm contains a ``Verification'' step, which is used to confirm that all six particles neighboring the food are present in the correct states and directions, before continuing to build the spiral.
To do this, the six particles neighboring the food, while stable, can switch between the verified and unverified states, depending on whether they are verifiable.

If we reach a state where every particle is stable, assuming that there are at least six particles in the system, the only possible end result is a single spiral built around the food particle (Figure~\ref{fig:spiral}).
The end goal of the algorithm is thus for every particle to be stable.

We also define an ``\emph{attachment probability}'' $\rho \in (0,\frac{1}{2})$, which is used when a dispersion state particle is to join the spiral by switching to a compression state.
This attachment probability limits the speed of growth of the spiral, so that when a food particle vanishes and the spiral needs to dissipate, it will do so at a greater rate than it can grow.


\paragraph*{Analysis of Adaptive Spiraling}
If the food particle remains in a fixed position for a sufficiently long time, our ultimate objective is for all particles in the system to form a single spiral, with the food particle in the center.
If no food particle is present for a sufficiently long time, all particles should return to the dispersion state.
We show that these outcomes are achieved within a polynomial number of time steps in expectation.
The \emph{spiral} is our desired final configuration if the food exists and remains in place. We define the spiral 
as follows:

\begin{definition}[Spiral]
\label{defn:spiral}
A {\it spiral} is a maximal sequence of compression state particles $v_0, v_1, \dots, v_{k-1}$ such that 
\begin{enumerate}
\item for $0 \leq i \leq 5$, $v_i$ is in state $i^*$. For $6 \leq i \leq k-1$, $v_i$ is in state $6$;
\item the parent of $v_0$ is the food particle, and for each $1 \leq i \leq k-1$, $v_{i-1}$ is the parent of $v_i$;
\item the sequence of particles are in a ``counterclockwise spiral'' configuration originating from the food particle, as illustrated in Figure~\ref{fig:spiral}. 
\end{enumerate}
\end{definition}

\noindent
We call a particle a \emph{spiral particle} if it is part of some spiral.
Note that for a fixed food particle location, there are only six possible spirals that can be formed, each corresponding to a starting position adjacent to the food.
We observe that multiple spirals may start forming in the system,
after seven or more  particles join the spiral, only one spiral will continue to grow.

Recall that our algorithm's effectiveness is contingent on there being no change in the food's existence or position for a sufficient amount of time.
Thus, our proof starts from a point of time after which no change in the food particle occurs.
However, the configuration we start from may be completely arbitrary,
as the food particle could have been moved adversarially prior to when we begin the analysis, potentially creating difficult starting configurations.
For example, the adversarial movement of the food can cause the construction of partially built spirals, which may continue to grow and possibly self-sustain as information about the incomplete nature of the spiral or lack of food particle in the center takes time to propagate through the system.


For the purposes of the proof, we assume that there are at least six particles in the system, and that there is sufficient space for the particles to move around. To be precise, we do not consider configurations with enough particles for the formed spiral to make contact with itself via wraparound, blocking the remaining particles from joining the spiral (Figure~\ref{fig:nospace}).

Suppose that the food particle exists and is static.
To form a single spiral around the food particle, the first challenge to deal with is the lack of global knowledge of orientation.
We resolve this through the initial six particles surrounding the food particle.
We show that the algorithm, by obeying the attachment property, guarantees the formation of what we call a ``complete circle'', a circle of six particles labeled~$0^*$ to~$5^*$ in counterclockwise order around the food particle.
Once this is established, construction of the spiral begins by building a counterclockwise sequence of particles starting from the particle labeled~$5^*$.
Appendix~\ref{apx:spiralprotocols} 
explains in detail how this is achieved from any starting configuration.
Breaking symmetry in this manner is important because of the highly structured nature of Adaptive Spiraling. On the other hand, Adaptive $\alpha$-Compression does not have to deal with such an issue.


Similar to Adaptive $\alpha$-Compression, a key part of our proof is the eventual elimination of ``residual particles'', particles we do not want in the configuration.
For Adaptive Spiraling, a residual particle is defined as follows:

\begin{definition}[Residual Particle]
\label{dfn:residualparticles}
A \emph{residual particle} in a configuration is a particle
in a compression state that is not a spiral article but is either in state $6$ or is not adjacent to food.
\end{definition}


Residual particles pose two main problems: They obstruct particles in the dispersion state, preventing them from reaching the required locations to extend the spiral. Also, dispersion state particles may mistakenly attach to residual particles, creating more residual particles instead of making meaningful progress.

As in Adaptive $\alpha$-Compression, we apply Lemma~\ref{lemma:generallemmafinal} to argue that the number of residual particles decreases more quickly than it increases.
In this case, the random variable $X_t$ represents the number of residual particles $t$ time steps after some reference point. To show that the preconditions of Lemma~\ref{lemma:generallemmafinal}, for each possible configuration, we define an auxiliary graph with the particles as vertices, and parent directions as directed edges (Definition~\ref{dfn:auxiliarygraph})
The full argument is outlined in Appendix~\ref{apx:spiralprotocols}
, which we outline in the following paragraphs.

Under specific conditions, we show that each dispersion state particle that can become a residual particle on activation has a path over the auxiliary graph to either the food particle or an unstable residual particle - one that does not satisfy the attachment property and thus returns to the dispersion state on activation.
By showing that no vertex of the auxiliary graph has in-degree greater than~$1$, this allows us to upper bound the probability of gaining a new residual particle in the next iteration by the probability of losing a residual particle.

Finally, we show that no residual particles will form again after the number of residual particles hits $0$. This, with the assumption that the ``complete circle'' around the food particle as described earlier has fully formed, ensures that the only remaining particles are dispersion state particles are or part of a single spiral with the food particle in the center. Known results on the hitting time of exclusion processes gives us the remainder of the proof.
The elimination of residual particles also applies to the case where no food particle exists, this time guaranteeing that all remaining particles are in the dispersion state.


The following two theorems constitute the proof of correctness of the Adaptive Spiraling algorithm. The proofs are given in Appendix~\ref{apx:correctnessofspiralalgorithm}.

\begin{theorem}
\label{theorem:spiralwithfood}
From any starting configuration, given that the food particle exists and remains in a fixed position for a sufficient amount of time, we will reach a configuration where all particles form a single spiral around the food particle in a polynomial number of steps in expectation.
\end{theorem}

\begin{theorem}
\label{theorem:spiralnofood}
From any starting configuration, if there is no food particle for a sufficient amount of time, all particles will return to the dispersion state within a polynomial number of steps in expectation.
\end{theorem}


\begin{figure}[t]
\centering
\includegraphics[width=.35\linewidth]{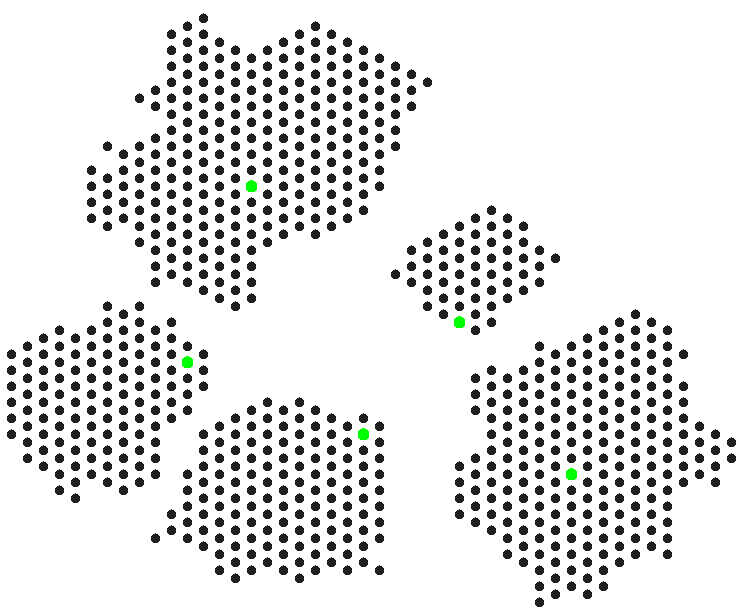}
\caption{Adaptive $\alpha$-Compression running with multiple food particles (shown in green).}
\label{fig:multiple_food_particles}
\end{figure}

\section{Conclusions}
We demonstrate 
through two algorithms 
that self-induced phase changes, here implemented via careful token passing in response to an environmental signal,  can be used as an algorithmic tool to transition between multiple desirable emergent macrostates.  We see this as a promising direction for collective systems of interacting agents in applications such as self-healing and search-and-rescue tasks in robot swarms.  In the first case, the macrostates represent two natural regimes in ferromagnetic systems corresponding to high temperature, where particles disperse, and low temperature, where particles compress.  Since both of these behaviors emerge from changing a simple parameter corresponding to temperature, the Adaptive $\alpha$-Compression algorithm can be thought of as particles adjusting their parameters based on individual experiences, such as whether they have been fed recently, in addition to bringing about a beneficial global outcome for the collective.  We show how careful token passing can overcome potential issues with particles self-actuating autonomously in a distributed setting.

The second algorithm demonstrates that the concept of self-induced phase changes can be used even in highly structured algorithms that are not motivated by simple physical systems of interacting particles.  The macrobehavior is encoded in states and the appropriate collective behavior emerges as the particles learn which phase of the algorithm is appropriate at any given time.  

In addition to seeming more reminiscent of how groups of organisms might gather and search for food when foraging, we see several advantages of Adaptive $\alpha$-Compression over Adaptive Spiraling.
First, the structured Adaptive Spiraling algorithm can fail in the presence of an obstruction, such as dead particles or other objects that prevent particles from occupying certain sites in the lattice.
Obstacles are especially disruptive to Adaptive Spiraling, as it relies on a very specific structure being maintained. If an obstacle is in the path of the spiral being built, the spiral will not be able to grow any further.
On the other hand, in the Adaptive $\alpha$-Compression Algorithm, any particle that has not been completely ``walled off'' by obstacles will eventually join a component containing the food particle. 
We note, however, that while all particles will eventually join a component containing a food particle, if there are many obstructions, this may prevent low-perimeter clusters from forming, even in the best situation.

A second advantage to the stochastic approach is
when we have multiple food particles where we  want the particles to form multiple compressed components.
In the Adaptive Spiraling algorithm, we have similar issues to the case of obstacles. If two food particles are close enough to each other for their clusters growing around them to meet, these clusters will obstruct each other, preventing any further growth.
The Adaptive $\alpha$-Compression Algorithm however, with a small modification to prevent clusters from merging,
continues to ensure that every particle eventually joins a cluster containing a food particle, although again, a crowded field of deliberately placed food particles may prevent clusters from reaching low perimeter configurations. One might note though, that the cases where undesirable low perimeter configurations are not reached are uncommon. Figure~\ref{fig:multiple_food_particles} illustrates the typical configurations that may form around multiple food particles, and Figure~\ref{fig:multiple_food_stuck} gives an example of food particles placed in specific locations preventing low-perimeter configurations from being reached.  We believe that the proofs for the single food particle case will extend to the multiple food particle setting  and suggest this for future study.

Also as future work, we believe that, when carefully translated to the amoebot model~\cite{Daymude2021-canonicalamoebot}, one can show that the Adaptive Spiral Algorithm to be the first self-stabilizing hexagon formation algorithm in the presence of a seed under this model.



\begin{figure}
\begin{subfigure}[b]{.5\linewidth}
  \begin{center}
  \begin{tikzpicture}[x=0.3cm,y=0.3cm]
  \draw[lightgray] (0,-7) -- (11.2583,-0.5);
\draw[lightgray] (0,-7) -- (6.9282,-11);
\draw[lightgray] (0,-10) -- (12.1244,-3);
\draw[lightgray] (0,-10) -- (1.73205,-11);
\draw[lightgray] (3.4641,0) -- (12.1244,-5);
\draw[lightgray] (3.4641,0) -- (3.4641,-11);
\draw[lightgray] (7.79423,-0.5) -- (7.79423,-10.5);
\draw[lightgray] (12.1244,-1) -- (12.1244,-10);
\draw[lightgray] (0,-4) -- (6.9282,0);
\draw[lightgray] (0,-4) -- (11.2583,-10.5);
\draw[lightgray] (5.19615,0) -- (12.1244,-4);
\draw[lightgray] (5.19615,0) -- (5.19615,-11);
\draw[lightgray] (9.52628,-0.5) -- (9.52628,-10.5);
\draw[lightgray] (10.3923,-11) -- (12.1244,-10);
\draw[lightgray] (0,-1) -- (1.73205,0);
\draw[lightgray] (0,-1) -- (12.1244,-8);
\draw[lightgray] (0,-1) -- (0,-10);
\draw[lightgray] (1.73205,-11) -- (12.1244,-5);
\draw[lightgray] (11.2583,-0.5) -- (11.2583,-10.5);
\draw[lightgray] (0.866025,-10.5) -- (12.1244,-4);
\draw[lightgray] (5.19615,-11) -- (12.1244,-7);
\draw[lightgray] (0,-5) -- (8.66025,0);
\draw[lightgray] (0,-5) -- (10.3923,-11);
\draw[lightgray] (0.866025,-0.5) -- (12.1244,-7);
\draw[lightgray] (0.866025,-0.5) -- (0.866025,-10.5);
\draw[lightgray] (0,-8) -- (12.1244,-1);
\draw[lightgray] (0,-8) -- (5.19615,-11);
\draw[lightgray] (3.4641,-11) -- (12.1244,-6);
\draw[lightgray] (0,-2) -- (3.4641,0);
\draw[lightgray] (0,-2) -- (12.1244,-9);
\draw[lightgray] (2.59808,-0.5) -- (2.59808,-10.5);
\draw[lightgray] (6.9282,0) -- (12.1244,-3);
\draw[lightgray] (6.9282,0) -- (6.9282,-11);
\draw[lightgray] (4.33013,-0.5) -- (4.33013,-10.5);
\draw[lightgray] (8.66025,0) -- (12.1244,-2);
\draw[lightgray] (8.66025,0) -- (8.66025,-11);
\draw[lightgray] (6.06218,-0.5) -- (6.06218,-10.5);
\draw[lightgray] (10.3923,0) -- (12.1244,-1);
\draw[lightgray] (10.3923,0) -- (10.3923,-11);
\draw[lightgray] (0,-6) -- (10.3923,0);
\draw[lightgray] (0,-6) -- (8.66025,-11);
\draw[lightgray] (0,-9) -- (12.1244,-2);
\draw[lightgray] (0,-9) -- (3.4641,-11);
\draw[lightgray] (8.66025,-11) -- (12.1244,-9);
\draw[lightgray] (0,-3) -- (5.19615,0);
\draw[lightgray] (0,-3) -- (12.1244,-10);
\draw[lightgray] (6.9282,-11) -- (12.1244,-8);
\draw[lightgray] (1.73205,0) -- (12.1244,-6);
\draw[lightgray] (1.73205,0) -- (1.73205,-11);
\draw[fill=black] (0.866025,-6.5) circle (0.288);
\draw[black, line width=0.4mm, fill=white] (2.59808,-3.5) circle (0.288);
\draw[black, line width=0.4mm, fill=white] (2.59808,-4.5) circle (0.288);
\draw[black, line width=0.4mm, fill=white] (2.59808,-5.5) circle (0.288);
\draw[black, line width=0.4mm, fill=white] (2.59808,-6.5) circle (0.288);
\draw[black, line width=0.4mm, fill=white] (2.59808,-7.5) circle (0.288);
\draw[black, line width=0.4mm, fill=white] (3.4641,-3) circle (0.288);
\draw[black, line width=0.4mm, fill=white] (4.33013,-2.5) circle (0.288);
\draw[black, line width=0.4mm, fill=white] (4.33013,-4.5) circle (0.288);
\draw[black, line width=0.4mm, fill=white] (4.33013,-5.5) circle (0.288);
\draw[black, line width=0.4mm, fill=white] (4.33013,-6.5) circle (0.288);
\draw[black, line width=0.4mm, fill=white] (4.33013,-7.5) circle (0.288);
\draw[black, line width=0.4mm, fill=white] (4.33013,-8.5) circle (0.288);
\draw[black, line width=0.4mm, fill=white] (5.19615,-2) circle (0.288);
\draw[black, line width=0.4mm, fill=white] (5.19615,-4) circle (0.288);
\draw[black, line width=0.4mm, fill=white] (5.19615,-9) circle (0.288);
\draw[black, line width=0.4mm, fill=white] (6.06218,-1.5) circle (0.288);
\draw[black, line width=0.4mm, fill=white] (6.06218,-3.5) circle (0.288);
\draw[black, line width=0.4mm, fill=white] (6.06218,-5.5) circle (0.288);
\draw[black, line width=0.4mm, fill=white] (6.06218,-6.5) circle (0.288);
\draw[black, line width=0.4mm, fill=white] (6.06218,-7.5) circle (0.288);
\draw[black, line width=0.4mm, fill=white] (6.06218,-9.5) circle (0.288);
\draw[black, line width=0.4mm, fill=white] (6.9282,-1) circle (0.288);
\draw[black, line width=0.4mm, fill=white] (6.9282,-3) circle (0.288);
\draw[black, line width=0.4mm, fill=white] (6.9282,-5) circle (0.288);
\draw[black, line width=0.4mm, fill=white] (6.9282,-8) circle (0.288);
\draw[black, line width=0.4mm, fill=white] (6.9282,-10) circle (0.288);
\draw[black, line width=0.4mm, fill=white] (7.79423,-0.5) circle (0.288);
\draw[black, line width=0.4mm, fill=white] (7.79423,-2.5) circle (0.288);
\draw[fill=black] (7.79423,-4.5) circle (0.288);
\draw[fill=black] (7.79423,-6.5) circle (0.288);
\draw[black, line width=0.4mm, fill=white] (7.79423,-8.5) circle (0.288);
\draw[black, line width=0.4mm, fill=white] (7.79423,-10.5) circle (0.288);
\draw[black, line width=0.4mm, fill=white] (8.66025,-1) circle (0.288);
\draw[black, line width=0.4mm, fill=white] (8.66025,-3) circle (0.288);
\draw[black, line width=0.4mm, fill=white] (8.66025,-8) circle (0.288);
\draw[black, line width=0.4mm, fill=white] (8.66025,-10) circle (0.288);
\draw[black, line width=0.4mm, fill=white] (9.52628,-1.5) circle (0.288);
\draw[black, line width=0.4mm, fill=white] (9.52628,-3.5) circle (0.288);
\draw[black, line width=0.4mm, fill=white] (9.52628,-4.5) circle (0.288);
\draw[black, line width=0.4mm, fill=white] (9.52628,-5.5) circle (0.288);
\draw[black, line width=0.4mm, fill=white] (9.52628,-6.5) circle (0.288);
\draw[black, line width=0.4mm, fill=white] (9.52628,-7.5) circle (0.288);
\draw[black, line width=0.4mm, fill=white] (9.52628,-9.5) circle (0.288);
\draw[black, line width=0.4mm, fill=white] (10.3923,-2) circle (0.288);
\draw[black, line width=0.4mm, fill=white] (10.3923,-9) circle (0.288);
\draw[black, line width=0.4mm, fill=white] (11.2583,-2.5) circle (0.288);
\draw[black, line width=0.4mm, fill=white] (11.2583,-3.5) circle (0.288);
\draw[black, line width=0.4mm, fill=white] (11.2583,-4.5) circle (0.288);
\draw[black, line width=0.4mm, fill=white] (11.2583,-5.5) circle (0.288);
\draw[black, line width=0.4mm, fill=white] (11.2583,-6.5) circle (0.288);
\draw[black, line width=0.4mm, fill=white] (11.2583,-7.5) circle (0.288);
\draw[black, line width=0.4mm, fill=white] (11.2583,-8.5) circle (0.288);
\draw[black, line width=0.5mm] (6.32199,-3.35) -- (6.6684,-3.15);
\draw[black, line width=0.5mm] (5.45596,-9.15) -- (5.80237,-9.35);
\draw[black, line width=0.5mm] (4.33013,-6.2) -- (4.33013,-5.8);
\draw[black, line width=0.5mm] (8.05404,-8.35) -- (8.40045,-8.15);
\draw[black, line width=0.5mm] (11.2583,-4.2) -- (11.2583,-3.8);
\draw[black, line width=0.5mm] (6.06218,-7.2) -- (6.06218,-6.8);
\draw[black, line width=0.5mm] (6.32199,-7.65) -- (6.6684,-7.85);
\draw[black, line width=0.5mm] (5.45596,-1.85) -- (5.80237,-1.65);
\draw[black, line width=0.5mm] (11.2583,-8.2) -- (11.2583,-7.8);
\draw[black, line width=0.5mm] (2.59808,-7.2) -- (2.59808,-6.8);
\draw[black, line width=0.5mm] (7.18801,-4.85) -- (7.53442,-4.65);
\draw[black, line width=0.5mm] (8.92006,-3.15) -- (9.26647,-3.35);
\draw[black, line width=0.5mm] (4.33013,-8.2) -- (4.33013,-7.8);
\draw[black, line width=0.5mm] (4.58993,-8.65) -- (4.93634,-8.85);
\draw[black, line width=0.5mm] (10.6521,-8.85) -- (10.9985,-8.65);
\draw[black, line width=0.5mm] (8.92006,-7.85) -- (9.26647,-7.65);
\draw[black, line width=0.5mm] (8.05404,-0.65) -- (8.40045,-0.85);
\draw[black, line width=0.5mm] (10.6521,-2.15) -- (10.9985,-2.35);
\draw[black, line width=0.5mm] (2.85788,-3.35) -- (3.20429,-3.15);
\draw[black, line width=0.5mm] (4.33013,-5.2) -- (4.33013,-4.8);
\draw[black, line width=0.5mm] (9.52628,-5.2) -- (9.52628,-4.8);
\draw[black, line width=0.5mm] (7.18801,-0.85) -- (7.53442,-0.65);
\draw[black, line width=0.5mm] (11.2583,-3.2) -- (11.2583,-2.8);
\draw[black, line width=0.5mm] (6.06218,-6.2) -- (6.06218,-5.8);
\draw[black, line width=0.5mm] (7.18801,-10.15) -- (7.53442,-10.35);
\draw[black, line width=0.5mm] (9.78609,-1.65) -- (10.1325,-1.85);
\draw[black, line width=0.5mm] (11.2583,-7.2) -- (11.2583,-6.8);
\draw[black, line width=0.5mm] (2.59808,-6.2) -- (2.59808,-5.8);
\draw[black, line width=0.5mm] (2.59808,-4.2) -- (2.59808,-3.8);
\draw[black, line width=0.5mm] (5.45596,-3.85) -- (5.80237,-3.65);
\draw[black, line width=0.5mm] (4.58993,-4.35) -- (4.93634,-4.15);
\draw[black, line width=0.5mm] (9.52628,-4.2) -- (9.52628,-3.8);
\draw[black, line width=0.5mm] (6.32199,-1.35) -- (6.6684,-1.15);
\draw[black, line width=0.5mm] (9.78609,-9.35) -- (10.1325,-9.15);
\draw[black, line width=0.5mm] (8.05404,-10.35) -- (8.40045,-10.15);
\draw[black, line width=0.5mm] (11.2583,-6.2) -- (11.2583,-5.8);
\draw[black, line width=0.5mm] (6.32199,-5.35) -- (6.6684,-5.15);
\draw[black, line width=0.5mm] (8.92006,-1.15) -- (9.26647,-1.35);
\draw[black, line width=0.5mm] (2.59808,-5.2) -- (2.59808,-4.8);
\draw[black, line width=0.5mm] (6.32199,-9.65) -- (6.6684,-9.85);
\draw[black, line width=0.5mm] (9.52628,-7.2) -- (9.52628,-6.8);
\draw[black, line width=0.5mm] (7.18801,-2.85) -- (7.53442,-2.65);
\draw[black, line width=0.5mm] (8.92006,-9.85) -- (9.26647,-9.65);
\draw[black, line width=0.5mm] (8.05404,-2.65) -- (8.40045,-2.85);
\draw[black, line width=0.5mm] (4.33013,-7.2) -- (4.33013,-6.8);
\draw[black, line width=0.5mm] (11.2583,-5.2) -- (11.2583,-4.8);
\draw[black, line width=0.5mm] (7.18801,-8.15) -- (7.53442,-8.35);
\draw[black, line width=0.5mm] (3.72391,-2.85) -- (4.07032,-2.65);
\draw[black, line width=0.5mm] (4.58993,-2.35) -- (4.93634,-2.15);
\draw[black, line width=0.5mm] (9.52628,-6.2) -- (9.52628,-5.8);
  \end{tikzpicture}
  \end{center}
  \label{fig:multiple_food_stuck1}
\end{subfigure}%
\hfill
\begin{subfigure}[b]{.5\linewidth}
  \begin{center}
  \begin{tikzpicture}[x=0.4cm,y=0.4cm]
  \draw[lightgray] (5.19615,-6) -- (7.79423,-4.5);
\draw[lightgray] (2.59808,-0.5) -- (2.59808,-5.5);
\draw[lightgray] (0,-2) -- (3.4641,0);
\draw[lightgray] (0,-2) -- (6.9282,-6);
\draw[lightgray] (6.9282,0) -- (7.79423,-0.5);
\draw[lightgray] (6.9282,0) -- (6.9282,-6);
\draw[lightgray] (0,-3) -- (5.19615,0);
\draw[lightgray] (0,-3) -- (5.19615,-6);
\draw[lightgray] (3.4641,0) -- (7.79423,-2.5);
\draw[lightgray] (3.4641,0) -- (3.4641,-6);
\draw[lightgray] (7.79423,-0.5) -- (7.79423,-5.5);
\draw[lightgray] (4.33013,-0.5) -- (4.33013,-5.5);
\draw[lightgray] (0,-4) -- (6.9282,0);
\draw[lightgray] (0,-4) -- (3.4641,-6);
\draw[lightgray] (6.9282,-6) -- (7.79423,-5.5);
\draw[lightgray] (0,0) -- (7.79423,-4.5);
\draw[lightgray] (0,0) -- (0,-6);
\draw[lightgray] (1.73205,-6) -- (7.79423,-2.5);
\draw[lightgray] (5.19615,0) -- (7.79423,-1.5);
\draw[lightgray] (5.19615,0) -- (5.19615,-6);
\draw[lightgray] (0,-5) -- (7.79423,-0.5);
\draw[lightgray] (0,-5) -- (1.73205,-6);
\draw[lightgray] (0.866025,-0.5) -- (0.866025,-5.5);
\draw[lightgray] (0,-1) -- (1.73205,0);
\draw[lightgray] (0,-1) -- (7.79423,-5.5);
\draw[lightgray] (6.06218,-0.5) -- (6.06218,-5.5);
\draw[lightgray] (3.4641,-6) -- (7.79423,-3.5);
\draw[lightgray] (0,-6) -- (7.79423,-1.5);
\draw[lightgray] (1.73205,0) -- (7.79423,-3.5);
\draw[lightgray] (1.73205,0) -- (1.73205,-6);
\draw[fill=black] (0.866025,-2.5) circle (0.288);
\draw[black, line width=0.4mm, fill=white] (0.866025,-4.5) circle (0.288);
\draw[black, line width=0.4mm, fill=white] (1.73205,-4) circle (0.288);
\draw[black, line width=0.4mm, fill=white] (2.59808,-3.5) circle (0.288);
\draw[fill=black] (2.59808,-5.5) circle (0.288);
\draw[black, line width=0.4mm, fill=white] (3.4641,-3) circle (0.288);
\draw[black, line width=0.4mm, fill=white] (4.33013,-2.5) circle (0.288);
\draw[black, line width=0.4mm, fill=white] (5.19615,-2) circle (0.288);
\draw[black, line width=0.4mm, fill=white] (6.06218,-1.5) circle (0.288);
\draw[fill=black] (6.9282,-1) circle (0.288);
\draw[black, line width=0.5mm] (2.85788,-3.35) -- (3.20429,-3.15);
\draw[black, line width=0.5mm] (6.32199,-1.35) -- (6.6684,-1.15);
\draw[black, line width=0.5mm] (1.12583,-4.35) -- (1.47224,-4.15);
\draw[black, line width=0.5mm] (5.45596,-1.85) -- (5.80237,-1.65);
\draw[black, line width=0.5mm] (3.72391,-2.85) -- (4.07032,-2.65);
\draw[black, line width=0.5mm] (4.58993,-2.35) -- (4.93634,-2.15);
\draw[black, line width=0.5mm] (1.99186,-3.85) -- (2.33827,-3.65);
  \end{tikzpicture}
  \end{center}
  \label{fig:multiple_food_stuck2}
\end{subfigure}%
\caption{Two configurations where the positions of the food particles (denoted by the black particles) prevent any of the compression state particles from making a move.}
\label{fig:multiple_food_stuck}
\end{figure}
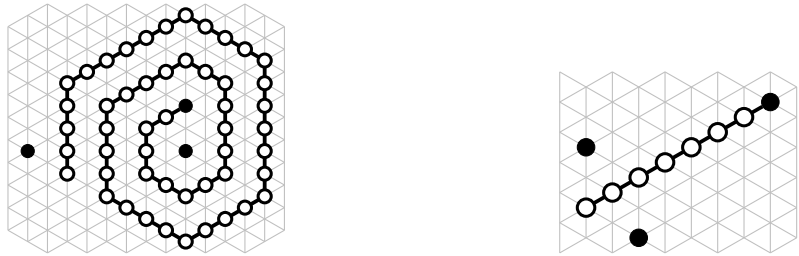

\bibliographystyle{plainurl}
\bibliography{bibfile}

\vskip.5in
\begin{appendices}

\section{Details for the Adaptive $\alpha$-Compression Algorithm}
\label{apx:compressionbasedalgorithm}
We proceed by giving the final details of the algorithm and proofs for the Adaptive $\alpha$-Compression Algorithm that exceeded the space in the main text.

\subsection{Establishing the state invariant (Lemma~\ref{lemma:keyinvariant})}\label{apx:proofofkeyinvariant}
Recall that the state invariant is satisfied if each  component  contains a particle in state $C_{GF}, C_F$ or $DT$. Lemma~\ref{lemma:keyinvariant} states that the state invariant always holds if we start with all particles in the dispersion state.

\begin{proof}[Proof of Lemma~\ref{lemma:keyinvariant}]
We would like to show that the state invariant holds throughout the Adaptive $\alpha$-Compression Algorithm.  Consider any configuration where the state invariant currently holds, and let $u$ be the next particle to be activated. We show that the state invariant continues to hold after both the state change step and the particle movement step of the algorithm.

\vskip.1in
\noindent \underbar{\textbf{State Change Step:}}
If particle $u$ is in the dispersion state,
if $u$ switches to state $C_F$ when next to a food particle, the component $u$ joins automatically satisfies the invariant. If $u$ is consuming a growth token from a neighboring compression state particle to switch to a compression state, it must be joining an existing component that previously already satisfied the invariant, and thus will continue to do so.

If $u$ is in the dispersion token state, activating $u$ can potentially split the component it is in into multiple components. However, as all neighbors of $u$ will also be set to state $DT$, each of these new components will contain a particle of state $DT$.

If $u$ is in a compression state, particle $u$ transferring a growth token or generating a new growth token does not affect the invariant. Thus we only need to consider the case where $u$ switches to the dispersion state while spreading the dispersion token, in which case the argument used when $u$ is in the dispersion token state applies in the same way.

\vskip.1in
\noindent \underbar{\textbf{Particle Movement Step:}}
If $u$ is not in a compression state, it is clear that this step will not affect the invariant. Thus we assume $u$ to be in a compression state. To check that the state invariant is satisfied, we only need to look at the components in the new configuration that have been changed in some way by the movement of $u$ - in other words, the component $u$ joins after its movement, and the component(s) $u$ leaves behind after its movement (which may be the same component as the one $u$ joins).
If $u$ joins a component it was not a part of in the original configuration, this component must previously already satisfy the invariant, and will continue to do so even with $u$ attached.
Thus, we only need to check the component(s) that $u$ leaves behind after its movement.

There are two possible ways moving $u$ could cause the components it left behind to no longer satisfy the invariant.
The first is when $u$ bridges two components, separating these components after the move, and only one of these two components has a particle of state $C_F$, $C_{GF}$ or $DT$.
The second possibility is when $u$ is in state $C_F$ or $C_{GF}$ - the component could have no other particle of state $C_F$, $C_{GF}$ of $DT$, and moving $u$ could make it detach from the component or make $u$ unset its food bit.

Since $u$ is currently in a compression state, $u$ could have either been in the dispersion state or a compression state prior to its activation. If $u$ was in the dispersion state, any component(s) $u$ is part of before the movement would have already satisfied the invariant without $u$, and so would continue to satisfy the invariant even after $u$ is moved.

We can thus assume $u$ was in a compression state prior to its activation. As $u$ had passed the state change step of the algorithm without switching to the dispersion state, if it is neighboring the food particle, all compression state neighbors it shares with the food particle must be in states $C_F$ or $C_{GF}$.

In the first case stated above, as $u$ cannot make moves that locally disconnects its neighbors (which can include the food particle), the only way $u$ can disconnect two components is when the food particle continues to connect the neighbors of $u$ corresponding to these components. This implies that each of the two components will have a particle in state $C_F$ or $C_{GF}$.

In the second case,
we consider the shared compression state neighbors between $u$ and the food particle, which must be in states $C_F$ or $C_{GF}$.
Any component $u$ leaves behind (disconnects from) after its movement must contain one of these shared neighbors. This is because otherwise, $u$'s movement would locally disconnect this component from the food particle, which it cannot do.
\end{proof}


\subsection{Expected time to remove residual particles (Lemma~\ref{lemma:generallemmafinal})}
\label{apx:proofofgenerallemma}
Next, we prove Lemma~\ref{lemma:generallemmafinal}, establishing a polynomial bound on the expected time to remove residual particles, by first showing the
following two lemmas.
\begin{lemma}
\label{lemma:generallemma1}
Fix $0 < \eta < 1$. For any integer $k \geq 1$, Define a sequence $(W_t^{(k)})_{t \in \mathbbNzero}$, where $W^{(k)}_0 = k$, and 
\begin{align*}
W^{(k)}_{t+1} = \begin{cases}
W^{(k)}_t-1 &\text{with probability ${1}/{(1+\eta)}$}\\
W^{(k)}_t+1 &\text{with probability ${\eta}/{(1+\eta)}$}
\end{cases}.
\end{align*}
Then $\ex\left[\min\{t \geq 0 \mid W^{(k)}_t = 0\}\right] = k {(1+\eta)}/{(1-\eta)}$
\end{lemma}

\begin{proof}
For any $k \in \mathbbNzero$, denote $S_k := \ex\left[\min\{t \geq 0 \mid W^{(k)}_t = 0\}\right]$.
Now consider the random sequence $(W^{(k)}_t)_{t \geq 0}$. We define $T^{(k)}_{k+1} := 0$ and for $i \in \{k,k-1,\dots,1\}$, let $T^{(k)}_i := \min\{t \geq 0 \mid X^{(k)}_t \leq i\} - T^{(k)}_{i+1}$ denote the number of time steps after $T^{(k)}_{i+1}$ before the first time step $t$ where $W^{(k)}_t \leq i$. We observe that each $W^{(k)}_i$ is identically distributed, with $\ex T^{(k)}_i = \ex T^{(k)}_1 = S_1$.
Also, we observe that $S_k = \ex [\sum_{i=1}^{k} T^{(k)}_i]$, so $S_k = k\cdot S_1$ for all $k$.

We compute $S_1$ by conditioning on the first step:
\begin{align*}
S_1 &= \frac{1}{1+\eta}(1) + \frac{\eta}{1+\eta}(S_2+1)\\
&= 1 + \frac{\eta}{1+\eta}\cdot 2S_1\\
S_1 &= \frac{1+\eta}{1-\eta}
\end{align*}
This gives us $\ex\left[\min\{t \geq 0 \mid W^{(k)}_t = 0\}\right] = S_k = k\cdot\frac{1+\eta}{1-\eta}$.
\end{proof}

\begin{lemma}
\label{lemma:generallemma2}
Fix integers $n, k \geq 1$, and $0 < \eta < 1$. Consider a random sequence of probabilities $(r_t)_{t \in \mathbbNzero}$, such that $\frac{1+\eta}{n} \leq r_t \leq 1$ for all $t \in \mathbbNzero$. Now consider a sequence $(Y_t)_{t \in \mathbbNzero}$, where $Y_0 = k$, and 
\begin{align*}
Y_{t+1} = \begin{cases}
Y_t-1 &\text{with probability $r_t /{(1+\eta)}$};\\
Y_t+1 &\text{with probability $r_t {\eta}/{(1+\eta)}$};\\
Y_t &\text{with probability $1-r_t$}.
\end{cases}
\end{align*}
Then $\ex\left[\min\{t \geq 0 \mid Y_t = 0\}\right] \leq {nk}/{(1-\eta)}$
\end{lemma}

\begin{proof}
We observe that the sequence $(Y_t)_{t \in \mathbb{N}_{\geq 0}}$ behaves in the same way as $(W^{(k)}_t)_{t \in \mathbb{N}_{\geq 0}}$ from Lemma~\ref{lemma:generallemma1}, except that on each round, we have a probability $1-r_t$ of doing nothing. As $r_t \geq \frac{1+\eta}{n}$ always, we have $\ex\left[\min\{t \geq 0 \mid Y_t = 0\}\right] \leq {nk}/{(1-\eta)}$.
\end{proof}

This allows us next to bound the time to remove residual particles.

\begin{proof}[Proof of Lemma~\ref{lemma:generallemmafinal}]
We define random variables $r_t := (1+\eta)p_t$ for all $t \in \mathbbNzero$, which allows us to couple the sequences $(X_t)_{t\in \mathbb{N}_{\geq 0}}$ and $(Y_t)_{t\in \mathbb{N}_{\geq 0}}$ so that $X_t$ decreases by (at least) one if and only if $Y_t$ decreases by one and $X_t$ increases by (at least) one if and only if $Y_t$ increases by one. This shows that for all $t$ and real numbers $x$, $Pr(X_t \geq x) \leq Pr(Y_t \geq x)$ ($X_t$ is stochastically dominated by $Y_t$). Let $T_X := \ex\left[\min\{t \geq 0 \mid X_t = 0\}\right]$ and $T_Y := \ex\left[\min\{t \geq 0 \mid Y_t = 0\}\right]$. We can show that $T_X$ is also stochastically dominated by $T_Y$ as follows:
\begin{align*}
Pr(T_X \geq t) &= Pr(X_{t-1} \geq 1) \\
&\leq Pr(Y_{t-1} \geq 1)\\
&= Pr(T_Y \geq t)
\end{align*}
Hence $\ex \left[T_X\right] = \sum_{t \in \mathbbNzero} Pr(T_X \geq t) \leq \sum_{t \in \mathbbNzero} Pr(T_Y \geq t) = \ex \left[T_Y\right] \leq {nk}/{(1-\eta)}$.
\end{proof}


\subsection{Eliminating residual particles}
\label{apx:compressioneliminatingresidual}
Showing the following three lemmas allow us to conclude that if there is no change in the food particle for long enough, we will reach a stage where $DT$-state particles will no longer form, and particles in compression states can only be in components that are in contact with the food particle. If there is no food particle, this necessarily means that all particles will end up in the dispersion state.

\begin{lemma}
\label{lemma:reachpotentialzero}
We start from a configuration where the state invariant currently holds. Assume that there is no change in the food particle from the current point on. Then the expected number of steps before we reach a configuration with no residual components is at most ${2n^2}/{(1-6\rho)}$.
\end{lemma}

\begin{proof}
We consider the sequence of potentials $(\Phi(\sigma_t))_{t \in \mathbbNzero}$ with $\sigma_0$ representing the configuration we start with, and make use of Lemma~\ref{lemma:generallemmafinal}.

By the definition of a residual component, as long as there exists a residual component, there will be at least one particle that when activated, will switch to the dispersion state, reducing the current potential by at least one. Thus, the probability of decreasing the current potential by at least one in any step is at least ${1}/{n}$.

The consumption of a growth token to add a new compression state particle to a residual component does not change the current potential. Switching particles to the dispersion token state also does not affect the potential.

With no movement of the food particle, there are only two ways for the potential to increase. The first is when a new growth token is generated, and the second is when a dispersion state particle switches to a compression state by being next to a food particle. Either of these increase the current potential by exactly $1$.
Both of these can only happen through the activation of particles adjacent to food, and only happen with probability $\rho$. As there are at most $6$ of these, this happens with probability at most ${6\rho}/{n} < {1}/{n}$.

As $\Phi(\sigma_0) \leq 2n$, by Lemma~\ref{lemma:generallemmafinal}, within ${2n^2}/{(1-6\rho)}$ steps in expectation, we will either reach a state $\sigma$ with $\Phi(\sigma) = 0$, or a state with no residual components, whichever comes first. A state with $\Phi(\sigma) = 0$ necessarily also has no residual components, completing the proof.
\end{proof}

\begin{lemma}
We start from a configuration where the state invariant currently holds. Assume that there is no change in the food particle from the current point on.
If there are no residual particles, then no residual particle can be generated from then on.
\end{lemma}

\begin{proof}
With no residual components, all particles in the configuration will have the food bit set if and only if they are adjacent to food. There are also no $DT$ particles, so no particle on activation will switch to the $DT$ state. When a particle moves, it will set its food bit accordingly, so no residual components will form.
\end{proof}

\begin{lemma}
Suppose that the state invariant holds and there are no residual particles. Then all components will neighbor the food particle.
\end{lemma}

\begin{proof}
With no residual components, each component must contain a compression state particle with the food bit set, which must necessarily neighbor the food particle.
\end{proof}


\subsection{Bounding the time to compress (Lemma~\ref{lemma:allswitchtocompressionstate})}
\label{apx:compressionstatelemmaproof}
We will now verify that particles revert to the compression state if a food particle remains present once there are no residual particles.   We will refer to  connected components consisting of particles in the compression state together with any food particle as {\it clusters}.  (Note that these are different from components.)


\begin{proof}[Proof of Lemma~\ref{lemma:allswitchtocompressionstate}]
We consider the amount of time it takes for a new particle to join the cluster assuming that another particle does not join the cluster first. For a particle to join the cluster, it must be activated while adjacent to a particle in the cluster that is holding a growth token. We can upper bound the expected amount of time before this happens by the amount of time it takes for all particles in the cluster to obtain a growth token, followed by the amount of time it takes for a particle to ``collide'' with the cluster.

First, we claim that assuming no new particles join the cluster, then within a polynomial amount of steps, all particles in the cluster will be holding a growth token.
To see this, note that as there are at most $n$ particles in the cluster, we only need an upper bound on the amount of time it takes for a new growth token to be generated, starting from an arbitrary configuration of the cluster. Unfortunately, as a particle can only hold one growth token at a time, new growth tokens can only be generated when a particle without a growth token neighbors the food particle. We thus need to bound the expected amount of time it takes for a compression state particle without a growth token to neighbor the food particle.

Mark a particle that currently does not hold a growth token. If a neighboring particle containing a growth token is activated and transfers its growth token to the marked particle, we move the mark to said neighboring particle. If a neighboring particle without a growth token is activated, for the sake of our analysis we may pretend it does, randomly picking an outgoing direction, and receiving the mark if the marked particle happens to be its neighbor in the picked direction (even though no growth token is transferred).
As we assume no new particle joins the cluster, the set of particles in the cluster (the set of food particles and compression state particles) is fixed. Movement of the particles in the cluster in effect only reconfigures the graph's edges.
The mark moving in this manner is thus equivalent to following the $d_{max}$-random walk over
a connected evolving graph, which is a sequence $(G_t)_{t \in \mathbb{N}_{\geq 0}}$ of connected graphs sharing the same vertex set but with possibly differing edges, as described in~\cite{RandomWalksDynamicGraphs, RandomWalksRecurringTopologies}.
This gives us an expected hitting time upper bound of $O(n^3 \log n)$~\cite{RandomWalksRecurringTopologies} activations of the marked particle, or $O(n^4 \log n)$ iterations (the asymptotic bound is the same even as we take into account that the probability of generating a new growth token is a non-zero constant $p < 1$).

Note that while in practice it is not possible for the mark to actually transfer to the food particle, we are only concerned with the chain up to the point where the mark neighbors the food particle, so this analysis sufficiently shows that a new growth token will be generated in expected polynomial time.

Finally, we claim that if  every particle in the cluster already has a growth token, and if there are still particles not in the compression state, then a
new particle will ``collide'' with the cluster and join it in polynomial expected time.
Note that the particles in the dispersion state follow a simple exclusion process. By considering obstructed movements as swap moves rather than rejected moves, we can analyze the hitting time of a simple exclusion process as independent simple random walks over the lattice.
If the food particle and the particles in the cluster do not exist,
as our triangular lattice is a regular graph, a simple upper bound for the number of steps before some dispersion state particle reaches the current is location of the food particle is $O(N^2)$~\cite{lovasz1993random} activations of any one particle, which translates to $O(N^2n)$ iterations.
In practice however, the dispersion state particle must collide with either the food particle or a compression state particle before this happens, which switches it to the compression state with a constant probability $p$. This gives us a loose upper bound of $O(N^2n + n^5 \log n)$ iterations for a new particle joining the cluster, giving us the claim, and hence the Lemma.
\end{proof}

\section{Details for the Adaptive Spiraling Algorithm}
\label{apx:spiralalgorithm}

The spiral algorithm proceeds in stages, with the most delicate being at the onset when the food particle is first discovered.  Particles in the compression state start attaching to the food in cyclic order until the food is completely surrounded by 6 particles.  These particles are {\it verified} or {\it unverified}, with the verification process ensuring that the particles are consistently numbered and not in a hybrid state where multiple particles are labeled as the initial point on the spiral or particles have detached and reattached out of sequence.  The verification process ensures that particles are labeled $0$ through $5$ in cyclic order.  In the verification process, particles move from label an unverified label $x$ to a verified label $x*$, $x \in \{0,...,5\},$ to indicate that the particle is correctly labeled.  All subsequent particles to join the spiral are labeled $6$ and do not need to be verified because they can only attach in one place.

Our goal is to use the compression states to make sure the spiral forms with parent nodes indicating the steps in the spiral.  Verification occurs until a complete and verified circle is formed, at which point the spiral can expand sequentially without an opportunity for error.  We divide the progress of the spiral into stages, as follows: 


\begin{itemize}
\item \textbf{Stage 1: Inconsistent}: Arbitrary configuration, whether a particle is in a verified state indicates nothing about the existence of a complete circle.
\item \textbf{Stage 2: Before complete circle}: the configuration is consistent (Definition~\ref{defn:consistency}), but has not yet formed a complete circle.
\item \textbf{Stage 3: complete circle:} the configuration is consistent and has a complete circle.
\item \textbf{Stage 4: Verified complete circle:} the configuration is consistent, has a complete circle and all six positions of the complete circle are filled by verified particles.
\end{itemize}

\noindent We will now define the stages more precisely and explain how the algorithm progresses forward through these stages during a compression phase without ever progressing backwards through stages.

\subsection{Protocols for consistent spiral formation}
\label{apx:spiralprotocols}

A particle in a compression state  defines a {\it parent direction} pointing toward one of its six neighbors in the lattice.
We enumerate the possible directions with numbers $0,1,2,3,4,5$, with addition and subtraction defined modulo $6$. Particles do not have a global sense of direction, but they can  set direction with reference to the particles surrounding it, using a common chirality.
We define directions numerically only to simplify our notation for relative directions; for a direction $d$, direction $d+1$ would be clockwise, and $d-1$ would be counterclockwise.
We call the immediate neighbor of this particle in the direction $d$ the \emph{parent} of the particle.
The intention is for the path of parent pointers to traverse the spiral inwards, terminating at the food particle.

\vskip.1in
 \underbar{\bf The attachment property:}
 \ 
In order to forbid moves that violate our definitions, we define an attachment property that restricts  particles can to join the spiral only when they satisfy an {attachment property}, that we now specify.
For any particle $v$ and  direction $d$ and state~$s$, denote by $v_p$ the neighbor of $v$ in direction $d$, if it exists. The neighbors of $v$ clockwise and counterclockwise of $v_p$ are denoted $v_{cw}$ and $v_{ccw}$ respectively.
We denote by $d_p$, $d_{cw}$, $d_{ccw}$ their respective parent directions if they are in a compression state, and we denote by $s_p$, $s_{cw}$, $s_{ccw}$ their states, and these variables are null otherwise. 
Particle $v$ satisfies the {\it attachment property} with direction $d$ and state $s$ if $v_p$ exists and $v$ meets the following criteria, enumerating the necessary local conditions, typically based on the state of the parent particle:
\begin{itemize}
\item if $v_p$ is not in the dispersion state;

\item if $v_p$ is a food particle, then $s\in\{0,0^*\}$ or $v_{ccw}$ is a particle in a compression state and $\overline{s}_{ccw}=5$;

\item if $\overline{s_p} = 0$, then $\overline{s} = 1$, $d_p = d+2$, and $v_{cw}$ is a food particle;

\item if $\overline{s_p} \in \{1,2,3,4\}$, then $\overline{s} = \overline{s_p}+1$, $d_p = d+1$, and $v_{cw}$ is a food particle;

\item if $\overline{s_p} = 5$, then $s=6$, $d_p = d$, and $s_{cw} = 0^*$;

\item if $\overline{s_p} = 6$, then $s=6$ and $\overline{s}_{cw} \in \{0,1,2,3,4,5,6\}$, and if $\overline{s}_{cw} = 0$, then $d_p = d+1$ and $d_{cw} \in \{d_p, d_p+1\}$; or if $\overline{s}_{cw} \geq 1$, then $d_p \in \{d,d+1\}$, and $d_{cw} = d_p$.
\end{itemize}


\vskip.1in
\underbar{\bf Circle formation:} \ We are now prepared to describe how the first 6 particles join the spiral to form the initial circle around the food.
The verified states $0^*$ to $5^*$ are used to confirm that all six particles around the food particle are part of a single spiral, before allowing additional particles to join the spiral.
However, as the food particle may relocate at any time, a verified state may not always present the most up-to-date information of the current state of the spiral.
Figure~\ref{fig:losing_consistency} gives an example of this happening.

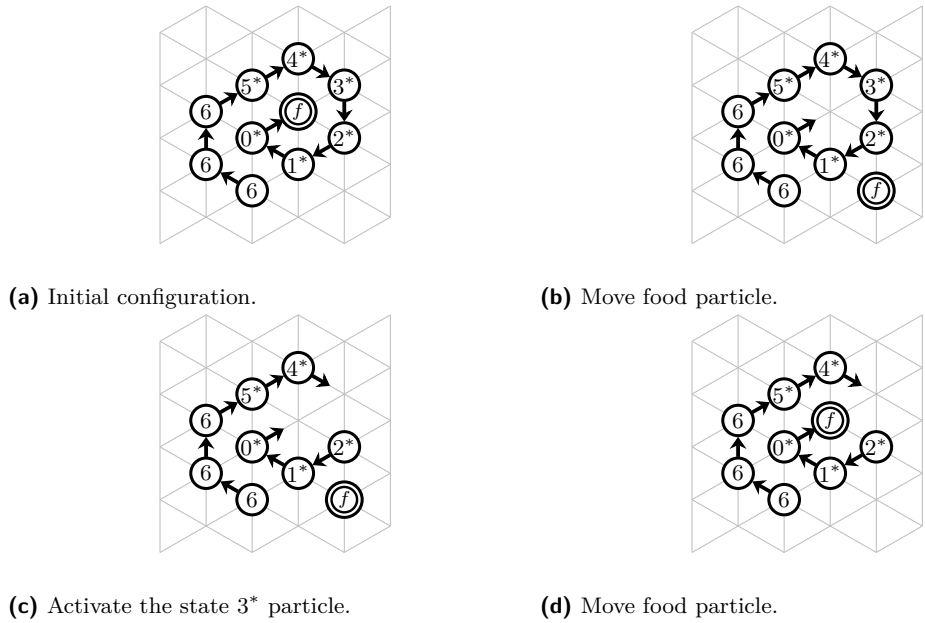
\begin{figure}
\begin{subfigure}[b]{.5\linewidth}
  \begin{center}
  \begin{tikzpicture}[x=0.7cm,y=0.7cm]
  \draw[lightgray] (2.59808,-0.5) -- (4.33013,-1.5);
\draw[lightgray] (2.59808,-0.5) -- (2.59808,-4.5);
\draw[lightgray] (0,-2) -- (2.59808,-0.5);
\draw[lightgray] (0,-2) -- (4.33013,-4.5);
\draw[lightgray] (1.73205,-1) -- (1.73205,-5);
\draw[lightgray] (1.73205,-5) -- (4.33013,-3.5);
\draw[lightgray] (0,-3) -- (4.33013,-0.5);
\draw[lightgray] (0,-3) -- (3.4641,-5);
\draw[lightgray] (4.33013,-0.5) -- (4.33013,-4.5);
\draw[lightgray] (0,-4) -- (4.33013,-1.5);
\draw[lightgray] (0,-4) -- (1.73205,-5);
\draw[lightgray] (3.4641,-1) -- (3.4641,-5);
\draw[lightgray] (3.4641,-5) -- (4.33013,-4.5);
\draw[lightgray] (0,-5) -- (4.33013,-2.5);
\draw[lightgray] (0.866025,-0.5) -- (4.33013,-2.5);
\draw[lightgray] (0.866025,-0.5) -- (0.866025,-4.5);
\draw[lightgray] (0,-1) -- (0.866025,-0.5);
\draw[lightgray] (0,-1) -- (4.33013,-3.5);
\draw[lightgray] (0,-1) -- (0,-5);
\draw[black, line width=0.4mm, fill=white] (0.866025,-2.5) circle (0.288);
\node[align=left] at (0.866025,-2.5) {\footnotesize $6$};
\draw[black,-{Stealth[length=1.6mm,width=2.2mm]},line width=0.5mm] (1.13449,-2.345) -- (1.46358,-2.155);
\draw[black, line width=0.4mm, fill=white] (0.866025,-3.5) circle (0.288);
\node[align=left] at (0.866025,-3.5) {\footnotesize $6$};
\draw[black,-{Stealth[length=1.6mm,width=2.2mm]},line width=0.5mm] (0.866025,-3.19) -- (0.866025,-2.81);
\draw[black, line width=0.4mm, fill=white] (1.73205,-2) circle (0.288);
\node[align=left] at (1.73205,-2) {\footnotesize $5^*$};
\draw[black,-{Stealth[length=1.6mm,width=2.2mm]},line width=0.5mm] (2.00052,-1.845) -- (2.32961,-1.655);
\draw[black, line width=0.4mm, fill=white] (1.73205,-3) circle (0.288);
\node[align=left] at (1.73205,-3) {\footnotesize $0^*$};
\draw[black,-{Stealth[length=1.6mm,width=2.2mm]},line width=0.5mm] (2.00052,-2.845) -- (2.32961,-2.655);
\draw[black, line width=0.4mm, fill=white] (1.73205,-4) circle (0.288);
\node[align=left] at (1.73205,-4) {\footnotesize $6$};
\draw[black,-{Stealth[length=1.6mm,width=2.2mm]},line width=0.5mm] (1.46358,-3.845) -- (1.13449,-3.655);
\draw[black, line width=0.4mm, fill=white] (2.59808,-1.5) circle (0.288);
\node[align=left] at (2.59808,-1.5) {\footnotesize $4^*$};
\draw[black,-{Stealth[length=1.6mm,width=2.2mm]},line width=0.5mm] (2.86654,-1.655) -- (3.19563,-1.845);
\draw[black, line width=0.4mm, fill=white] (2.59808,-2.5) circle (0.336);
\draw[black, line width=0.32mm] (2.59808,-2.5) circle (0.24);
\node[align=left] at (2.59808,-2.5) {\scriptsize $f$};
\draw[black, line width=0.4mm, fill=white] (2.59808,-3.5) circle (0.288);
\node[align=left] at (2.59808,-3.5) {\footnotesize $1^*$};
\draw[black,-{Stealth[length=1.6mm,width=2.2mm]},line width=0.5mm] (2.32961,-3.345) -- (2.00052,-3.155);
\draw[black, line width=0.4mm, fill=white] (3.4641,-2) circle (0.288);
\node[align=left] at (3.4641,-2) {\footnotesize $3^*$};
\draw[black,-{Stealth[length=1.6mm,width=2.2mm]},line width=0.5mm] (3.4641,-2.31) -- (3.4641,-2.69);
\draw[black, line width=0.4mm, fill=white] (3.4641,-3) circle (0.288);
\node[align=left] at (3.4641,-3) {\footnotesize $2^*$};
\draw[black,-{Stealth[length=1.6mm,width=2.2mm]},line width=0.5mm] (3.19563,-3.155) -- (2.86654,-3.345);
  \end{tikzpicture}
  \end{center}
  \caption{Initial configuration.}
\end{subfigure}%
\begin{subfigure}[b]{.5\linewidth}
  \begin{center}
  \begin{tikzpicture}[x=0.7cm,y=0.7cm]
  \draw[lightgray] (2.59808,-0.5) -- (4.33013,-1.5);
\draw[lightgray] (2.59808,-0.5) -- (2.59808,-4.5);
\draw[lightgray] (0,-2) -- (2.59808,-0.5);
\draw[lightgray] (0,-2) -- (4.33013,-4.5);
\draw[lightgray] (1.73205,-1) -- (1.73205,-5);
\draw[lightgray] (1.73205,-5) -- (4.33013,-3.5);
\draw[lightgray] (0,-3) -- (4.33013,-0.5);
\draw[lightgray] (0,-3) -- (3.4641,-5);
\draw[lightgray] (4.33013,-0.5) -- (4.33013,-4.5);
\draw[lightgray] (0,-4) -- (4.33013,-1.5);
\draw[lightgray] (0,-4) -- (1.73205,-5);
\draw[lightgray] (3.4641,-1) -- (3.4641,-5);
\draw[lightgray] (3.4641,-5) -- (4.33013,-4.5);
\draw[lightgray] (0,-5) -- (4.33013,-2.5);
\draw[lightgray] (0.866025,-0.5) -- (4.33013,-2.5);
\draw[lightgray] (0.866025,-0.5) -- (0.866025,-4.5);
\draw[lightgray] (0,-1) -- (0.866025,-0.5);
\draw[lightgray] (0,-1) -- (4.33013,-3.5);
\draw[lightgray] (0,-1) -- (0,-5);
\draw[black, line width=0.4mm, fill=white] (0.866025,-2.5) circle (0.288);
\node[align=left] at (0.866025,-2.5) {\footnotesize $6$};
\draw[black,-{Stealth[length=1.6mm,width=2.2mm]},line width=0.5mm] (1.13449,-2.345) -- (1.46358,-2.155);
\draw[black, line width=0.4mm, fill=white] (0.866025,-3.5) circle (0.288);
\node[align=left] at (0.866025,-3.5) {\footnotesize $6$};
\draw[black,-{Stealth[length=1.6mm,width=2.2mm]},line width=0.5mm] (0.866025,-3.19) -- (0.866025,-2.81);
\draw[black, line width=0.4mm, fill=white] (1.73205,-2) circle (0.288);
\node[align=left] at (1.73205,-2) {\footnotesize $5^*$};
\draw[black,-{Stealth[length=1.6mm,width=2.2mm]},line width=0.5mm] (2.00052,-1.845) -- (2.32961,-1.655);
\draw[black, line width=0.4mm, fill=white] (1.73205,-3) circle (0.288);
\node[align=left] at (1.73205,-3) {\footnotesize $0^*$};
\draw[black,-{Stealth[length=1.6mm,width=2.2mm]},line width=0.5mm] (2.00052,-2.845) -- (2.32961,-2.655);
\draw[black, line width=0.4mm, fill=white] (1.73205,-4) circle (0.288);
\node[align=left] at (1.73205,-4) {\footnotesize $6$};
\draw[black,-{Stealth[length=1.6mm,width=2.2mm]},line width=0.5mm] (1.46358,-3.845) -- (1.13449,-3.655);
\draw[black, line width=0.4mm, fill=white] (2.59808,-1.5) circle (0.288);
\node[align=left] at (2.59808,-1.5) {\footnotesize $4^*$};
\draw[black,-{Stealth[length=1.6mm,width=2.2mm]},line width=0.5mm] (2.86654,-1.655) -- (3.19563,-1.845);
\draw[black, line width=0.4mm, fill=white] (2.59808,-3.5) circle (0.288);
\node[align=left] at (2.59808,-3.5) {\footnotesize $1^*$};
\draw[black,-{Stealth[length=1.6mm,width=2.2mm]},line width=0.5mm] (2.32961,-3.345) -- (2.00052,-3.155);
\draw[black, line width=0.4mm, fill=white] (3.4641,-2) circle (0.288);
\node[align=left] at (3.4641,-2) {\footnotesize $3^*$};
\draw[black,-{Stealth[length=1.6mm,width=2.2mm]},line width=0.5mm] (3.4641,-2.31) -- (3.4641,-2.69);
\draw[black, line width=0.4mm, fill=white] (3.4641,-3) circle (0.288);
\node[align=left] at (3.4641,-3) {\footnotesize $2^*$};
\draw[black,-{Stealth[length=1.6mm,width=2.2mm]},line width=0.5mm] (3.19563,-3.155) -- (2.86654,-3.345);
\draw[black, line width=0.4mm, fill=white] (3.4641,-4) circle (0.336);
\draw[black, line width=0.32mm] (3.4641,-4) circle (0.24);
\node[align=left] at (3.4641,-4) {\scriptsize $f$};
  \end{tikzpicture}
  \end{center}
  \caption{Move food particle.}
\end{subfigure}
\begin{subfigure}[b]{.5\linewidth}
  \begin{center}
  \begin{tikzpicture}[x=0.7cm,y=0.7cm]
  \draw[lightgray] (2.59808,-0.5) -- (4.33013,-1.5);
\draw[lightgray] (2.59808,-0.5) -- (2.59808,-4.5);
\draw[lightgray] (0,-2) -- (2.59808,-0.5);
\draw[lightgray] (0,-2) -- (4.33013,-4.5);
\draw[lightgray] (1.73205,-1) -- (1.73205,-5);
\draw[lightgray] (1.73205,-5) -- (4.33013,-3.5);
\draw[lightgray] (0,-3) -- (4.33013,-0.5);
\draw[lightgray] (0,-3) -- (3.4641,-5);
\draw[lightgray] (4.33013,-0.5) -- (4.33013,-4.5);
\draw[lightgray] (0,-4) -- (4.33013,-1.5);
\draw[lightgray] (0,-4) -- (1.73205,-5);
\draw[lightgray] (3.4641,-1) -- (3.4641,-5);
\draw[lightgray] (3.4641,-5) -- (4.33013,-4.5);
\draw[lightgray] (0,-5) -- (4.33013,-2.5);
\draw[lightgray] (0.866025,-0.5) -- (4.33013,-2.5);
\draw[lightgray] (0.866025,-0.5) -- (0.866025,-4.5);
\draw[lightgray] (0,-1) -- (0.866025,-0.5);
\draw[lightgray] (0,-1) -- (4.33013,-3.5);
\draw[lightgray] (0,-1) -- (0,-5);
\draw[black, line width=0.4mm, fill=white] (0.866025,-2.5) circle (0.288);
\node[align=left] at (0.866025,-2.5) {\footnotesize $6$};
\draw[black,-{Stealth[length=1.6mm,width=2.2mm]},line width=0.5mm] (1.13449,-2.345) -- (1.46358,-2.155);
\draw[black, line width=0.4mm, fill=white] (0.866025,-3.5) circle (0.288);
\node[align=left] at (0.866025,-3.5) {\footnotesize $6$};
\draw[black,-{Stealth[length=1.6mm,width=2.2mm]},line width=0.5mm] (0.866025,-3.19) -- (0.866025,-2.81);
\draw[black, line width=0.4mm, fill=white] (1.73205,-2) circle (0.288);
\node[align=left] at (1.73205,-2) {\footnotesize $5^*$};
\draw[black,-{Stealth[length=1.6mm,width=2.2mm]},line width=0.5mm] (2.00052,-1.845) -- (2.32961,-1.655);
\draw[black, line width=0.4mm, fill=white] (1.73205,-3) circle (0.288);
\node[align=left] at (1.73205,-3) {\footnotesize $0^*$};
\draw[black,-{Stealth[length=1.6mm,width=2.2mm]},line width=0.5mm] (2.00052,-2.845) -- (2.32961,-2.655);
\draw[black, line width=0.4mm, fill=white] (1.73205,-4) circle (0.288);
\node[align=left] at (1.73205,-4) {\footnotesize $6$};
\draw[black,-{Stealth[length=1.6mm,width=2.2mm]},line width=0.5mm] (1.46358,-3.845) -- (1.13449,-3.655);
\draw[black, line width=0.4mm, fill=white] (2.59808,-1.5) circle (0.288);
\node[align=left] at (2.59808,-1.5) {\footnotesize $4^*$};
\draw[black,-{Stealth[length=1.6mm,width=2.2mm]},line width=0.5mm] (2.86654,-1.655) -- (3.19563,-1.845);
\draw[black, line width=0.4mm, fill=white] (2.59808,-3.5) circle (0.288);
\node[align=left] at (2.59808,-3.5) {\footnotesize $1^*$};
\draw[black,-{Stealth[length=1.6mm,width=2.2mm]},line width=0.5mm] (2.32961,-3.345) -- (2.00052,-3.155);
\draw[black, line width=0.4mm, fill=white] (3.4641,-3) circle (0.288);
\node[align=left] at (3.4641,-3) {\footnotesize $2^*$};
\draw[black,-{Stealth[length=1.6mm,width=2.2mm]},line width=0.5mm] (3.19563,-3.155) -- (2.86654,-3.345);
\draw[black, line width=0.4mm, fill=white] (3.4641,-4) circle (0.336);
\draw[black, line width=0.32mm] (3.4641,-4) circle (0.24);
\node[align=left] at (3.4641,-4) {\scriptsize $f$};
  \end{tikzpicture}
  \end{center}
  \caption{Activate the state $3^*$ particle.}
\end{subfigure}%
\begin{subfigure}[b]{.5\linewidth}
  \begin{center}
  \begin{tikzpicture}[x=0.7cm,y=0.7cm]
  \draw[lightgray] (2.59808,-0.5) -- (4.33013,-1.5);
\draw[lightgray] (2.59808,-0.5) -- (2.59808,-4.5);
\draw[lightgray] (0,-2) -- (2.59808,-0.5);
\draw[lightgray] (0,-2) -- (4.33013,-4.5);
\draw[lightgray] (1.73205,-1) -- (1.73205,-5);
\draw[lightgray] (1.73205,-5) -- (4.33013,-3.5);
\draw[lightgray] (0,-3) -- (4.33013,-0.5);
\draw[lightgray] (0,-3) -- (3.4641,-5);
\draw[lightgray] (4.33013,-0.5) -- (4.33013,-4.5);
\draw[lightgray] (0,-4) -- (4.33013,-1.5);
\draw[lightgray] (0,-4) -- (1.73205,-5);
\draw[lightgray] (3.4641,-1) -- (3.4641,-5);
\draw[lightgray] (3.4641,-5) -- (4.33013,-4.5);
\draw[lightgray] (0,-5) -- (4.33013,-2.5);
\draw[lightgray] (0.866025,-0.5) -- (4.33013,-2.5);
\draw[lightgray] (0.866025,-0.5) -- (0.866025,-4.5);
\draw[lightgray] (0,-1) -- (0.866025,-0.5);
\draw[lightgray] (0,-1) -- (4.33013,-3.5);
\draw[lightgray] (0,-1) -- (0,-5);
\draw[black, line width=0.4mm, fill=white] (0.866025,-2.5) circle (0.288);
\node[align=left] at (0.866025,-2.5) {\footnotesize $6$};
\draw[black,-{Stealth[length=1.6mm,width=2.2mm]},line width=0.5mm] (1.13449,-2.345) -- (1.46358,-2.155);
\draw[black, line width=0.4mm, fill=white] (0.866025,-3.5) circle (0.288);
\node[align=left] at (0.866025,-3.5) {\footnotesize $6$};
\draw[black,-{Stealth[length=1.6mm,width=2.2mm]},line width=0.5mm] (0.866025,-3.19) -- (0.866025,-2.81);
\draw[black, line width=0.4mm, fill=white] (1.73205,-2) circle (0.288);
\node[align=left] at (1.73205,-2) {\footnotesize $5^*$};
\draw[black,-{Stealth[length=1.6mm,width=2.2mm]},line width=0.5mm] (2.00052,-1.845) -- (2.32961,-1.655);
\draw[black, line width=0.4mm, fill=white] (1.73205,-3) circle (0.288);
\node[align=left] at (1.73205,-3) {\footnotesize $0^*$};
\draw[black,-{Stealth[length=1.6mm,width=2.2mm]},line width=0.5mm] (2.00052,-2.845) -- (2.32961,-2.655);
\draw[black, line width=0.4mm, fill=white] (1.73205,-4) circle (0.288);
\node[align=left] at (1.73205,-4) {\footnotesize $6$};
\draw[black,-{Stealth[length=1.6mm,width=2.2mm]},line width=0.5mm] (1.46358,-3.845) -- (1.13449,-3.655);
\draw[black, line width=0.4mm, fill=white] (2.59808,-1.5) circle (0.288);
\node[align=left] at (2.59808,-1.5) {\footnotesize $4^*$};
\draw[black,-{Stealth[length=1.6mm,width=2.2mm]},line width=0.5mm] (2.86654,-1.655) -- (3.19563,-1.845);
\draw[black, line width=0.4mm, fill=white] (2.59808,-2.5) circle (0.336);
\draw[black, line width=0.32mm] (2.59808,-2.5) circle (0.24);
\node[align=left] at (2.59808,-2.5) {\scriptsize $f$};
\draw[black, line width=0.4mm, fill=white] (2.59808,-3.5) circle (0.288);
\node[align=left] at (2.59808,-3.5) {\footnotesize $1^*$};
\draw[black,-{Stealth[length=1.6mm,width=2.2mm]},line width=0.5mm] (2.32961,-3.345) -- (2.00052,-3.155);
\draw[black, line width=0.4mm, fill=white] (3.4641,-3) circle (0.288);
\node[align=left] at (3.4641,-3) {\footnotesize $2^*$};
\draw[black,-{Stealth[length=1.6mm,width=2.2mm]},line width=0.5mm] (3.19563,-3.155) -- (2.86654,-3.345);
  \end{tikzpicture}
  \end{center}
  \caption{Move food particle.}
\end{subfigure}
\caption{A sequence of moves causing the verified states to give inaccurate information on the status of six particles surrounding the food. Particles in the dispersion state are not drawn. Note that the food particle can be moved at any time between steps of the algorithm.}
\label{fig:losing_consistency}
\end{figure}

To get around this, we define the notion of \emph{consistency}. Informally, the consistency of a configuration tells us if verified states are giving accurate information on the current state of the spiral. While a configuration may become inconsistent after a food particle relocates, if the food particle remains in place from then on, we will show that the configuration will eventually become consistent, and will remain consistent from then on.
We first make the following definitions:

\begin{definition}[Circles]
A \emph{circle} is defined by its center and one neighboring site, designated position $0$. 
Going counterclockwise around the center starting at 0 defines the positions $1,2,3,4,5$ associated with the circle. 
\end{definition}
%

\noindent Let $\sigma$ be any configuration of particles with states. For any circle, we say position $x \in \{0,1,\dots,5\}$ is {\it correctly filled} in the circle if position $x$ contains a particle in compression state $x$ or $x^*$, and its parent direction points towards the center, if $x=0$, and pointing towards position $x-1$, if $x > 0$. 

\begin{definition}[Inconsistency Value]
 Fix a configuration $\sigma$. If a circle in $\sigma$ has the food particle in the center, then the circle is assigned an {\it inconsistency value} $x+1$, where $x$ is the largest position in $\{0,1,2,3,4,5\}$ that contains a correctly filled particle in compression state $x^*,$ as long as at least one of $x+1,\dots,5$ is not correctly filled. If such a position does not exist, then the circle is assigned the value $0$. 
If the circle does not have the food particle in the center, the {\it inconsistency value} is the number of positions $y$ that are correctly filled with a particle in state $y^*$.
\end{definition}
\noindent The {\it inconsistency value of a configuration} is the sum of the inconsistency values over all possible circles on the lattice.

\begin{definition}[Consistency of a Configuration]
\label{defn:consistency}
We say a configuration is {\it consistent} if and only if its inconsistency value is $0$.
\end{definition}

In other words, a configuration is consistent if the only particles in verified states are adjacent to the food particle, 
and for any position adjacent to the food with a particle in state $x^*$, for $x \in \{0,1,2,3,4,5\}$, there is a circle centered at the food with positions $x, x+1, \dots 5$ all correctly filled.

\begin{definition}[Complete Circle]
We say a circle is a \emph{complete circle} if it has the food particle in the center and all six of its positions are correctly filled.
\end{definition}

There can only be one complete circle around a food particle at a time. While the food particle remains fixed, all six adjacent particles on the circle are stable, so they will never switch to the dispersion state.
Note also that in a consistent configuration, a particle in compression state $0^*$ only exists if it is part of a complete circle.

\vskip.1in
\underbar{\bf Achieving consistency:} \ 
We now show that from any starting configuration, assuming the food is constant, we will achieve consistency within an expected polynomial number of steps and the configuration will remain consistent until the food changes.

\begin{lemma}
\label{lem:inconsistencynonincreasing}
The activation of a particle can never increase the inconsistency value of any circle.
\end{lemma}

\begin{proof}
Fix some circle and suppose it has some inconsistency value $z$ in $\{0,1,2,3,4,5\}$.
If the circle has no food particle in the center, its inconsistency value can only increase if it gains a new correctly-filled particle in a verified state. This is however not possible without a food particle in the center, as this new particle would be unstable and hence would instead switch to the dispersion state on activation.

On the other hand, suppose that there is a food particle in the center of the circle. The only way the inconsistency value of the circle can increase is for the new configuration to be correctly filled with a verified particle in position $y > z-1$, but not all of the positions $y,y+1,\dots,5$ are correctly filled. Note that this rules out $y=5$. This gives us two possibilities. Either position $y$ was not previously correctly filled with a verified particle, or it was correctly filled with a verified particle but all of $y,y+1,\dots,5$ were correctly filled but now aren't.

For the first case, the last iteration must have changed the correctly filled particle in some position $y$ from state $y$ to $y^*$. As $y<5$, this can only happen when position $y+1$ contained a correctly filled particle in state $(y+1)^*$ in the previous configuration, which is a contradiction as that would imply the inconsistency value of the circle was at least $y+2$.

For the second case, this must mean some position $x > y$ was previously correctly filled but no longer is. However, this is not possible as with the food particle in the center, the particle in position $x$ was stable in the previous configuration.
\end{proof}

\begin{lemma}
\label{lem:inconsistencydecreases}
If the inconsistency value of a configuration is greater than $0$, the probability that the inconsistency value decreases by at least $1$ on the next iteration is at least ${1}/{n}$.
\end{lemma}

\begin{proof}
Consider any circle with some inconsistency value $z > 0$. It must have some position correctly filled by a particle in a verified state.
If the circle has no food particle in the center, activating this particle will switch it to the dispersion state, reducing the circle's value by $1$.

If the circle has the food particle in the center, the position $z-1$ must contain a correctly filled particle in the verified state $(z-1)^*$, but not all of $z,z+1,\dots,5$ are correctly filled. Thus $z-1$ cannot be $5$, and as $z-1$ is the largest such position, position $z$ will not contain a correctly filled particle in a verified state.
Thus, when the particle in position $z-1$ is activated, it will switch to either the unverified state $z-1$ or the dispersion state, reducing the inconsistency value of the circle by at least $1$. As no circle will increase in inconsistency value by Lemma~\ref{lem:inconsistencynonincreasing}, the inconsistency value of the configuration must strictly decrease.
\end{proof}

Out of the $6N$ possible circles we can define on the lattice, we note that as a configuration only has $n < N$ particles, at most $6\times 6\times n = 36n$ of these circles can contain a particle. The inconsistency value of any circle is at most $5$, so the inconsistency value of any configuration is at most $180n$.
Thus by Lemmas~\ref{lem:inconsistencynonincreasing} and \ref{lem:inconsistencydecreases}, after $O(n^2)$ expected steps, we will reach a consistent configuration, which remains consistent from then on.

The following lemma shows that we cannot progress backward through stages. Furthermore, we show in Appendix~\ref{apx:correctnessofspiralalgorithm} that the algorithm progresses through the stages efficiently.

\begin{lemma}
\label{lemma:stageregress}
Once we reach a stage $i \in \{2,3,4\}$, we cannot regress to any earlier stage.
\end{lemma}

\begin{proof}
We cannot regress to Stage 1 from any later stage, as by Lemma~\ref{lem:inconsistencynonincreasing}, a configuration once consistent remains consistent from then on.
We cannot return to Stage 2 after reaching Stage 3 or 4, as all particles in the six positions defining a complete circle will always be stable, and so can no longer change state or direction beyond switching between their verified and unverified states.
After reaching Stage 4, none of these six particles can return to an unverified state, thus we remain in Stage 4 from then on.
\end{proof}

\vskip.1in
\underbar{\bf Removing residual particles:} \ 
After achieving consistency, the next important step is the removal of residual particles (Definition~\ref{dfn:residualparticles}), similar to what was necessary for Adaptive $\alpha$-Compression.
Our objective  is to prove the following lemma:

\begin{lemma}
\label{lemma:removeresidual}
Suppose that either of these conditions currently hold, and continue to hold.
\begin{enumerate}
\item There is no particle in compression state $0^*$ in the current configuration.
\item The configuration is in Stage 4 (verified complete circle).
\end{enumerate}
Then we will reach a state with no residual particles
within $\frac{n^2}{1-2\rho} = O(n^2)$ steps.
\end{lemma}

\begin{definition}[Auxiliary Graph]
\label{dfn:auxiliarygraph}
For a given configuration, we define an auxiliary directed graph $H$ with particles as vertices. For each stable particle, we define an edge to its parent. For each attachable particle, we define an edge to each of its potential parents.
\end{definition}

It is important to note that edges of the graph $H$ do not always correspond to the parent directions of particles in the configuration. In particular, unstable particles may not have an outgoing edge in $H$ despite being in a compression state, and attachable particles will have an outgoing edge in $H$ despite not having a parent direction. The following Lemma gives us a crucial property of this auxiliary graph $H$:

\begin{lemma}
\label{lemma:indegree}
Other than the food particle, no vertex of $H$ has in-degree greater than $1$.
\end{lemma}

\begin{proof}
Fix any non-food vertex $u$. If there is an edge into $u$, $u$ must be in a compression state and thus has a parent in some direction $d$. Thus for any vertex $v$ with an edge into $u$, the direction of $u$ from $v$ must be either $d$ or $d-1$, limiting the possible locations of $v$ to two.
Let the particles at these locations be $v_1$ and $v_2$, where $u$ is at directions $d$ from $v_1$, and direction $d-1$ from $v_2$.
If there is an edge from $v_1$ to $u$ in $H$, then $v_1$ satisfies the attachment property in direction $d$. As $v_2$ is the particle counterclockwise of $u$ from the perspective of $v_1$, $v_2$ must either be a food particle or be in a compression state with direction either $d$ or $d+1$, meaning there is no edge from $v_2$ to $u$ in $H$.
Thus, $v_1$ and $v_2$ cannot simultaneously have edges towards $u$.
\end{proof}

This property of auxiliary graphs gives us the following lemma:

\begin{lemma}
\label{lemma:unstablevsattachable}
Suppose that we are in either of these two situations:
\begin{enumerate}
\item There is no particle in compression state $0^*$ in the current configuration.
\item The configuration is in Stage 4 (verified complete circle).
\end{enumerate}
Let $S_a$ denote the set of attachable particles in the dispersion state that can become a residual particle by changing to some compression state in some direction, and let $S_u$ denote the set of unstable residual particles that are not attachable.
Then $|S_a| \leq |S_u|+1$.
\end{lemma}

\begin{proof}
Consider an attachable particle $v$ in the dispersion state. If $v$ can become a residual particle in some compression state $s$ and direction $d$, it has an outward edge in the direction $d$ on the auxiliary graph $H$. We follow the a path out of $v$ over $H$, until we reach a particle $u$ with no outgoing edge.
Suppose $u$ is not the food particle. The predecessor to $u$ in this path was stable or attachable, so $u$ must be in some compression state. As $u$ has no outgoing edge, it must be unstable and not attachable.
Hence $u$ is not a spiral particle.
As $v$ would be a residual particle in compression state $s$ and direction $d$, if $v$ is adjacent to food, we must have $s=6$. If $v$ is not adjacent to food, it cannot be attachable unless $s=6$.

If $u$ is not state $6$, there must have been a particle earlier in the path which is attachable with a parent of state $5^*$. This indicates that there is a particle of state $0^*$ in the configuration ,and so the configuration must be in Stage 4 by the preconditions of this lemma. This means $u$ cannot be adjacent to food as all neighbors of the food particle are stable.
Hence, we can conclude that $u$ is also a residual particle.

As no vertex of $H$ has in-degree greater than 1, this gives an injection from $S_a$ to $S_u \cup \{v_f\}$, where $v_f$ is the food particle, which gives us the required result.
\end{proof}

\begin{lemma}
\label{lemma:newresidualparticle}
Suppose that we are in either of these two situations:
\begin{enumerate}
\item There is no particle in compression state $0^*$ in the current configuration.
\item The configuration is in Stage 4 (verified complete circle).
\end{enumerate}
A new residual particle can only come from an attachable particle being activated.
\end{lemma}

\begin{proof}
Suppose there exists a new residual particle $v$ that was not an attachable particle in the previous step. Thus, in the previous step, $v$ is a stable particle in a compression state. As $v$ was not residual, it was either a spiral particle, or is a particle in a state $s \leq 5$ adjacent to food.

If $v$ was in a state $s \leq 5$ and was adjacent to food, it cannot switch to state $6$ without being an attachable particle, and would remain adjacent to food as it cannot move. Hence it would not become a residual particle on the current step, which is a contradiction.

If $v$ was a spiral particle in the previous step, there must be a particle in compression state $0^*$, so the configuration must be in Stage 4. In Stage 4, $v$ could not have become a non-spiral particle in this step as every particle before it in the spiral is stable and verifiable, and hence could not have changed states.
\end{proof}

\begin{lemma}
\label{lemma:atleastoneunstableresidual}
Suppose that we are in either of these two situations:
\begin{enumerate}
\item There is no particle in compression state $0^*$ in the current configuration.
\item The configuration is in Stage 4 (verified complete circle).
\end{enumerate}
Then if there is a residual particle, then at least one of them is unstable.
\end{lemma}

\begin{proof}
Consider any residual particle.
If it is unstable, we are done. Particles with base states $0$ to $5$ cannot be stable unless adjacent to food, so a stable residual particle must be in state $6$. We start from this particle, and follow its edge out on the auxiliary graph $H$. If we continue following this path, there are multiple possibilities. The path could end on a particle of state $6$, a particle of base state $5$ or less, the food particle, or loop back to the particle we started with.

In the first case, the path ends on a particle in state $6$. This implies this last particle is unstable. This particle will also be a residual particle.

In the second case, the path ends on a particle not in state $6$ (either states $0$-$5$, $0^*$-$5^*$, or the food particle). Let $v$ be the last state $6$ particle in this path. If $v$ is unstable, we are done as $v$ is also a residual particle.
So we may assume $v$ is stable. For it to be stable, the particle clockwise of its parent direction must be in state $0^*$. By the preconditions of this Lemma, our configuration is thus in Stage 4, and as the parent of $v$ must be in state $5^*$, this path must end on the food particle.

Thus, consider the sequence of particles starting from the particle in state $0^*$ in the configuration, following the parent directions in reverse order. Each particle has in-degree at most $1$, so this sequence is well-defined. 
Denote by $u_1$ and $u_2$ the first non-spiral particle in this sequence and last spiral particle in the sequence respectively (both of these will exist, and $u_2$ is the parent of $u_1$). Denote by $u'$ the position where the next particle in the spiral after $u_2$ should have been.
We note that necessarily, $u_1$ is in state $6$ and $u_2$ is in states $5^*$ or~$6$. Let $d_1$, $d_2$ denote the parent directions of $u_1$, $u_2$ respectively.
If $u_1$ is unstable, we are done. We thus assume $u_1$ to be stable, meaning that either $d_1 = d_2$ or $d_1 = d_2-1$. As $u_1$ is not in the next position of the spiral after $u_2$, we must have $d_2 = d_1$, and the direction of $u_2$ from position $u'$ is $d_2-1$. Hence $u'$ is the position clockwise of $u_2$ from the perspective of $u_1$, and as $u_1$ is stable, $u'$ must be occupied by a particle in a compression state.
As we are in Stage 4, $u'$ cannot be in state $0^*$ or $u_1$ would be a spiral particle. If $u'$ is in state $0$, it cannot be adjacent to the food particle, and so is an unstable residual particle, which would mean we are done.
If $u'$ is in any other state, as $u_1$ is stable, $u'$ must have parent direction $d_2$. This is illustrated in Figure~\ref{fig:first_nonspiral}. This would make its parent a particle in the spiral other than $u_2$ (the sites marked by crosses in Figure~\ref{fig:first_nonspiral} must contain particles in the spiral), which necessarily implies that it is unstable as by Lemma~\ref{lemma:indegree}, two stable particles cannot have the same parent particle.

The third and final case is when the path forms a cycle. All particles in this path will be residual particles as they are not spiral particles, so if there are no unstable residual particles, these particles must all be stable and in state $6$.
We start from an arbitrary particle in the cycle and follow the sequence of parent particles one round around the cycle.
For each particle in this sequence, we consider the particle clockwise of its current direction (represented by $v_{cw}$ in the definition of the attachment property). After removing duplicates, these ``clockwise particles'' gives us another sequence $\{p0,...,p_{k-1}\}$ of particles.

Our claim is that the parent of each particle $p_i$ is the next particle $p_{i+1 \text{mod} k}$ in this sequence, modulo $k$. In other words, following the parent directions gives us another cycle.
To prove this claim, we first eliminate the possibility of any of these clockwise particles being of state $0^*$. We first note that at this stage, with one food particle, there can only be one particle of state $0^*$, which will be part of a complete circle of particles of states $0^*$ to $5^*$ going counterclockwise around the food particle. As a clockwise particle being $0^*$ forces $d_p = d+1$ for that particle, any path of particles where $0^*$ is a clockwise particle will inevitably wrap around the $0^*$ particle, which leads to the particle of state $5^*$ of the circle, implying that the path ends in the food particle, rather than forming a cycle.

Thus, as every particle in the cycle is stable, its clockwise particle must have the same direction as its parent.
Each particle $p_i$ (with direction $d_{p_i}$) is the clockwise particle of some particle $w_0$ (with direction $d_{w_0}$) in the original cycle. Let $w_1$ and $w_2$ (with directions $d_{w_1}$ and $d_{w_2}$ respectively) be the next two particles following the parent directions in the original cycle. 
As $w_0$ is stable, we have $d_{w_1} \in \{d_{w_0}, d_{w_0}+1\}$. If $d_{w_1} = d_{w_0}$, the clockwise particle of $w_1$ will be $p_{i+1}$ and the parent of $p_i$. If $d_{w_1} = d_{w_0}+1$, the clockwise particle of $w_1$ will also be $p_i$, which implies its parent $w_2$ must have direction $d_{w_2} = d_{w_1}$, which like the first case, would imply that the clockwise particle of $w_2$ will be $p_{i+1}$ and the parent of $p_i$. These two cases are illustrated in Figures~\ref{fig:cycle_in_cycle_c1} and~\ref{fig:cycle_in_cycle_c2} respectively.


\begin{figure}
\begin{minipage}[b]{.6\linewidth}
\begin{center}
\begin{subfigure}[b]{0.8\linewidth}
  \begin{center}
  \begin{tikzpicture}[x=0.6cm,y=0.6cm]
  \draw[lightgray] (0,-2) -- (2.59808,-0.5);
\draw[lightgray] (0,-2) -- (2.59808,-3.5);
\draw[lightgray] (0.866025,-0.5) -- (5.19615,-3);
\draw[lightgray] (0.866025,-0.5) -- (0.866025,-3.5);
\draw[lightgray] (0.866025,-3.5) -- (5.19615,-1);
\draw[lightgray] (2.59808,-0.5) -- (5.19615,-2);
\draw[lightgray] (2.59808,-0.5) -- (2.59808,-3.5);
\draw[lightgray] (2.59808,-3.5) -- (5.19615,-2);
\draw[lightgray] (4.33013,-0.5) -- (5.19615,-1);
\draw[lightgray] (4.33013,-0.5) -- (4.33013,-3.5);
\draw[lightgray] (4.33013,-3.5) -- (5.19615,-3);
\draw[lightgray] (0,-1) -- (0.866025,-0.5);
\draw[lightgray] (0,-1) -- (4.33013,-3.5);
\draw[lightgray] (0,-1) -- (0,-3);
\draw[lightgray] (1.73205,-1) -- (1.73205,-3);
\draw[lightgray] (5.19615,-1) -- (5.19615,-3);
\draw[lightgray] (3.4641,-1) -- (3.4641,-3);
\draw[lightgray] (0,-3) -- (4.33013,-0.5);
\draw[lightgray] (0,-3) -- (0.866025,-3.5);
\draw[black, line width=0.4mm, fill=white] (1.73205,-2) circle (0.288);
\draw[black,-{Stealth[length=1.6mm,width=2.2mm]},line width=0.5mm] (2.00052,-1.845) -- (2.32961,-1.655);
\node[align=left] at (1.18645,-1.6) {\footnotesize $u_1$};
\draw[black, line width=0.4mm, fill=white] (2.59808,-1.5) circle (0.288);
\draw[black,-{Stealth[length=1.6mm,width=2.2mm]},line width=0.5mm] (2.86654,-1.345) -- (3.19563,-1.155);
\node[align=left] at (2.05248,-1.1) {\footnotesize $u_2$};
\draw[black, line width=0.4mm, fill=white] (2.59808,-2.5) circle (0.288);
\draw[black,-{Stealth[length=1.6mm,width=2.2mm]},line width=0.5mm] (2.86654,-2.345) -- (3.19563,-2.155);
\node[align=left] at (3.16099,-2.8) {\footnotesize $u'$};
\draw[line width=0.4mm] (3.71866,-0.745442) -- (3.20954,-1.25456);
\draw[line width=0.4mm] (3.20954,-0.745442) -- (3.71866,-1.25456);
\draw[line width=0.4mm] (3.71866,-1.74544) -- (3.20954,-2.25456);
\draw[line width=0.4mm] (3.20954,-1.74544) -- (3.71866,-2.25456);
  \end{tikzpicture}
  \end{center}
  \caption{Illustration of the second case of the proof.}
  \label{fig:first_nonspiral}
\end{subfigure}\\~\\
\begin{subfigure}[b]{.45\linewidth}
  \begin{center}
  \begin{tikzpicture}[x=0.6cm,y=0.6cm]
  \draw[lightgray] (0,-2) -- (2.59808,-0.5);
\draw[lightgray] (0,-2) -- (3.4641,-4);
\draw[lightgray] (0.866025,-0.5) -- (5.19615,-3);
\draw[lightgray] (0.866025,-0.5) -- (0.866025,-3.5);
\draw[lightgray] (0.866025,-3.5) -- (5.19615,-1);
\draw[lightgray] (2.59808,-0.5) -- (5.19615,-2);
\draw[lightgray] (2.59808,-0.5) -- (2.59808,-3.5);
\draw[lightgray] (4.33013,-0.5) -- (5.19615,-1);
\draw[lightgray] (4.33013,-0.5) -- (4.33013,-3.5);
\draw[lightgray] (0,-1) -- (0.866025,-0.5);
\draw[lightgray] (0,-1) -- (4.33013,-3.5);
\draw[lightgray] (0,-1) -- (0,-3);
\draw[lightgray] (1.73205,-4) -- (5.19615,-2);
\draw[lightgray] (1.73205,-1) -- (1.73205,-4);
\draw[lightgray] (5.19615,-1) -- (5.19615,-3);
\draw[lightgray] (3.4641,-1) -- (3.4641,-4);
\draw[lightgray] (3.4641,-4) -- (5.19615,-3);
\draw[lightgray] (0,-3) -- (4.33013,-0.5);
\draw[lightgray] (0,-3) -- (1.73205,-4);
\draw[black, line width=0.4mm, fill=white] (1.73205,-2) circle (0.288);
\draw[black,-{Stealth[length=1.6mm,width=2.2mm]},line width=0.5mm] (2.00052,-1.845) -- (2.32961,-1.655);
\node[align=left] at (1.12583,-1.7) {\footnotesize $w_0$};
\draw[black, line width=0.4mm, fill=white] (2.59808,-1.5) circle (0.288);
\draw[black,-{Stealth[length=1.6mm,width=2.2mm]},line width=0.5mm] (2.86654,-1.345) -- (3.19563,-1.155);
\node[align=left] at (1.99186,-1.2) {\footnotesize $w_1$};
\draw[black, line width=0.4mm, fill=white] (2.59808,-2.5) circle (0.288);
\draw[black,-{Stealth[length=1.6mm,width=2.2mm]},line width=0.5mm] (2.86654,-2.345) -- (3.19563,-2.155);
\node[align=left] at (2.03516,-2.8) {\footnotesize $p_i$};
\draw[black, line width=0.4mm, fill=white] (3.4641,-2) circle (0.288);
\node[align=left] at (4.20022,-2.4) {\footnotesize $p_{i+1}$};
  \end{tikzpicture}
  \end{center}
  \caption{Third case of the proof, when $d_{w_1} = d_{w_0}$.}
  \label{fig:cycle_in_cycle_c1}
\end{subfigure}%
~
\begin{subfigure}[b]{.45\linewidth}
  \begin{center}
  \begin{tikzpicture}[x=0.6cm,y=0.6cm]
  \draw[lightgray] (0,-2) -- (2.59808,-0.5);
\draw[lightgray] (0,-2) -- (3.4641,-4);
\draw[lightgray] (0.866025,-0.5) -- (5.19615,-3);
\draw[lightgray] (0.866025,-0.5) -- (0.866025,-3.5);
\draw[lightgray] (0.866025,-3.5) -- (5.19615,-1);
\draw[lightgray] (2.59808,-0.5) -- (5.19615,-2);
\draw[lightgray] (2.59808,-0.5) -- (2.59808,-3.5);
\draw[lightgray] (4.33013,-0.5) -- (5.19615,-1);
\draw[lightgray] (4.33013,-0.5) -- (4.33013,-3.5);
\draw[lightgray] (0,-1) -- (0.866025,-0.5);
\draw[lightgray] (0,-1) -- (4.33013,-3.5);
\draw[lightgray] (0,-1) -- (0,-3);
\draw[lightgray] (1.73205,-4) -- (5.19615,-2);
\draw[lightgray] (1.73205,-1) -- (1.73205,-4);
\draw[lightgray] (5.19615,-1) -- (5.19615,-3);
\draw[lightgray] (3.4641,-1) -- (3.4641,-4);
\draw[lightgray] (3.4641,-4) -- (5.19615,-3);
\draw[lightgray] (0,-3) -- (4.33013,-0.5);
\draw[lightgray] (0,-3) -- (1.73205,-4);
\draw[black, line width=0.4mm, fill=white] (1.73205,-2) circle (0.288);
\draw[black,-{Stealth[length=1.6mm,width=2.2mm]},line width=0.5mm] (2.00052,-1.845) -- (2.32961,-1.655);
\node[align=left] at (1.12583,-1.7) {\footnotesize $w_0$};
\draw[black, line width=0.4mm, fill=white] (2.59808,-1.5) circle (0.288);
\draw[black,-{Stealth[length=1.6mm,width=2.2mm]},line width=0.5mm] (2.86654,-1.655) -- (3.19563,-1.845);
\node[align=left] at (3.20429,-1.2) {\footnotesize $w_1$};
\draw[black, line width=0.4mm, fill=white] (2.59808,-2.5) circle (0.288);
\draw[black,-{Stealth[length=1.6mm,width=2.2mm]},line width=0.5mm] (2.86654,-2.655) -- (3.19563,-2.845);
\node[align=left] at (2.03516,-2.8) {\footnotesize $p_i$};
\draw[black, line width=0.4mm, fill=white] (3.4641,-2) circle (0.288);
\draw[black,-{Stealth[length=1.6mm,width=2.2mm]},line width=0.5mm] (3.73257,-2.155) -- (4.06166,-2.345);
\node[align=left] at (4.07032,-1.7) {\footnotesize $w_2$};
\draw[black, line width=0.4mm, fill=white] (3.4641,-3) circle (0.288);
\node[align=left] at (4.20022,-3.4) {\footnotesize $p_{i+1}$};
  \end{tikzpicture}
  \end{center}
  \caption{Third case of the proof, when $d_{w_1} = d_{w_0}+1$.}
  \label{fig:cycle_in_cycle_c2}
\end{subfigure}
\end{center}
\end{minipage}%
~
\begin{subfigure}[b]{.35\linewidth}
  \begin{center}
  \begin{tikzpicture}[x=0.6cm,y=0.6cm]
  \draw[lightgray] (2.59808,-0.5) -- (6.9282,-3);
\draw[lightgray] (2.59808,-0.5) -- (2.59808,-7.5);
\draw[lightgray] (0,-2) -- (2.59808,-0.5);
\draw[lightgray] (0,-2) -- (6.9282,-6);
\draw[lightgray] (6.06218,-7.5) -- (6.9282,-7);
\draw[lightgray] (0,-7) -- (6.9282,-3);
\draw[lightgray] (0,-7) -- (0.866025,-7.5);
\draw[lightgray] (1.73205,-1) -- (1.73205,-7);
\draw[lightgray] (2.59808,-7.5) -- (6.9282,-5);
\draw[lightgray] (0,-3) -- (4.33013,-0.5);
\draw[lightgray] (0,-3) -- (6.9282,-7);
\draw[lightgray] (6.9282,-1) -- (6.9282,-7);
\draw[lightgray] (0,-4) -- (6.06218,-0.5);
\draw[lightgray] (0,-4) -- (6.06218,-7.5);
\draw[lightgray] (4.33013,-0.5) -- (6.9282,-2);
\draw[lightgray] (4.33013,-0.5) -- (4.33013,-7.5);
\draw[lightgray] (3.4641,-1) -- (3.4641,-7);
\draw[lightgray] (0,-5) -- (6.9282,-1);
\draw[lightgray] (0,-5) -- (4.33013,-7.5);
\draw[lightgray] (0.866025,-0.5) -- (6.9282,-4);
\draw[lightgray] (0.866025,-0.5) -- (0.866025,-7.5);
\draw[lightgray] (0,-1) -- (0.866025,-0.5);
\draw[lightgray] (0,-1) -- (6.9282,-5);
\draw[lightgray] (0,-1) -- (0,-7);
\draw[lightgray] (6.06218,-0.5) -- (6.9282,-1);
\draw[lightgray] (6.06218,-0.5) -- (6.06218,-7.5);
\draw[lightgray] (5.19615,-1) -- (5.19615,-7);
\draw[lightgray] (4.33013,-7.5) -- (6.9282,-6);
\draw[lightgray] (0,-6) -- (6.9282,-2);
\draw[lightgray] (0,-6) -- (2.59808,-7.5);
\draw[lightgray] (0.866025,-7.5) -- (6.9282,-4);
\draw[black, line width=0.4mm, fill=white] (0.866025,-3.5) circle (0.288);
\draw[black,-{Stealth[length=1.6mm,width=2.2mm]},line width=0.5mm] (1.13449,-3.345) -- (1.46358,-3.155);
\node[align=left] at (0.294449,-3.5) {\scriptsize $C_1$};
\draw[black, line width=0.4mm, fill=white] (0.866025,-4.5) circle (0.288);
\draw[black,-{Stealth[length=1.6mm,width=2.2mm]},line width=0.5mm] (0.866025,-4.19) -- (0.866025,-3.81);
\draw[black, line width=0.4mm, fill=white] (0.866025,-5.5) circle (0.288);
\draw[black,-{Stealth[length=1.6mm,width=2.2mm]},line width=0.5mm] (0.866025,-5.19) -- (0.866025,-4.81);
\draw[black, line width=0.4mm, fill=white] (1.73205,-3) circle (0.288);
\draw[black,-{Stealth[length=1.6mm,width=2.2mm]},line width=0.5mm] (2.00052,-2.845) -- (2.32961,-2.655);
\draw[fill=black] (1.73205,-4) circle (0.288);
\draw[black,-{Stealth[length=1.6mm,width=2.2mm]},line width=0.5mm] (2.00052,-3.845) -- (2.32961,-3.655);
\draw[fill=black] (1.73205,-5) circle (0.288);
\draw[black,-{Stealth[length=1.6mm,width=2.2mm]},line width=0.5mm] (1.73205,-4.69) -- (1.73205,-4.31);
\draw[black, line width=0.4mm, fill=white] (1.73205,-6) circle (0.288);
\draw[black,-{Stealth[length=1.6mm,width=2.2mm]},line width=0.5mm] (1.46358,-5.845) -- (1.13449,-5.655);
\draw[black, line width=0.4mm, fill=white] (2.59808,-2.5) circle (0.288);
\draw[black,-{Stealth[length=1.6mm,width=2.2mm]},line width=0.5mm] (2.86654,-2.345) -- (3.19563,-2.155);
\draw[fill=black] (2.59808,-3.5) circle (0.288);
\draw[black,-{Stealth[length=1.6mm,width=2.2mm]},line width=0.5mm] (2.86654,-3.345) -- (3.19563,-3.155);
\node[align=left] at (2.19797,-4.5) {\scriptsize $C_2$};
\draw[fill=black] (2.59808,-5.5) circle (0.288);
\draw[black,-{Stealth[length=1.6mm,width=2.2mm]},line width=0.5mm] (2.32961,-5.345) -- (2.00052,-5.155);
\draw[black, line width=0.4mm, fill=white] (2.59808,-6.5) circle (0.288);
\draw[black,-{Stealth[length=1.6mm,width=2.2mm]},line width=0.5mm] (2.32961,-6.345) -- (2.00052,-6.155);
\draw[black, line width=0.4mm, fill=white] (3.4641,-2) circle (0.288);
\draw[black,-{Stealth[length=1.6mm,width=2.2mm]},line width=0.5mm] (3.73257,-1.845) -- (4.06166,-1.655);
\draw[fill=black] (3.4641,-3) circle (0.288);
\draw[black,-{Stealth[length=1.6mm,width=2.2mm]},line width=0.5mm] (3.73257,-2.845) -- (4.06166,-2.655);
\draw[fill=black] (3.4641,-6) circle (0.288);
\draw[black,-{Stealth[length=1.6mm,width=2.2mm]},line width=0.5mm] (3.19563,-5.845) -- (2.86654,-5.655);
\draw[black, line width=0.4mm, fill=white] (3.4641,-7) circle (0.288);
\draw[black,-{Stealth[length=1.6mm,width=2.2mm]},line width=0.5mm] (3.19563,-6.845) -- (2.86654,-6.655);
\draw[black, line width=0.4mm, fill=white] (4.33013,-1.5) circle (0.288);
\draw[black,-{Stealth[length=1.6mm,width=2.2mm]},line width=0.5mm] (4.59859,-1.655) -- (4.92768,-1.845);
\draw[fill=black] (4.33013,-2.5) circle (0.288);
\draw[black,-{Stealth[length=1.6mm,width=2.2mm]},line width=0.5mm] (4.59859,-2.655) -- (4.92768,-2.845);
\draw[fill=black] (4.33013,-5.5) circle (0.288);
\draw[black,-{Stealth[length=1.6mm,width=2.2mm]},line width=0.5mm] (4.06166,-5.655) -- (3.73257,-5.845);
\draw[black, line width=0.4mm, fill=white] (4.33013,-6.5) circle (0.288);
\draw[black,-{Stealth[length=1.6mm,width=2.2mm]},line width=0.5mm] (4.06166,-6.655) -- (3.73257,-6.845);
\draw[black, line width=0.4mm, fill=white] (5.19615,-2) circle (0.288);
\draw[black,-{Stealth[length=1.6mm,width=2.2mm]},line width=0.5mm] (5.46462,-2.155) -- (5.79371,-2.345);
\draw[fill=black] (5.19615,-3) circle (0.288);
\draw[black,-{Stealth[length=1.6mm,width=2.2mm]},line width=0.5mm] (5.19615,-3.31) -- (5.19615,-3.69);
\draw[fill=black] (5.19615,-4) circle (0.288);
\draw[black,-{Stealth[length=1.6mm,width=2.2mm]},line width=0.5mm] (5.19615,-4.31) -- (5.19615,-4.69);
\draw[fill=black] (5.19615,-5) circle (0.288);
\draw[black,-{Stealth[length=1.6mm,width=2.2mm]},line width=0.5mm] (4.92768,-5.155) -- (4.59859,-5.345);
\draw[black, line width=0.4mm, fill=white] (5.19615,-6) circle (0.288);
\draw[black,-{Stealth[length=1.6mm,width=2.2mm]},line width=0.5mm] (4.92768,-6.155) -- (4.59859,-6.345);
\draw[black, line width=0.4mm, fill=white] (6.06218,-2.5) circle (0.288);
\draw[black,-{Stealth[length=1.6mm,width=2.2mm]},line width=0.5mm] (6.06218,-2.81) -- (6.06218,-3.19);
\draw[black, line width=0.4mm, fill=white] (6.06218,-3.5) circle (0.288);
\draw[black,-{Stealth[length=1.6mm,width=2.2mm]},line width=0.5mm] (6.06218,-3.81) -- (6.06218,-4.19);
\draw[black, line width=0.4mm, fill=white] (6.06218,-4.5) circle (0.288);
\draw[black,-{Stealth[length=1.6mm,width=2.2mm]},line width=0.5mm] (6.06218,-4.81) -- (6.06218,-5.19);
\draw[black, line width=0.4mm, fill=white] (6.06218,-5.5) circle (0.288);
\draw[black,-{Stealth[length=1.6mm,width=2.2mm]},line width=0.5mm] (5.79371,-5.655) -- (5.46462,-5.845);
  \end{tikzpicture}
  \end{center}
  \caption{In the third case of the proof, as every particle in cycle $C_1$ is stable, the set of ``clockwise particles'' of $C_1$ forms another cycle $C_2$.}
  \label{fig:cycle_in_cycle}
\end{subfigure}
\caption{Illustrations used by the proof of Lemma~\ref{lemma:atleastoneunstableresidual}.}
\end{figure}
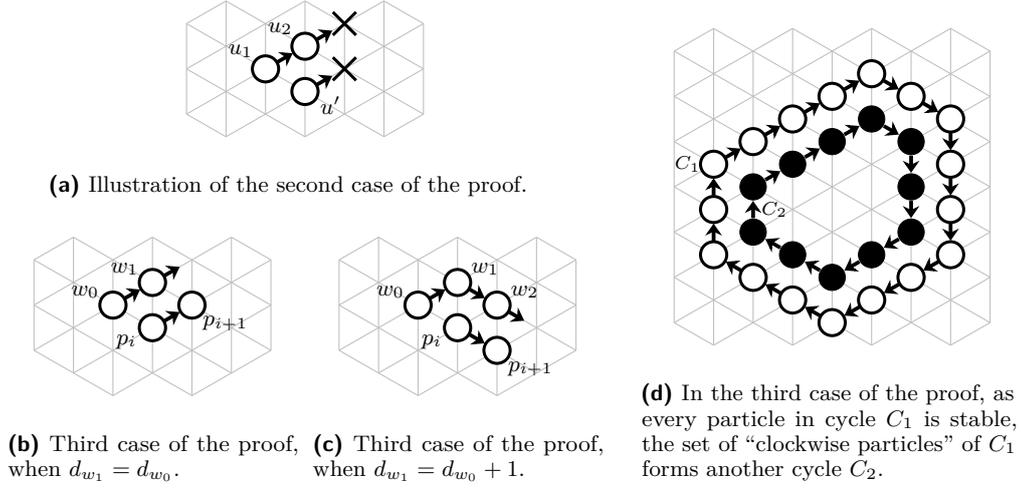

As illustrated in Figure~\ref{fig:cycle_in_cycle}, if every particle in this new cycle is also stable, we can repeat this process to generate yet another cycle. We repeat this to continue generating new cycles until we find one with an unstable particle.
Suppose that this never happens. If the first cycle happens to be a straight line (looping at the boundary of the torus), every subsequent cycle must also be a straight line, necessarily filling the entire lattice with straight lines.
This cannot happen as we do not consider cases where there are enough particles to fill every site of the lattice.
Thus, each of these cycles must have a ``bend'', which is a particle whose parent's direction $d_p$ is one step clockwise from its current direction (rather than being in the same direction as the current particle). This implies that each cycle created will have strictly fewer particles than the previous, creating a contradiction as we cannot do this indefinitely.
\end{proof}


These Lemmas allow us to prove Lemma~\ref{lemma:removeresidual}, which states that under certain conditions, all residual particles will be removed in polynomial expected time:

\begin{proof}[Proof of Lemma~\ref{lemma:removeresidual}]
We make use of Lemma~\ref{lemma:generallemmafinal} for this proof.
We first denote by $(\sigma_t)_{t \in \mathbbNzero}$ a random process where $\sigma_t$ represents the configuration $t$ steps after the current step. For a configuration $\sigma_t$, let $X(\sigma_t)$ denote the number of residual particles in $\sigma_t$.
On each step, by Lemma~\ref{lemma:newresidualparticle}, $X(\sigma_t)$ can only increase by at most $1$.
The probability $q(\sigma_t)$ of $X(\sigma_t)$ increasing by $1$ is at most $\frac{\rho|S_a|}{n}$, and the probability $p(\sigma_t)$ of $X(\sigma_t)$ decreasing by at least $1$ is at least $\frac{|S_u|}{n}$.
By Lemma~\ref{lemma:atleastoneunstableresidual}, we have $p(\sigma_t) \geq \frac{1}{n}$, and by Lemma~\ref{lemma:unstablevsattachable}, we have $q(\sigma_t) \leq \frac{\rho|S_a|}{n} \leq \frac{\rho(|S_u|+1)}{n} \leq 2\rho\frac{|S_u|}{n} \leq 2\rho p(\sigma_t)$. As $\rho < \frac{1}{2}$ in Algorithm~\ref{alg:spiralalgorithm}, by Lemma~\ref{lemma:generallemmafinal} the expected amount of steps before $X(\sigma_t)$ first reaches $0$ is at most $\frac{nX(\sigma_0)}{1-2\rho} \leq \frac{n^2}{1-2\rho}$.
\end{proof}

\subsection{Correctness of adaptive spiraling}
\label{apx:correctnessofspiralalgorithm}
We are now prepared to show the main theorems (Theorems~\ref{theorem:spiralwithfood} and \ref{theorem:spiralnofood}) to verify the correctness of the Adaptive Spiraling algorithm. We first state and prove two final lemmas.

\begin{lemma}[Stage 4 Lemma]
\label{lem:stage4lemma}
If there are no residual particles and we are currently in Stage 4 (verified complete circle) the following will be true in all subsequent iterations:
\begin{enumerate}
\item There will be no residual particles.
\item All compression state particles are spiral particles.
\item The number of spiral particles can never decrease.
\item Consider the set of possible locations not containing a particle in a compression state. There is exactly one location in this set (the ``end'' of the spiral) where a particle in the dispersion state can be attachable. A particle switching from the dispersion state to a compression state in this location increases the number of spiral particles by $1$.
\end{enumerate}
\end{lemma}

\begin{proof}
Suppose that there are no residual particles in a configuration. As all particles neighboring the food particle are spiral particles, this implies that all compression state particles are spiral particles (2).

In Stage 4, all spiral particles necessarily belong to the same spiral starting from the food particle, and it can be easily checked that every particle in this spiral satisfies the conditions required for stability. Hence, all spiral particles are stable.

This implies that as long as there are no residual particles in the current iteration, no particles will switch from a compression state to the dispersion state in the next iteration, so the spiral remains intact (3).
Furthermore, with a single spiral, there will be exactly one location where a particle in the dispersion state can switch to a compression state, to increase the number of spiral particles by $1$ (4).
Thus, there will also be no residual particles in the next iteration, and consequently from then on (1). This gives us statements (2), (3) and (4) in all subsequent iterations as well.
\end{proof}

\begin{lemma}
\label{lemma:nonewresidual}
Suppose that there is no particle in compression state $0^*$ in the current configuration.
Then we cannot gain a new residual particle in the next iteration if none currently exist.
\end{lemma}

\begin{proof}
We show that the configuration in the next iteration will have no residual particle, given that the current configuration has no residual particle.

If there is no particle in state $0^*$, no particle can be a spiral particle, so the only particles in a compression state must be next to a food particle, and cannot be in state $6$.
As long as there is no particle of state $0^*$ in the current configuration, we cannot add a new compression state particle in state $6$. We also cannot add a new compression state particle in states other than $6$ in locations not neighboring the food particle. Thus, it is not possible to gain a new residual particle.

\end{proof}

\begin{lemma}
\label{lemma:stageprogress}
It takes a polynomial number of steps in expectation to progress to the next stage. In particular, a loose upper bound for the expected number of steps to reach Stage 4 from an arbitrary configuration is $O(N^3 n)$.
\end{lemma}

\begin{proof}
As seen in Appendix~\ref{apx:spiralprotocols}
, a configuration in Stage 1 will reach at least Stage 2 by achieving consistency in $O(n^2)$ steps.

In Stage 2, there will be no particle in compression state $0^*$. We can only obtain a particle in compression state $0^*$ after we reach Stage 3. Thus by Lemma~\ref{lemma:removeresidual}, we will reach a state with no residual particles within $O(n^2)$ steps, and by Lemma~\ref{lemma:nonewresidual}, we will never gain any new residual particles as long as we remain in Stage 2. This means that the only particles in compression states will be adjacent to food.
We then note that when a particle switches to some compression state while adjacent to food, it will always either be stable or attachable, and thus will never leave its position. Assuming at least six particles in the system, similar to the analysis in the proof of Lemma~\ref{lemma:allswitchtocompressionstate} with a loose $O(N^3)$~\cite{lovasz1993random} bound on hitting times of simple random walks on graphs of $N$ vertices, we can bound the amount of steps to fill all six positions in the neighborhood of the food particle by $O(6 \cdot N^3 n) = O(N^3 n)$.
Once this neighborhood is filled, we note that activating each of the six particles once in counterclockwise order will give us a complete circle. This happens after $O(n)$ steps in expectation, and puts us in Stage 3.

In Stage 3, if the particles in the neighborhood of the food particle corresponding to the states $5,4,3,2,1,0$ are activated in sequence, all of these particles will switch to the states $5^*,4^*,3^*,2^*,1^*,0^*$ respectively, bringing us to Stage 4. It takes $O(n)$ steps for this to happen in expectation.

Combining these bounds and noting that we cannot regress to earlier stages (Lemma~\ref{lemma:stageregress}) gives us a polynomial upper bound of $O(N^3 n + n^2) = O(N^3 n)$ for reaching Stage 4 when starting from an arbitrary configuration.
\end{proof}


We are now prepared to prove the two main theorems.

\begin{proof}[Proof of Theorem~\ref{theorem:spiralwithfood}]
By Lemmas~\ref{lemma:stageregress} and \ref{lemma:stageprogress}, we will reach Stage 4 within $O(N^3 n)$ steps in expectation, and remain in Stage 4 from then on.
By Lemma~\ref{lemma:removeresidual} within $O(n^2)$ steps in expectation, we will reach a configuration with no residual particles.

By Lemma~\ref{lem:stage4lemma}, all compression state particles from then on will be spiral particles, which means they will necessarily all be part of a single spiral.
Furthermore, there is always exactly one position wbere a particle can join the spiral by switching to some compression state. We assume a reasonable amount of extra space for particles in the dispersion state to move around (see Figure~\ref{fig:nospace} for when this is not true).

The particles in the dispersion state follow random walks, and once again with an analysis similar to that in the proof of Lemma~\ref{lemma:allswitchtocompressionstate} but with $O(N^3)$~\cite{lovasz1993random} hitting times of simple random walks,
and taking into account that a new particle only joins the spiral with some constant probability $\rho \in (0,\frac{1}{2})$, a loose upper bound for the expected time for the next particle to join the spiral is $O(\frac{1}{\rho} N^3 n) = O(N^3 n)$.
Thus with $n$ particles in the configuration, all particles will be part of the spiral after $O(N^3 n^2)$ steps in expectation. Combining this with the time it takes to reach Stage 4 and remove all residual particles gives us a loose upper bound of $O(N^3n^2 + N^3n + n^2) = O(N^3 n^2)$ steps for all particles to gather.
\end{proof}


\begin{proof}[Proof of Theorem~\ref{theorem:spiralnofood}]
If there is no food particle, all compression state particles will be residual particles. Any particle in state $0^*$ will be unstable and non-attachable, and thus will switch to the dispersion state on activation. No new particles of state $0^*$ can form either. In the worst case over all possible starting configurations, by the analysis of the coupon collector's problem, there will be no more particles in state $0^*$ after $O(n\log n)$ iterations in expectation.

After that, by Lemma~\ref{lemma:removeresidual}, we will reach a state with no residual particles after $O(n^2)$ steps in expectation. This means all particles will be in the dispersion state. Furthermore, with no food particle or particles in compression states, no particle will be able to switch to a compression state after this.
\end{proof}





\newcommand{\spine}{{spine}}
\newcommand{\Spine}{{Spine}}
\newcommand{\spines}{{spines}}
\newcommand{\SpineComb}{{Spine Comb}}
\newcommand{\spinecomb}{{spine comb}}
\newcommand{\spinecombs}{{spine combs}}
\newcommand{\mainspine}{{main spine}}
\newcommand{\anchorparticle}{{anchor particle}}
\newcommand{\anchorparticles}{{anchor particles}}
\newcommand{\tailparticle}{{tail particle}}
\newcommand{\tailparticles}{{tail particles}}
\newcommand{\sourcespine}{{source \spine{}}}
\newcommand{\targetspine}{{target \spine{}}}
\newcommand{\ResidualRegion}{{Residual Region}}
\newcommand{\residualregion}{{residual region}}


%
%
%
%


\section{Ergodicity of $\mathcal{M}_{\rm COM}$}
\label{apx:irreducibilityproof}

We conclude by proving Lemma~\ref{lemma:irreducible}, showing that the Markov chain $\mathcal{M}_{\rm COM}$ is irreducible (and thus ergodic) by showing that all states are reachable using compression moves in the presence of food (as in Section~\ref{section:statemovement}).  The proof of ergodicity given in Cannon et al. \cite{Cannon2016} without food particles was already fairly involved, but with the addition food particles that cannot themselves move, the proof becomes substantially more challenging.  We overcome this by introducing a ``comb'' operation that organizes particles radially around the food.

We will treat the immobile food particle as the ``center'' of the configuration. From the center, there are six directions one can move in a straight line on the triangular lattice - up, down, up-left, down-left, up-right, down-right. We call these six straight lines of particles extending from the food particle \emph{\spines{}}. We refer to particles on the \spines{} as \spine{} particles, and particles not on \spines{} as non-\spine{} particles. We similarly use the names \spine{} and non-\spine{} locations to refer to sites on the triangular lattices.

On each \spine{}, we call the furthest out (in terms of distance from the food particle) particle on the \spine{} with no adjacent non-\spine{} particles the \emph{\anchorparticle{}} of the \spine{}. The \spine{} particles further out than the \anchorparticle{} are called \emph{\tailparticles{}}. The \emph{distance} of a \spine{} location from the center refers to its shortest-path distance (which would be along the \spine{}) to the food particle on the triangular lattice. For each integer $r \geq 1$, the hexagon of radius $r$ refers to the regular hexagon with corners defined by the six \spine{} locations of distance $r$ from the center. The distance of a non-\spine{} location from the center would then be the radius of the smallest such hexagon it is contained within.
An important concept that we will use in the proof is the length of a \spine{}. The length of a \spine{} is defined to be the distance of its \anchorparticle{} to the center. If it has no \anchorparticle{}, the length of the \spine{} is $0$.

\begin{figure*}[t]
\begin{center}
\begin{tikzpicture}[x=0.6cm,y=0.6cm]
\input{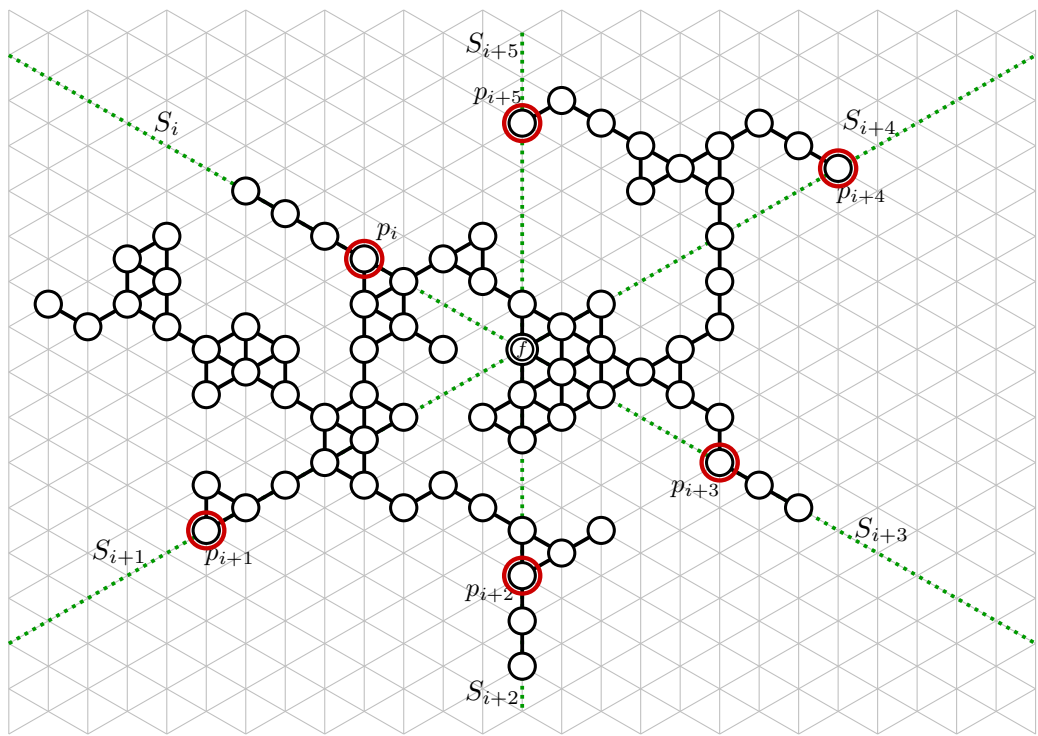}
\end{tikzpicture}
\end{center}
\caption{Illustration of the \spines{} extending from the food particle $f$.
The six \spines{} $S_i,\ldots,S_{i+5}$ have lengths $4,8,5,5,8,5$ respectively. These \spines{} have respective \anchorparticles{} $p_i,\ldots,p_{i+5}$. As an illustration of the coordinate system, these six anchor particles are at coordinates $(4,0)$, $(8,8)$, $(0,5)$, $(-5,0)$, $(-8,-8)$ and $(0,-5)$ respectively.}
\label{fig:spines}
\end{figure*}

We notate the six \spines{} using one of the \spines{} as a reference \spine{}. If the reference \spine{} is denoted $S_i$, where $i$ is an integer modulo $6$, then $S_{i+1}, S_{i+2}, \dots, S_{i+5}$ denote the subsequent \spines{} in a counterclockwise order from $S_i$.

The proof centers around a specific transformation we call a ``comb'' operation. This comb operation is applied from one \spine{} (which we refer to as the \sourcespine{}) to an adjacent \spine{} (which we refer to as the \targetspine{}), and has the effect of ``pushing'' the particles between the two \spines{} towards the \targetspine{}.

Our system exhibits reflection symmetry and 6-fold rotational symmetry, so this comb operation can be defined in $6\times 2 = 12$ different ways. However, for simplicity of discussion, we will only define the comb operation in one orientation, specifically on the left side, downwards. This is a comb from the \spine{} going in the up-left direction to the \spine{} going into the down-left direction. We rotate or reflect the configuration freely, depending on which pair of adjacent \spines{} we want to comb between. 

\subsection{The comb operation}
We define our two-dimensional coordinate system $(lane,depth)$ with reference to the \sourcespine{}, assumed to be going in the up-left direction. A position $(\ell,0)$ for $\ell \geq 0$ refers to the position on the \sourcespine{} $\ell$ steps away from the food particle. If $\ell < 0$, this refers to the position $-\ell$ steps in the direction of the \spine{} directly opposite the \sourcespine{}. A position $(\ell,d)$ refers to the location $d$ steps downwards from position $(\ell,0)$. Thus, $d$ denotes the (signed) distance of the position from the \sourcespine{}.

\begin{figure*}
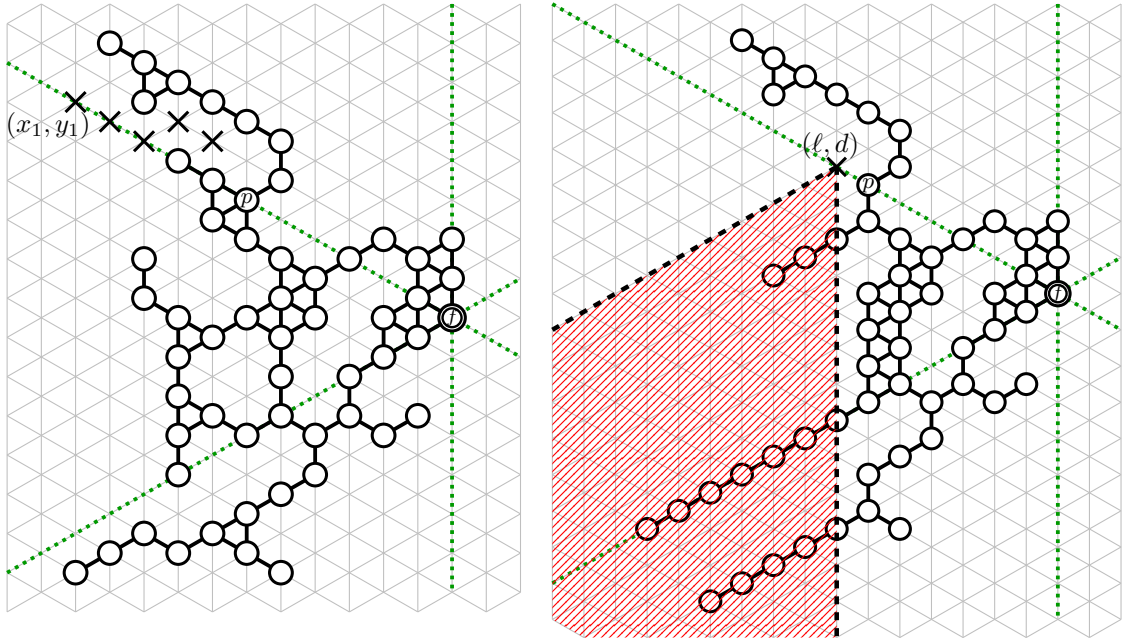

\begin{subfigure}[t]{.48\textwidth}
  \begin{center}
  \begin{tikzpicture}[x=0.52cm,y=0.52cm]
  \input{diagrams_tex/spine_comb.tex}
  \end{tikzpicture}
  \end{center}
  \caption{Illustration of a \spinecomb{} (Definition~\ref{proc:spinecomb}), which in this case is a comb over the sequence (Definition~\ref{proc:combingasequence}) denoted by crosses in the Figure.}
  \label{fig:spine_comb}
\end{subfigure}%
\hfill
\begin{subfigure}[t]{.48\textwidth}
  \begin{center}
  \begin{tikzpicture}[x=0.48cm,y=0.48cm]
  \input{diagrams_tex/combed_residual_region.tex}
  \end{tikzpicture}
  \end{center}
  \caption{After combing $(\ell,d)$, position $(\ell,d)$ is combed (Definition~\ref{defn:combed}). The shaded region is the \residualregion{} of $(\ell,d)$ (Definition~\ref{defn:residualregion})}
  \label{fig:combed_residual_region}
\end{subfigure}%
\caption{The comb operation is applied in to the five points marked with crosses in Figure~\ref{fig:spine_comb} from left to right in sequence, starting from $(x_1,y_1)$. Figure~\ref{fig:combed_residual_region} illustrates the end result.}
\label{fig:comb_before_after}
\end{figure*}

Before we define the comb operation, the following definitions tells us what can and cannot be combed.

\begin{definition}[\ResidualRegion{}]
\label{defn:residualregion}
Consider a position $(\ell,d)$ and the diagonal half-line extending down-left from $(\ell,d)$, including $(\ell,d)$ itself. The \emph{\residualregion{}} of this position refers to the set of all positions on or below this half-line (Figure~\ref{fig:combed_residual_region}).
\end{definition}

\begin{definition}[Combed]
\label{defn:combed}
For $\ell > 0$ and $d \geq 0$, we say a position $(\ell,d)$ is \emph{combed} (Figure~\ref{fig:combed_residual_region}) if:
\begin{enumerate}
\item All sites directly above a topmost particle of the \residualregion{} of $(\ell,d)$ are empty.
\item All particles in the \residualregion{} of $(\ell,d)$ form straight lines stretching down and left.
\item Consider the column of sites directly to the right of the \residualregion{} of $(\ell,d)$. Each of the abovementioned lines of particles stretches down and left from a particle from this column with no particle directly below.
\end{enumerate}
\end{definition}

\begin{definition}[Combable]
\label{defn:combable}
For $\ell > 0$ and $d \geq 0$, we say a position $(\ell,d)$ is \emph{combable} if:
\begin{enumerate}
\item The position $(\ell+1,d+1)$, which is one step diagonally down-left from $(\ell,d)$, is combed.
\item The site directly above $(\ell,d)$ is empty.
\end{enumerate}
\end{definition}

With this, we can define the comb procedure. 
Aside from operating only below a given depth, this comb procedure is identical to the process used the proof of irreducibility in \cite{Cannon2016}. As this operation is covered in detail in said paper, we will be brief with its explanation here.

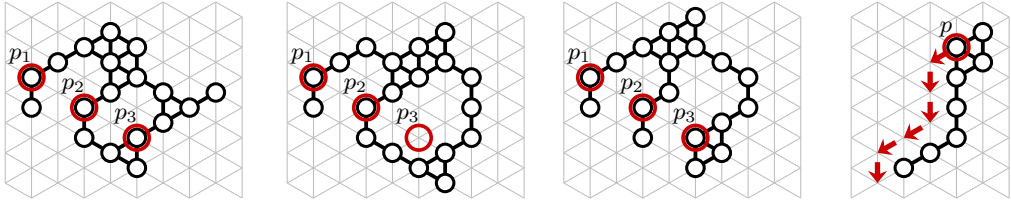
\begin{figure*}
  \begin{subfigure}{.26\textwidth}
    \centering
    \begin{center}
    \begin{tikzpicture}[x=0.4cm,y=0.4cm]
    \draw[lightgray] (2.59808,-0.5) -- (2.59808,-6.5);
\draw[lightgray] (0,-2) -- (3.4641,0);
\draw[lightgray] (0,-2) -- (7.79423,-6.5);
\draw[lightgray] (4.33013,-6.5) -- (7.79423,-4.5);
\draw[lightgray] (6.9282,0) -- (7.79423,-0.5);
\draw[lightgray] (6.9282,0) -- (6.9282,-6);
\draw[lightgray] (0.866025,-6.5) -- (7.79423,-2.5);
\draw[lightgray] (0,-3) -- (5.19615,0);
\draw[lightgray] (0,-3) -- (6.06218,-6.5);
\draw[lightgray] (3.4641,0) -- (7.79423,-2.5);
\draw[lightgray] (3.4641,0) -- (3.4641,-6);
\draw[lightgray] (7.79423,-0.5) -- (7.79423,-6.5);
\draw[lightgray] (6.06218,-6.5) -- (7.79423,-5.5);
\draw[lightgray] (2.59808,-6.5) -- (7.79423,-3.5);
\draw[lightgray] (0,-4) -- (6.9282,0);
\draw[lightgray] (0,-4) -- (4.33013,-6.5);
\draw[lightgray] (4.33013,-0.5) -- (4.33013,-6.5);
\draw[lightgray] (0,0) -- (7.79423,-4.5);
\draw[lightgray] (0,0) -- (0,-6);
\draw[lightgray] (5.19615,0) -- (7.79423,-1.5);
\draw[lightgray] (5.19615,0) -- (5.19615,-6);
\draw[lightgray] (0,-5) -- (7.79423,-0.5);
\draw[lightgray] (0,-5) -- (2.59808,-6.5);
\draw[lightgray] (0.866025,-0.5) -- (0.866025,-6.5);
\draw[lightgray] (0,-1) -- (1.73205,0);
\draw[lightgray] (0,-1) -- (7.79423,-5.5);
\draw[lightgray] (6.06218,-0.5) -- (6.06218,-6.5);
\draw[lightgray] (0,-6) -- (7.79423,-1.5);
\draw[lightgray] (0,-6) -- (0.866025,-6.5);
\draw[lightgray] (1.73205,0) -- (7.79423,-3.5);
\draw[lightgray] (1.73205,0) -- (1.73205,-6);
\draw[black, line width=0.4mm, fill=white] (0.866025,-2.5) circle (0.288);
\node[align=left] at (0.519615,-1.8) {\small $p_1$};
\draw[black, line width=0.4mm, fill=white] (0.866025,-3.5) circle (0.288);
\draw[black, line width=0.4mm, fill=white] (1.73205,-2) circle (0.288);
\draw[black, line width=0.4mm, fill=white] (2.59808,-1.5) circle (0.288);
\draw[black, line width=0.4mm, fill=white] (2.59808,-3.5) circle (0.288);
\node[align=left] at (2.25167,-2.8) {\small $p_2$};
\draw[black, line width=0.4mm, fill=white] (2.59808,-4.5) circle (0.288);
\draw[black, line width=0.4mm, fill=white] (3.4641,-1) circle (0.288);
\draw[black, line width=0.4mm, fill=white] (3.4641,-2) circle (0.288);
\draw[black, line width=0.4mm, fill=white] (3.4641,-3) circle (0.288);
\draw[black, line width=0.4mm, fill=white] (3.4641,-5) circle (0.288);
\draw[black, line width=0.4mm, fill=white] (4.33013,-1.5) circle (0.288);
\draw[black, line width=0.4mm, fill=white] (4.33013,-2.5) circle (0.288);
\draw[black, line width=0.4mm, fill=white] (4.33013,-4.5) circle (0.288);
\node[align=left] at (3.98372,-3.8) {\small $p_3$};
\draw[black, line width=0.4mm, fill=white] (4.33013,-5.5) circle (0.288);
\draw[black, line width=0.4mm, fill=white] (5.19615,-3) circle (0.288);
\draw[black, line width=0.4mm, fill=white] (5.19615,-4) circle (0.288);
\draw[black, line width=0.4mm, fill=white] (6.06218,-3.5) circle (0.288);
\draw[black, line width=0.4mm, fill=white] (6.9282,-3) circle (0.288);
\draw[black, line width=0.5mm] (1.12583,-2.35) -- (1.47224,-2.15);
\draw[black, line width=0.5mm] (5.19615,-3.7) -- (5.19615,-3.3);
\draw[black, line width=0.5mm] (5.45596,-3.85) -- (5.80237,-3.65);
\draw[black, line width=0.5mm] (3.72391,-1.15) -- (4.07032,-1.35);
\draw[black, line width=0.5mm] (0.866025,-3.2) -- (0.866025,-2.8);
\draw[black, line width=0.5mm] (2.85788,-3.35) -- (3.20429,-3.15);
\draw[black, line width=0.5mm] (4.58993,-4.35) -- (4.93634,-4.15);
\draw[black, line width=0.5mm] (4.33013,-5.2) -- (4.33013,-4.8);
\draw[black, line width=0.5mm] (3.72391,-4.85) -- (4.07032,-4.65);
\draw[black, line width=0.5mm] (3.72391,-5.15) -- (4.07032,-5.35);
\draw[black, line width=0.5mm] (6.32199,-3.35) -- (6.6684,-3.15);
\draw[black, line width=0.5mm] (2.85788,-1.35) -- (3.20429,-1.15);
\draw[black, line width=0.5mm] (2.85788,-1.65) -- (3.20429,-1.85);
\draw[black, line width=0.5mm] (4.33013,-2.2) -- (4.33013,-1.8);
\draw[black, line width=0.5mm] (4.58993,-2.65) -- (4.93634,-2.85);
\draw[black, line width=0.5mm] (5.45596,-3.15) -- (5.80237,-3.35);
\draw[black, line width=0.5mm] (3.4641,-2.7) -- (3.4641,-2.3);
\draw[black, line width=0.5mm] (3.72391,-2.85) -- (4.07032,-2.65);
\draw[black, line width=0.5mm] (1.99186,-1.85) -- (2.33827,-1.65);
\draw[black, line width=0.5mm] (3.4641,-1.7) -- (3.4641,-1.3);
\draw[black, line width=0.5mm] (3.72391,-1.85) -- (4.07032,-1.65);
\draw[black, line width=0.5mm] (3.72391,-2.15) -- (4.07032,-2.35);
\draw[black, line width=0.5mm] (2.59808,-4.2) -- (2.59808,-3.8);
\draw[black, line width=0.5mm] (2.85788,-4.65) -- (3.20429,-4.85);
\draw[black!20!red, line width=0.48mm, ] (0.866025,-2.5) circle (0.408);
\draw[black!20!red, line width=0.48mm, ] (2.59808,-3.5) circle (0.408);
\draw[black!20!red, line width=0.48mm, ] (4.33013,-4.5) circle (0.408);
    \end{tikzpicture}
    \end{center}
    \caption{Shift $p_1,p_2$ down-right.}
    \label{fig:shiftable_1}
  \end{subfigure}%
  \hfill
  \begin{subfigure}{.27\textwidth}
    \centering
    \begin{center}
    \begin{tikzpicture}[x=0.4cm,y=0.4cm]
    \draw[lightgray] (2.59808,-0.5) -- (2.59808,-6.5);
\draw[lightgray] (0,-2) -- (3.4641,0);
\draw[lightgray] (0,-2) -- (7.79423,-6.5);
\draw[lightgray] (4.33013,-6.5) -- (7.79423,-4.5);
\draw[lightgray] (6.9282,0) -- (7.79423,-0.5);
\draw[lightgray] (6.9282,0) -- (6.9282,-6);
\draw[lightgray] (0.866025,-6.5) -- (7.79423,-2.5);
\draw[lightgray] (0,-3) -- (5.19615,0);
\draw[lightgray] (0,-3) -- (6.06218,-6.5);
\draw[lightgray] (3.4641,0) -- (7.79423,-2.5);
\draw[lightgray] (3.4641,0) -- (3.4641,-6);
\draw[lightgray] (7.79423,-0.5) -- (7.79423,-6.5);
\draw[lightgray] (6.06218,-6.5) -- (7.79423,-5.5);
\draw[lightgray] (2.59808,-6.5) -- (7.79423,-3.5);
\draw[lightgray] (0,-4) -- (6.9282,0);
\draw[lightgray] (0,-4) -- (4.33013,-6.5);
\draw[lightgray] (4.33013,-0.5) -- (4.33013,-6.5);
\draw[lightgray] (0,0) -- (7.79423,-4.5);
\draw[lightgray] (0,0) -- (0,-6);
\draw[lightgray] (5.19615,0) -- (7.79423,-1.5);
\draw[lightgray] (5.19615,0) -- (5.19615,-6);
\draw[lightgray] (0,-5) -- (7.79423,-0.5);
\draw[lightgray] (0,-5) -- (2.59808,-6.5);
\draw[lightgray] (0.866025,-0.5) -- (0.866025,-6.5);
\draw[lightgray] (0,-1) -- (1.73205,0);
\draw[lightgray] (0,-1) -- (7.79423,-5.5);
\draw[lightgray] (6.06218,-0.5) -- (6.06218,-6.5);
\draw[lightgray] (0,-6) -- (7.79423,-1.5);
\draw[lightgray] (0,-6) -- (0.866025,-6.5);
\draw[lightgray] (1.73205,0) -- (7.79423,-3.5);
\draw[lightgray] (1.73205,0) -- (1.73205,-6);
\draw[black, line width=0.4mm, fill=white] (0.866025,-2.5) circle (0.288);
\node[align=left] at (0.519615,-1.8) {\small $p_1$};
\draw[black, line width=0.4mm, fill=white] (0.866025,-3.5) circle (0.288);
\draw[black, line width=0.4mm, fill=white] (1.73205,-2) circle (0.288);
\draw[black, line width=0.4mm, fill=white] (2.59808,-1.5) circle (0.288);
\draw[black, line width=0.4mm, fill=white] (2.59808,-3.5) circle (0.288);
\node[align=left] at (2.25167,-2.8) {\small $p_2$};
\draw[black, line width=0.4mm, fill=white] (2.59808,-4.5) circle (0.288);
\draw[black, line width=0.4mm, fill=white] (3.4641,-2) circle (0.288);
\draw[black, line width=0.4mm, fill=white] (3.4641,-3) circle (0.288);
\draw[black, line width=0.4mm, fill=white] (3.4641,-5) circle (0.288);
\draw[black, line width=0.4mm, fill=white] (4.33013,-1.5) circle (0.288);
\draw[black, line width=0.4mm, fill=white] (4.33013,-2.5) circle (0.288);
\node[align=left] at (3.98372,-3.8) {\small $p_3$};
\draw[black, line width=0.4mm, fill=white] (4.33013,-5.5) circle (0.288);
\draw[black, line width=0.4mm, fill=white] (5.19615,-1) circle (0.288);
\draw[black, line width=0.4mm, fill=white] (5.19615,-2) circle (0.288);
\draw[black, line width=0.4mm, fill=white] (5.19615,-5) circle (0.288);
\draw[black, line width=0.4mm, fill=white] (5.19615,-6) circle (0.288);
\draw[black, line width=0.4mm, fill=white] (6.06218,-2.5) circle (0.288);
\draw[black, line width=0.4mm, fill=white] (6.06218,-3.5) circle (0.288);
\draw[black, line width=0.4mm, fill=white] (6.06218,-4.5) circle (0.288);
\draw[black, line width=0.5mm] (6.06218,-3.2) -- (6.06218,-2.8);
\draw[black, line width=0.5mm] (0.866025,-3.2) -- (0.866025,-2.8);
\draw[black, line width=0.5mm] (5.19615,-5.7) -- (5.19615,-5.3);
\draw[black, line width=0.5mm] (5.19615,-1.7) -- (5.19615,-1.3);
\draw[black, line width=0.5mm] (5.45596,-2.15) -- (5.80237,-2.35);
\draw[black, line width=0.5mm] (4.58993,-1.35) -- (4.93634,-1.15);
\draw[black, line width=0.5mm] (4.58993,-1.65) -- (4.93634,-1.85);
\draw[black, line width=0.5mm] (1.12583,-2.35) -- (1.47224,-2.15);
\draw[black, line width=0.5mm] (2.85788,-3.35) -- (3.20429,-3.15);
\draw[black, line width=0.5mm] (4.58993,-5.35) -- (4.93634,-5.15);
\draw[black, line width=0.5mm] (4.58993,-5.65) -- (4.93634,-5.85);
\draw[black, line width=0.5mm] (1.99186,-1.85) -- (2.33827,-1.65);
\draw[black, line width=0.5mm] (2.59808,-4.2) -- (2.59808,-3.8);
\draw[black, line width=0.5mm] (2.85788,-4.65) -- (3.20429,-4.85);
\draw[black, line width=0.5mm] (3.72391,-5.15) -- (4.07032,-5.35);
\draw[black, line width=0.5mm] (3.72391,-1.85) -- (4.07032,-1.65);
\draw[black, line width=0.5mm] (3.72391,-2.15) -- (4.07032,-2.35);
\draw[black, line width=0.5mm] (5.45596,-4.85) -- (5.80237,-4.65);
\draw[black, line width=0.5mm] (2.85788,-1.65) -- (3.20429,-1.85);
\draw[black, line width=0.5mm] (6.06218,-4.2) -- (6.06218,-3.8);
\draw[black, line width=0.5mm] (3.4641,-2.7) -- (3.4641,-2.3);
\draw[black, line width=0.5mm] (3.72391,-2.85) -- (4.07032,-2.65);
\draw[black, line width=0.5mm] (4.33013,-2.2) -- (4.33013,-1.8);
\draw[black, line width=0.5mm] (4.58993,-2.35) -- (4.93634,-2.15);
\draw[black!20!red, line width=0.48mm, ] (0.866025,-2.5) circle (0.408);
\draw[black!20!red, line width=0.48mm, ] (2.59808,-3.5) circle (0.408);
\draw[black!20!red, line width=0.48mm, ] (4.33013,-4.5) circle (0.408);
    \end{tikzpicture}
    \end{center}
    \caption{Shift $p_1,p_2$ down-right.}
    \label{fig:shiftable_2}
  \end{subfigure}%
  \begin{subfigure}{.25\textwidth}
    \centering
    \begin{center}
    \begin{tikzpicture}[x=0.4cm,y=0.4cm]
    \draw[lightgray] (2.59808,-0.5) -- (2.59808,-6.5);
\draw[lightgray] (0,-2) -- (3.4641,0);
\draw[lightgray] (0,-2) -- (7.79423,-6.5);
\draw[lightgray] (4.33013,-6.5) -- (7.79423,-4.5);
\draw[lightgray] (6.9282,0) -- (7.79423,-0.5);
\draw[lightgray] (6.9282,0) -- (6.9282,-6);
\draw[lightgray] (0.866025,-6.5) -- (7.79423,-2.5);
\draw[lightgray] (0,-3) -- (5.19615,0);
\draw[lightgray] (0,-3) -- (6.06218,-6.5);
\draw[lightgray] (3.4641,0) -- (7.79423,-2.5);
\draw[lightgray] (3.4641,0) -- (3.4641,-6);
\draw[lightgray] (7.79423,-0.5) -- (7.79423,-6.5);
\draw[lightgray] (6.06218,-6.5) -- (7.79423,-5.5);
\draw[lightgray] (2.59808,-6.5) -- (7.79423,-3.5);
\draw[lightgray] (0,-4) -- (6.9282,0);
\draw[lightgray] (0,-4) -- (4.33013,-6.5);
\draw[lightgray] (4.33013,-0.5) -- (4.33013,-6.5);
\draw[lightgray] (0,0) -- (7.79423,-4.5);
\draw[lightgray] (0,0) -- (0,-6);
\draw[lightgray] (5.19615,0) -- (7.79423,-1.5);
\draw[lightgray] (5.19615,0) -- (5.19615,-6);
\draw[lightgray] (0,-5) -- (7.79423,-0.5);
\draw[lightgray] (0,-5) -- (2.59808,-6.5);
\draw[lightgray] (0.866025,-0.5) -- (0.866025,-6.5);
\draw[lightgray] (0,-1) -- (1.73205,0);
\draw[lightgray] (0,-1) -- (7.79423,-5.5);
\draw[lightgray] (6.06218,-0.5) -- (6.06218,-6.5);
\draw[lightgray] (0,-6) -- (7.79423,-1.5);
\draw[lightgray] (0,-6) -- (0.866025,-6.5);
\draw[lightgray] (1.73205,0) -- (7.79423,-3.5);
\draw[lightgray] (1.73205,0) -- (1.73205,-6);
\draw[black, line width=0.4mm, fill=white] (0.866025,-2.5) circle (0.288);
\node[align=left] at (0.519615,-1.8) {\small $p_1$};
\draw[black, line width=0.4mm, fill=white] (0.866025,-3.5) circle (0.288);
\draw[black, line width=0.4mm, fill=white] (1.73205,-2) circle (0.288);
\draw[black, line width=0.4mm, fill=white] (2.59808,-1.5) circle (0.288);
\draw[black, line width=0.4mm, fill=white] (2.59808,-3.5) circle (0.288);
\node[align=left] at (2.25167,-2.8) {\small $p_2$};
\draw[black, line width=0.4mm, fill=white] (2.59808,-4.5) circle (0.288);
\draw[black, line width=0.4mm, fill=white] (3.4641,-1) circle (0.288);
\draw[black, line width=0.4mm, fill=white] (3.4641,-2) circle (0.288);
\draw[black, line width=0.4mm, fill=white] (3.4641,-3) circle (0.288);
\draw[black, line width=0.4mm, fill=white] (4.33013,-0.5) circle (0.288);
\draw[black, line width=0.4mm, fill=white] (4.33013,-1.5) circle (0.288);
\draw[black, line width=0.4mm, fill=white] (4.33013,-4.5) circle (0.288);
\node[align=left] at (3.98372,-3.8) {\small $p_3$};
\draw[black, line width=0.4mm, fill=white] (4.33013,-5.5) circle (0.288);
\draw[black, line width=0.4mm, fill=white] (5.19615,-2) circle (0.288);
\draw[black, line width=0.4mm, fill=white] (5.19615,-4) circle (0.288);
\draw[black, line width=0.4mm, fill=white] (5.19615,-5) circle (0.288);
\draw[black, line width=0.4mm, fill=white] (6.06218,-2.5) circle (0.288);
\draw[black, line width=0.4mm, fill=white] (6.06218,-3.5) circle (0.288);
\draw[black, line width=0.5mm] (1.12583,-2.35) -- (1.47224,-2.15);
\draw[black, line width=0.5mm] (5.45596,-3.85) -- (5.80237,-3.65);
\draw[black, line width=0.5mm] (3.72391,-0.85) -- (4.07032,-0.65);
\draw[black, line width=0.5mm] (3.72391,-1.15) -- (4.07032,-1.35);
\draw[black, line width=0.5mm] (0.866025,-3.2) -- (0.866025,-2.8);
\draw[black, line width=0.5mm] (4.33013,-1.2) -- (4.33013,-0.8);
\draw[black, line width=0.5mm] (4.58993,-1.65) -- (4.93634,-1.85);
\draw[black, line width=0.5mm] (2.85788,-3.35) -- (3.20429,-3.15);
\draw[black, line width=0.5mm] (4.58993,-4.35) -- (4.93634,-4.15);
\draw[black, line width=0.5mm] (4.58993,-4.65) -- (4.93634,-4.85);
\draw[black, line width=0.5mm] (4.33013,-5.2) -- (4.33013,-4.8);
\draw[black, line width=0.5mm] (4.58993,-5.35) -- (4.93634,-5.15);
\draw[black, line width=0.5mm] (6.06218,-3.2) -- (6.06218,-2.8);
\draw[black, line width=0.5mm] (2.85788,-1.35) -- (3.20429,-1.15);
\draw[black, line width=0.5mm] (2.85788,-1.65) -- (3.20429,-1.85);
\draw[black, line width=0.5mm] (5.45596,-2.15) -- (5.80237,-2.35);
\draw[black, line width=0.5mm] (3.4641,-2.7) -- (3.4641,-2.3);
\draw[black, line width=0.5mm] (1.99186,-1.85) -- (2.33827,-1.65);
\draw[black, line width=0.5mm] (3.4641,-1.7) -- (3.4641,-1.3);
\draw[black, line width=0.5mm] (3.72391,-1.85) -- (4.07032,-1.65);
\draw[black, line width=0.5mm] (2.59808,-4.2) -- (2.59808,-3.8);
\draw[black, line width=0.5mm] (5.19615,-4.7) -- (5.19615,-4.3);
\draw[black!20!red, line width=0.48mm, ] (0.866025,-2.5) circle (0.408);
\draw[black!20!red, line width=0.48mm, ] (2.59808,-3.5) circle (0.408);
\draw[black!20!red, line width=0.48mm, ] (4.33013,-4.5) circle (0.408);
    \end{tikzpicture}
    \end{center}
    \caption{Shift $p_2,p_3$ up-left.}
    \label{fig:shiftable_3}
  \end{subfigure}%
  \begin{subfigure}{.22\textwidth}
    \centering
    \begin{center}
    \begin{tikzpicture}[x=0.4cm,y=0.4cm]
    \draw[lightgray] (3.4641,-7) -- (5.19615,-6);
\draw[lightgray] (2.59808,-0.5) -- (5.19615,-2);
\draw[lightgray] (2.59808,-0.5) -- (2.59808,-6.5);
\draw[lightgray] (0,-2) -- (2.59808,-0.5);
\draw[lightgray] (0,-2) -- (5.19615,-5);
\draw[lightgray] (0,-7) -- (5.19615,-4);
\draw[lightgray] (1.73205,-1) -- (1.73205,-7);
\draw[lightgray] (0,-3) -- (4.33013,-0.5);
\draw[lightgray] (0,-3) -- (5.19615,-6);
\draw[lightgray] (0,-4) -- (5.19615,-1);
\draw[lightgray] (0,-4) -- (5.19615,-7);
\draw[lightgray] (4.33013,-0.5) -- (5.19615,-1);
\draw[lightgray] (4.33013,-0.5) -- (4.33013,-6.5);
\draw[lightgray] (3.4641,-1) -- (3.4641,-7);
\draw[lightgray] (0,-5) -- (5.19615,-2);
\draw[lightgray] (0,-5) -- (3.4641,-7);
\draw[lightgray] (0.866025,-0.5) -- (5.19615,-3);
\draw[lightgray] (0.866025,-0.5) -- (0.866025,-6.5);
\draw[lightgray] (0,-1) -- (0.866025,-0.5);
\draw[lightgray] (0,-1) -- (5.19615,-4);
\draw[lightgray] (0,-1) -- (0,-7);
\draw[lightgray] (1.73205,-7) -- (5.19615,-5);
\draw[lightgray] (5.19615,-1) -- (5.19615,-7);
\draw[lightgray] (0,-6) -- (5.19615,-3);
\draw[lightgray] (0,-6) -- (1.73205,-7);
\draw[black, line width=0.4mm, fill=white] (1.73205,-6) circle (0.288);
\draw[black, line width=0.4mm, fill=white] (2.59808,-5.5) circle (0.288);
\draw[black, line width=0.4mm, fill=white] (3.4641,-2) circle (0.288);
\node[align=left] at (3.11769,-1.4) {\small $p$};
\draw[black, line width=0.4mm, fill=white] (3.4641,-3) circle (0.288);
\draw[black, line width=0.4mm, fill=white] (3.4641,-4) circle (0.288);
\draw[black, line width=0.4mm, fill=white] (3.4641,-5) circle (0.288);
\draw[black, line width=0.4mm, fill=white] (4.33013,-1.5) circle (0.288);
\draw[black, line width=0.4mm, fill=white] (4.33013,-2.5) circle (0.288);
\draw[black, line width=0.5mm] (1.99186,-5.85) -- (2.33827,-5.65);
\draw[black, line width=0.5mm] (3.4641,-4.7) -- (3.4641,-4.3);
\draw[black, line width=0.5mm] (3.4641,-3.7) -- (3.4641,-3.3);
\draw[black, line width=0.5mm] (3.4641,-2.7) -- (3.4641,-2.3);
\draw[black, line width=0.5mm] (3.72391,-2.85) -- (4.07032,-2.65);
\draw[black, line width=0.5mm] (3.72391,-1.85) -- (4.07032,-1.65);
\draw[black, line width=0.5mm] (3.72391,-2.15) -- (4.07032,-2.35);
\draw[black, line width=0.5mm] (4.33013,-2.2) -- (4.33013,-1.8);
\draw[black, line width=0.5mm] (2.85788,-5.35) -- (3.20429,-5.15);
\draw[black!20!red,-{Stealth[length=1.6mm,width=2.5mm]},line width=0.7mm] (0.866025,-5.81) -- (0.866025,-6.5);
\draw[black!20!red,-{Stealth[length=1.6mm,width=2.5mm]},line width=0.7mm] (1.46358,-5.155) -- (0.866025,-5.5);
\draw[black!20!red,-{Stealth[length=1.6mm,width=2.5mm]},line width=0.7mm] (2.59808,-2.81) -- (2.59808,-3.5);
\draw[black!20!red,-{Stealth[length=1.6mm,width=2.5mm]},line width=0.7mm] (2.59808,-3.81) -- (2.59808,-4.5);
\draw[black!20!red,-{Stealth[length=1.6mm,width=2.5mm]},line width=0.7mm] (2.32961,-4.655) -- (1.73205,-5);
\draw[black!20!red, line width=0.48mm, ] (3.4641,-2) circle (0.408);
\draw[black!20!red,-{Stealth[length=1.6mm,width=2.5mm]},line width=0.7mm] (3.19563,-2.155) -- (2.59808,-2.5);
    \end{tikzpicture}
    \end{center}
    \caption{$p$ is non-shiftable.}
    \label{fig:non_shiftable}
  \end{subfigure}%
\caption{Figures~\ref{fig:shiftable_1}, \ref{fig:shiftable_2} and \ref{fig:shiftable_3} illustrate different cases for the ``shift'' operation. In all of these images, sites $p_1$ and $p_2$ are shiftable particles, while $p_3$ is not.}
\end{figure*}

\vskip.1in
\noindent \underbar{\bf The comb operation:}
The comb procedure (applied to a combable position $(\ell,d)$) has two phases, line formation and line merging.
After the line formation phase, the first two conditions for $(\ell,d)$ to be combed will be satisfied by the configuration (Figure~\ref{fig:line_merging_before}). The line merging phase gives us the third condition (Figure~\ref{fig:line_merging_after}).

\vskip.1in
\noindent \underbar{\textbf{Line formation}}
Let $L$ denote the set of particles on lane $\ell$ on or below $(\ell,d)$. The particles in $L$ can be grouped into connected components within $L$.
The line formation phase operates from top to bottom on $L$, removing the topmost particle of each component with size greater than $1$ at each turn, until every component on $L$ has size $1$. We maintain the invariant that $(\ell+1,d+1)$ is combed after each turn, while reducing the number of particles in $L$ by $1$.

We call a site ``shiftable'' if there is a particle on the site, and it has exactly two neighboring particles, one directly below and one directly up-right of it. In a turn, there are two possible cases. Denote by $p$ the topmost particle of the topmost component with size greater than $1$.

If $p$ is shiftable, we apply what we call a ``shift'', which moves a set of particles on a line either down-right or up-right, so that $p$ either becomes unoccupied or non-shiftable. To apply a shift, we consider the sequence of sites $p = p_1, p_2, \dots$, where each site $p_{i+1}$ is exactly two steps down-right of site $p_i$. Let $k$ be the largest integer such that all of the sites $p_1, p_2, \dots, p_k$ are shiftable. Figures~\ref{fig:shiftable_1}, \ref{fig:shiftable_2} and \ref{fig:shiftable_3} illustrate examples where $k=2$.
Consider the first non-shiftable site $p_{k+1}$ in the sequence, and the sites directly above and directly down-left of $p_{k+1}$, which we will call $p_{k+1}^{U}$ and $p_{k+1}^{DL}$ respectively. If $p_{k+1}$ is unoccupied or either of $p_{k+1}^{U}$ or $p_{k+1}^{DL}$ are occupied, then moving $p_{k}$ one step down-right is a valid move (Figures~\ref{fig:shiftable_1}, \ref{fig:shiftable_2}). We can thus go backwards through the sequence from $p_{k}$ to $p_1$, moving each particle one step down-right, ending with shifting $p = p_1$ one step down-right, so that the site $p$ originally was occupying now becomes unoccupied.
On the other hand, if $p_{k+1}$ is occupied but $p_{k+1}^{U}$ and $p_{k+1}^{DL}$ aren't, as $p_{k+1}$ is non-shiftable, the remaining three neighbors (down, down-right and up-right) must form a single component, meaning moving $p_{k+1}$ one step up-left is a valid move (Figure~\ref{fig:shiftable_3}). We can then subsequently move each particle from $p_{k+1}$ to $p_2$ one step up-left, culminating in $p = p_1$ becoming non-shiftable, leading in to the second case which we will describe next. Note that after a shift, the invariant that $(\ell+1,d+1)$ is combed still holds.

If $p$ is not shiftable, by the invariant we maintain, the sites above, up-left and down-left of $p$ must be unoccupied, while the site directly below $p$ is occupied. Thus, if the site up-right of $p$ is occupied, so must the site down-right of $p$. As $(\ell+1,d+1)$ is combed, the component on $L$ $p$ belongs to will have no particles down-left of it, except for potentially one line of particles extending from the bottommost particle of the component. The particle at $p$ can thus be moved down-left, down along the component on $L$ $p$ belongs to, then down-left to reach the end of the beforementioned line if it exists, and down once more to join the end of this line (Figure~\ref{fig:non_shiftable}). In all, this reduces the number of particles in $L$ by $1$, while maintaining the invariant that $(\ell+1,d+1)$ is combed.

Thus, after the line formation phase, every component in $L$ will have exactly $1$ particle, while $(\ell+1,d+1)$ remains combed. As the site above $(\ell,d)$ is empty, the first two conditions for $(\ell,d)$ being combed are satisfied. The line merging phase will give us the third condition. Figures~\ref{fig:line_merging_before} and \ref{fig:line_merging_after} illustrate configurations before and after the line merging phase.

\vskip.1in
\noindent \underbar{\textbf{Line merging}}
Let $C$ denote the column of sites directly to the right of $(\ell,d)$. The lines extending down and left in the residual region of $(\ell,d)$ may extend from particles in $C$ that are not the bottommost particles of their respective components. To fix this, the line merging phase processs these lines from the lowest to the highest. To move a line downwards by one step, the particles of the line are shifted down one by one, starting from the rightmost particle of the line and ending with the particle on the end of the line. These moves are always possible as long as there is no line directly below the current line. If there is a line directly below, we merge the current line into the line below by moving the particles one at the time to the end of the line below with a straightforward sequence of moves, starting with the leftmost particle of the current line.

\begin{figure}[h]
\begin{subfigure}{.5\linewidth}
  \centering
  \begin{center}
  \begin{tikzpicture}[x=0.4cm,y=0.4cm]
  \draw[lightgray] (0,-7) -- (6.06218,-3.5);
\draw[lightgray] (0,-7) -- (5.19615,-10);
\draw[lightgray] (0,-10) -- (6.06218,-6.5);
\draw[lightgray] (3.4641,0) -- (6.06218,-1.5);
\draw[lightgray] (3.4641,0) -- (3.4641,-10);
\draw[lightgray] (5.19615,-10) -- (6.06218,-9.5);
\draw[lightgray] (0,-4) -- (6.06218,-0.5);
\draw[lightgray] (0,-4) -- (6.06218,-7.5);
\draw[lightgray] (3.4641,-10) -- (6.06218,-8.5);
\draw[lightgray] (5.19615,0) -- (6.06218,-0.5);
\draw[lightgray] (5.19615,0) -- (5.19615,-10);
\draw[lightgray] (0,-1) -- (1.73205,0);
\draw[lightgray] (0,-1) -- (6.06218,-4.5);
\draw[lightgray] (0,-5) -- (6.06218,-1.5);
\draw[lightgray] (0,-5) -- (6.06218,-8.5);
\draw[lightgray] (0.866025,-0.5) -- (0.866025,-9.5);
\draw[lightgray] (0,-8) -- (6.06218,-4.5);
\draw[lightgray] (0,-8) -- (3.4641,-10);
\draw[lightgray] (0,-2) -- (3.4641,0);
\draw[lightgray] (0,-2) -- (6.06218,-5.5);
\draw[lightgray] (2.59808,-0.5) -- (2.59808,-9.5);
\draw[lightgray] (4.33013,-0.5) -- (4.33013,-9.5);
\draw[lightgray] (6.06218,-0.5) -- (6.06218,-9.5);
\draw[lightgray] (0,-6) -- (6.06218,-2.5);
\draw[lightgray] (0,-6) -- (6.06218,-9.5);
\draw[lightgray] (0,-9) -- (6.06218,-5.5);
\draw[lightgray] (0,-9) -- (1.73205,-10);
\draw[lightgray] (0,-3) -- (5.19615,0);
\draw[lightgray] (0,-3) -- (6.06218,-6.5);
\draw[lightgray] (0,0) -- (6.06218,-3.5);
\draw[lightgray] (0,0) -- (0,-10);
\draw[lightgray] (1.73205,-10) -- (6.06218,-7.5);
\draw[lightgray] (1.73205,0) -- (6.06218,-2.5);
\draw[lightgray] (1.73205,0) -- (1.73205,-10);
\draw[black, line width=0.4mm, fill=white] (2.59808,-5.5) circle (0.288);
\draw[black, line width=0.4mm, fill=white] (3.4641,-3) circle (0.288);
\draw[black, line width=0.4mm, fill=white] (3.4641,-5) circle (0.288);
\draw[black, line width=0.4mm, fill=white] (3.4641,-8) circle (0.288);
\draw[line width=0.4mm] (4.58469,-1.24544) -- (4.07557,-1.75456);
\draw[line width=0.4mm] (4.07557,-1.24544) -- (4.58469,-1.75456);
\node[align=left] at (3.81051,-0.8) {\normalsize $(\ell,d)$};
\draw[black, line width=0.4mm, fill=white] (4.33013,-2.5) circle (0.288);
\draw[black, line width=0.4mm, fill=white] (4.33013,-4.5) circle (0.288);
\draw[black, line width=0.4mm, fill=white] (4.33013,-7.5) circle (0.288);
\draw[black, line width=0.4mm, fill=white] (5.19615,-2) circle (0.288);
\draw[black, line width=0.4mm, fill=white] (5.19615,-3) circle (0.288);
\draw[black, line width=0.4mm, fill=white] (5.19615,-4) circle (0.288);
\draw[black, line width=0.4mm, fill=white] (5.19615,-5) circle (0.288);
\draw[black, line width=0.4mm, fill=white] (5.19615,-7) circle (0.288);
\draw[black, line width=0.4mm, fill=white] (5.19615,-8) circle (0.288);
\draw[black, line width=0.5mm] (5.19615,-3.7) -- (5.19615,-3.3);
\draw[black, line width=0.5mm] (4.58993,-4.35) -- (4.93634,-4.15);
\draw[black, line width=0.5mm] (4.58993,-4.65) -- (4.93634,-4.85);
\draw[black, line width=0.5mm] (5.19615,-7.7) -- (5.19615,-7.3);
\draw[black, line width=0.5mm] (3.72391,-7.85) -- (4.07032,-7.65);
\draw[black, line width=0.5mm] (3.72391,-4.85) -- (4.07032,-4.65);
\draw[black, line width=0.5mm] (4.58993,-7.35) -- (4.93634,-7.15);
\draw[black, line width=0.5mm] (4.58993,-7.65) -- (4.93634,-7.85);
\draw[black, line width=0.5mm] (5.19615,-2.7) -- (5.19615,-2.3);
\draw[black, line width=0.5mm] (3.72391,-2.85) -- (4.07032,-2.65);
\draw[black, line width=0.5mm] (4.58993,-2.35) -- (4.93634,-2.15);
\draw[black, line width=0.5mm] (4.58993,-2.65) -- (4.93634,-2.85);
\draw[black, line width=0.5mm] (5.19615,-4.7) -- (5.19615,-4.3);
\draw[black, line width=0.5mm] (2.85788,-5.35) -- (3.20429,-5.15);
\path [pattern=north east lines, pattern color=yellow] (4.33013,-1.5) -- (0,-4) -- (0,-10) -- (3.4641,-10) -- (4.33013,-9.5) -- cycle;
\draw[dashed, line width=0.6mm] (4.33013,-1.5) -- (0,-4);
\draw[dashed, line width=0.6mm] (4.33013,-1.5) -- (4.33013,-9.5);
  \end{tikzpicture}
  \end{center}
  \caption{After line formation, before line merging.}
  \label{fig:line_merging_before}
\end{subfigure}%
\hfill
\begin{subfigure}{.5\linewidth}
  \centering
  \begin{center}
  \begin{tikzpicture}[x=0.4cm,y=0.4cm]
  \draw[lightgray] (0,-7) -- (6.06218,-3.5);
\draw[lightgray] (0,-7) -- (5.19615,-10);
\draw[lightgray] (0,-10) -- (6.06218,-6.5);
\draw[lightgray] (3.4641,0) -- (6.06218,-1.5);
\draw[lightgray] (3.4641,0) -- (3.4641,-10);
\draw[lightgray] (5.19615,-10) -- (6.06218,-9.5);
\draw[lightgray] (0,-4) -- (6.06218,-0.5);
\draw[lightgray] (0,-4) -- (6.06218,-7.5);
\draw[lightgray] (3.4641,-10) -- (6.06218,-8.5);
\draw[lightgray] (5.19615,0) -- (6.06218,-0.5);
\draw[lightgray] (5.19615,0) -- (5.19615,-10);
\draw[lightgray] (0,-1) -- (1.73205,0);
\draw[lightgray] (0,-1) -- (6.06218,-4.5);
\draw[lightgray] (0,-5) -- (6.06218,-1.5);
\draw[lightgray] (0,-5) -- (6.06218,-8.5);
\draw[lightgray] (0.866025,-0.5) -- (0.866025,-9.5);
\draw[lightgray] (0,-8) -- (6.06218,-4.5);
\draw[lightgray] (0,-8) -- (3.4641,-10);
\draw[lightgray] (0,-2) -- (3.4641,0);
\draw[lightgray] (0,-2) -- (6.06218,-5.5);
\draw[lightgray] (2.59808,-0.5) -- (2.59808,-9.5);
\draw[lightgray] (4.33013,-0.5) -- (4.33013,-9.5);
\draw[lightgray] (6.06218,-0.5) -- (6.06218,-9.5);
\draw[lightgray] (0,-6) -- (6.06218,-2.5);
\draw[lightgray] (0,-6) -- (6.06218,-9.5);
\draw[lightgray] (0,-9) -- (6.06218,-5.5);
\draw[lightgray] (0,-9) -- (1.73205,-10);
\draw[lightgray] (0,-3) -- (5.19615,0);
\draw[lightgray] (0,-3) -- (6.06218,-6.5);
\draw[lightgray] (0,0) -- (6.06218,-3.5);
\draw[lightgray] (0,0) -- (0,-10);
\draw[lightgray] (1.73205,-10) -- (6.06218,-7.5);
\draw[lightgray] (1.73205,0) -- (6.06218,-2.5);
\draw[lightgray] (1.73205,0) -- (1.73205,-10);
\draw[black, line width=0.4mm, fill=white] (0.866025,-7.5) circle (0.288);
\draw[black, line width=0.4mm, fill=white] (1.73205,-7) circle (0.288);
\draw[black, line width=0.4mm, fill=white] (2.59808,-6.5) circle (0.288);
\draw[black, line width=0.4mm, fill=white] (3.4641,-6) circle (0.288);
\draw[black, line width=0.4mm, fill=white] (3.4641,-9) circle (0.288);
\draw[line width=0.4mm] (4.58469,-1.24544) -- (4.07557,-1.75456);
\draw[line width=0.4mm] (4.07557,-1.24544) -- (4.58469,-1.75456);
\node[align=left] at (3.81051,-0.8) {\normalsize $(\ell,d)$};
\draw[black, line width=0.4mm, fill=white] (4.33013,-5.5) circle (0.288);
\draw[black, line width=0.4mm, fill=white] (4.33013,-8.5) circle (0.288);
\draw[black, line width=0.4mm, fill=white] (5.19615,-2) circle (0.288);
\draw[black, line width=0.4mm, fill=white] (5.19615,-3) circle (0.288);
\draw[black, line width=0.4mm, fill=white] (5.19615,-4) circle (0.288);
\draw[black, line width=0.4mm, fill=white] (5.19615,-5) circle (0.288);
\draw[black, line width=0.4mm, fill=white] (5.19615,-7) circle (0.288);
\draw[black, line width=0.4mm, fill=white] (5.19615,-8) circle (0.288);
\draw[black, line width=0.5mm] (1.99186,-6.85) -- (2.33827,-6.65);
\draw[black, line width=0.5mm] (5.19615,-3.7) -- (5.19615,-3.3);
\draw[black, line width=0.5mm] (3.72391,-8.85) -- (4.07032,-8.65);
\draw[black, line width=0.5mm] (3.72391,-5.85) -- (4.07032,-5.65);
\draw[black, line width=0.5mm] (4.58993,-5.35) -- (4.93634,-5.15);
\draw[black, line width=0.5mm] (5.19615,-7.7) -- (5.19615,-7.3);
\draw[black, line width=0.5mm] (5.19615,-2.7) -- (5.19615,-2.3);
\draw[black, line width=0.5mm] (5.19615,-4.7) -- (5.19615,-4.3);
\draw[black, line width=0.5mm] (2.85788,-6.35) -- (3.20429,-6.15);
\draw[black, line width=0.5mm] (1.12583,-7.35) -- (1.47224,-7.15);
\draw[black, line width=0.5mm] (4.58993,-8.35) -- (4.93634,-8.15);
\path [pattern=north east lines, pattern color=yellow] (4.33013,-1.5) -- (0,-4) -- (0,-10) -- (3.4641,-10) -- (4.33013,-9.5) -- cycle;
\draw[dashed, line width=0.6mm] (4.33013,-1.5) -- (0,-4);
\draw[dashed, line width=0.6mm] (4.33013,-1.5) -- (4.33013,-9.5);
  \end{tikzpicture}
  \end{center}
  \caption{After line merging.}
  \label{fig:line_merging_after}
\end{subfigure}%
\caption{The results after the line formation and line merging phases when the comb procedure is applied to a combable position $(\ell,d)$.}
\end{figure}
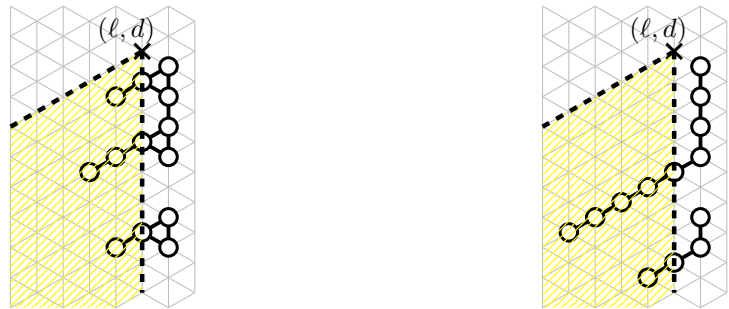

The description of the comb procedure above gives us the following Lemma:

\begin{lemma}
\label{lem:combmakescombed}
After executing a comb operation on a combable position $(\ell,d)$, the position $(\ell,d)$ will be combed (Definition~\ref{defn:combed}).
\end{lemma}

The following Lemma states that combs only affect sites below it.

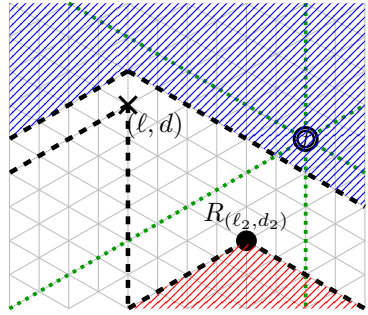
\begin{figure}[h]
\begin{center}
\begin{tikzpicture}[x=0.45cm,y=0.45cm]
\draw[lightgray] (5.19615,-9) -- (10.3923,-6);
\draw[lightgray] (0,-7) -- (10.3923,-1);
\draw[lightgray] (0,-7) -- (3.4641,-9);
\draw[lightgray] (3.4641,0) -- (10.3923,-4);
\draw[lightgray] (3.4641,0) -- (3.4641,-9);
\draw[lightgray] (7.79423,-0.5) -- (7.79423,-8.5);
\draw[lightgray] (0,-4) -- (6.9282,0);
\draw[lightgray] (0,-4) -- (8.66025,-9);
\draw[lightgray] (5.19615,0) -- (10.3923,-3);
\draw[lightgray] (5.19615,0) -- (5.19615,-9);
\draw[lightgray] (9.52628,-0.5) -- (9.52628,-8.5);
\draw[lightgray] (0,-1) -- (1.73205,0);
\draw[lightgray] (0,-1) -- (10.3923,-7);
\draw[lightgray] (6.9282,-9) -- (10.3923,-7);
\draw[lightgray] (8.66025,-9) -- (10.3923,-8);
\draw[lightgray] (0,-5) -- (8.66025,0);
\draw[lightgray] (0,-5) -- (6.9282,-9);
\draw[lightgray] (0.866025,-0.5) -- (0.866025,-8.5);
\draw[lightgray] (0,-8) -- (10.3923,-2);
\draw[lightgray] (0,-8) -- (1.73205,-9);
\draw[lightgray] (0,-2) -- (3.4641,0);
\draw[lightgray] (0,-2) -- (10.3923,-8);
\draw[lightgray] (2.59808,-0.5) -- (2.59808,-8.5);
\draw[lightgray] (6.9282,0) -- (10.3923,-2);
\draw[lightgray] (6.9282,0) -- (6.9282,-9);
\draw[lightgray] (4.33013,-0.5) -- (4.33013,-8.5);
\draw[lightgray] (8.66025,0) -- (10.3923,-1);
\draw[lightgray] (8.66025,0) -- (8.66025,-9);
\draw[lightgray] (6.06218,-0.5) -- (6.06218,-8.5);
\draw[lightgray] (10.3923,0) -- (10.3923,-9);
\draw[lightgray] (0,-6) -- (10.3923,0);
\draw[lightgray] (0,-6) -- (5.19615,-9);
\draw[lightgray] (0,-9) -- (10.3923,-3);
\draw[lightgray] (0,-3) -- (5.19615,0);
\draw[lightgray] (0,-3) -- (10.3923,-9);
\draw[lightgray] (3.4641,-9) -- (10.3923,-5);
\draw[lightgray] (1.73205,-9) -- (10.3923,-4);
\draw[lightgray] (0,0) -- (10.3923,-6);
\draw[lightgray] (0,0) -- (0,-9);
\draw[lightgray] (1.73205,0) -- (10.3923,-5);
\draw[lightgray] (1.73205,0) -- (1.73205,-9);
\draw[black!40!green, dotted, line width=0.5mm] (0,-9) -- (10.3923,-3);
\draw[black!40!green, dotted, line width=0.5mm] (1.73205,0) -- (10.3923,-5);
\draw[black!40!green, dotted, line width=0.5mm] (8.66025,0) -- (8.66025,-9);
\draw[line width=0.4mm] (3.71866,-2.74544) -- (3.20954,-3.25456);
\draw[line width=0.4mm] (3.20954,-2.74544) -- (3.71866,-3.25456);
\node[align=left] at (4.24352,-3.6) {\normalsize $(\ell,d)$};
\draw[fill=black] (6.9282,-7) circle (0.288);
\node[align=left] at (6.9282,-6.3) {\normalsize $R_{(\ell_2,d_2)}$};
\draw[black, line width=0.4mm, fill=white] (8.66025,-4) circle (0.336);
\draw[black, line width=0.32mm] (8.66025,-4) circle (0.24);
\node[align=left] at (8.66025,-4) {\scriptsize $f$};
\path [pattern=north east lines, pattern color=blue] (0,-4) -- (3.4641,-2) -- (10.3923,-6) -- (10.3923,0) -- (0,0) -- cycle;
\draw[dashed, line width=0.6mm] (3.4641,-2) -- (0,-4);
\draw[dashed, line width=0.6mm] (3.4641,-2) -- (10.3923,-6);
\draw[dashed, line width=0.6mm] (3.4641,-3) -- (0,-5);
\draw[dashed, line width=0.6mm] (3.4641,-3) -- (3.4641,-9);
\path [pattern=north east lines, pattern color=red] (3.4641,-9) -- (6.9282,-7) -- (10.3923,-9) -- cycle;
\draw[dashed, line width=0.6mm] (6.9282,-7) -- (3.4641,-9);
\draw[dashed, line width=0.6mm] (6.9282,-7) -- (10.3923,-9);
\end{tikzpicture}
\end{center}
\caption{Illustration of the unaffected region above the position to be combed $(\ell,d)$ from Lemma~\ref{lem:aboveunaffected}, and an unenterable region $R_{(\ell_2,d_2)}$ from Lemma~\ref{lem:unenterableregion} corresponding to some position $(\ell_2,d_2)$ strictly to the right of $(\ell,d)$. Both of these regions include the boundaries drawn in the Figure.}
\label{fig:untouched_region}
\end{figure}

\begin{lemma}[Unaffected Region Above]
\label{lem:aboveunaffected}
Consider the two half-lines extending down-left and down-right from a combable position $(\ell,d)$ as in Figure~\ref{fig:untouched_region}. A comb operation on $(\ell,d)$ will not affect (will not move any particle into or out of) any site above these lines, not including the lines themselves.
In addition, if the site $(\ell-1,d)$ (one step directly down-right) is occupied, no site on the half-line going down-right from $(\ell-1,d)$, including $(\ell-1,d)$ itself, will be affected either.
\end{lemma}

\begin{proof}
We observe that in the comb procedure, aside from the ``shift'' moves, all of the moves occur only within the residual region of $(\ell,d)$. The shift moves only go down-right or up-left, and if the shift originates from some position $p$, the shift does not move any particle up-left from $p$. Hence, the shifts also do not affect any site above the two half-lines described in the Lemma.

We observe that the only part of the procedure that can affect particles on the half-line going down-right from $(\ell-1,d)$ is a potential shift move on a particle on position $(\ell,d)$. However, if $(\ell-1,d)$ is occupied, $(\ell,d)$ will not be shiftable, and so this shift move will not occur.
\end{proof}

To apply a sequence of comb operations to ``push'' particles down towards the \targetspine{}, we only need to find a (not necessarily straight) ``line'' of vacant positions. The following definition and lemma makes this more formal.

\begin{lemma}[Combable Sequence]
Consider a sequence of pairs $((x_1,y_1),(x_2,y_2),\dots,(x_k,y_k))$, where each item in the sequence represents a $(lane,depth)$ pair. We call this a \emph{combable sequence} if:
\begin{enumerate}
\item $x_1$ vertically coincides with the leftmost particle of the configuration.
\item $x_{i+1} = x_i-1$ for all $i \in \{1,2,\dots,k-1\}$ and $x_k > 0$.
\item $y_i \geq 0$ for all $i \in \{1,2,\dots,k\}$.
\item $y_{i+1} \in \{y_i,y_i-1\}$ for all $i \in \{2,\dots,k\}$.
\item The locations $(x_i,y_i-1)$ are all vacant.
\item For each $i \in \{1,2,\dots,k\}$, if $y_i=0$, then the position $(x_i-1,0)$ must be occupied by a particle.
\end{enumerate}
An example of such a sequence is illustrated in Figure~\ref{fig:spine_comb}.
\end{lemma}

\begin{definition}[Combing a Sequence]
\label{proc:combingasequence}
Consider a combable sequence $((x_1,y_1),(x_2,y_2),\dots,(x_k,y_k))$. Combing this sequence refers to combing each pair $(x_i,y_i)$ in succession. The following Lemma justifies that this is always possible.
\end{definition}

\begin{lemma}[Combability of Each Step in a Sequence]
When combing a combable sequence $((x_1,y_1),(x_2,y_2),\dots,(x_k,y_k))$ as described in Definition~\ref{proc:combingasequence}, when $(x_i,y_i)$ is the next position to be combed, $(x_i,y_i)$ will be combable.
\end{lemma}

\begin{proof}
We first note that for any $i$,
by the definition of a combable sequence,
the location directly above $(x_i,y_i)$ must be empty,
and if $y_i = 0$, then $(x_i-1,0)$ is occupied.
These two conditions continue to be true even as combs $1,\dots,i-1$ are executed, as by Lemma~\ref{lem:aboveunaffected}, none of these prior combs will affect $(x_i,y_i-1)$ or $(x_i-1,y_i)$.
This covers two of the conditions necessary for $(x_i,y_i)$ to be combable.

When $i=1$, $(x_1,y_1)$ is clearly combable as there are no particles to the left of $x_1$, and the site directly above $(x_1,y_1)$ is empty.

For $i \geq 2$, as $(x_{i-1},y_{i-1})$ was combed in the previous step, $(x_{i-1},y_{i-1})$ is combed (Lemma~\ref{lem:combmakescombed}). There are two cases for $y_i$.
If $y_i = y_{i-1}-1$, then $(x_i+1,y_i+1) = (x_{i-1},y_{i-1})$ is combed.
On the other hand, if $y_i = y_{i-1}$, then $(x_i+1,y_i) = (x_{i-1},y_{i-1})$ is combed. As $(x_i,y_i-1)$ is unoccupied, every location starting from $(x_i+1,y_i)$ extending down-left must also be unoccupied, which implies $(x_i+1,y_i+1)$ is combed. This gives us our final condition, so $(x_i,y_i)$ is combable.
\end{proof}

In addition to the two properties of the comb operation given as Lemmas~\ref{lem:combmakescombed} and \ref{lem:aboveunaffected}, we state and show a few more properties of the comb operation that we will use later in the proof.

\begin{lemma}[Combing and \Spine{} Lengths]
\label{lem:combingspinelength}
After executing a comb on some position $(\ell,d)$ where $d < \ell$, the length of the \spine{} going down-left will be at most $\ell-1$.
\end{lemma}

\begin{proof}
Let $S$ denote the \spine{} going down-left. The coordinates $(\ell,\ell)$ denotes the position on the \spine{} $S$ vertically below $(\ell,d)$. If $(\ell,\ell)$ is empty after the comb, then every site down-left of $(\ell,\ell)$ is also empty, so the \spine{} $S$ has length at most $\ell-1$. If $(\ell,\ell)$ is not empty after the comb, $(\ell,d)$ being combed ensures that $(\ell,\ell)$ and every particle on the \spine{} $S$ down-left of $(\ell,\ell)$ are tail particles, so \spine{} $S$ has length at most $\ell-1$.
\end{proof}

\begin{lemma}[Preservation of the Rightmost extent]
\label{lem:rightmostextent}
Let $\ell$ denote the lane (x-coordinate) of the rightmost particle of a configuration. After a comb operation is applied, the lanes to the right of $\ell$ (sites with lane less than $\ell$) will continue to be empty.
\end{lemma}

\begin{proof}
Consider a comb applied to some position $(\ell^*,d)$. As we enforce that $\ell^* > 0$ for a comb, this position is necessarily to the left of the food particle, while the rightmost particle of the configuration must be either on the same lane as the food particle or further right.
All moves aside from the ``shift'' moves in a comb procedure of a position $(\ell^*,d)$ operate only within the residual region of $(\ell^*,d)$, and so will not affect any site on lane $\ell$ or further right.

In the shift moves, a particle is only moved rightward (down-right) if it is shiftable. A shiftable particle must have a particle directly up-right of it, so a shift move cannot move a particle rightward of $\ell$.
\end{proof}

\begin{lemma}[Unenterable Region Below]
\label{lem:unenterableregion}
Consider a position $(\ell,d)$ and the diagonal half-lines extending down-left and down-right from $(\ell,d)$. Consider the region $R_{\ell,d}$ containing every location on or below these lines (Figure~\ref{fig:untouched_region}).

If the region $R_{\ell,d}$ is unoccupied, if a comb operation is applied on a lane strictly to the left of lane $\ell$, $R_{\ell,d}$ will continue to be unoccupied after the comb.
\end{lemma}

\begin{proof}
For the shift movements in the line formation phase, a particle is only moved down-right if it is shiftable, which means it must have a particle directly below it. Thus, $R_{\ell,d}$ cannot be entered from the left side (left of lane $\ell$) by this movement unless there is already a particle in $R_{\ell,d}$.
In addition, as the shift movements only move particles down-right or up-left, it cannot move particles into $R_{\ell,d}$ from the right side (right of lane $\ell$).

For the other movements in the comb operation, we only need to consider possibly entry into $R_{\ell,d}$ from the left side, as these movements occur only within the residual region of the comb, which is strictly to the left of $(\ell,d)$.

In the line formation phase, only down-left and downward movements are used. Down-left movements cannot enter $R_{\ell,d}$, and downward movements only occur when there is a particle directly down-right of the particle to be moved.

In the line merging phase, we simply need to consider the end state of the comb. The end state consists of straight lines stretching down-left from the lane (column) $\ell'$ one lane right of the lane to be combed. All of the particles in this lane are above $R_{\ell,d}$ after the line formation phase, and as $\ell' \leq \ell$, all lines stretching down-left from these particles will also be above $R_{\ell,d}$.
\end{proof}

\subsection{Using combs to show ergodicity}
Our objective is to show that from any (connected) configuration, there exists a sequence of valid moves to transform the configuration into a straight line with the food particle at one end. As valid moves cannot introduce holes into a configuration, and all valid moves between hole-free configurations are reversible, this would imply that one can transform any connected configuration of particles into any hole-free configuration of particles using only valid moves.  We proceed by showing that we can always reduce the minimum spine length of any configuration with a series of moves to reach a straight line of particles.

\begin{lemma}
\label{lem:reducespinelength}
If the minimum \spine{} length of a configuration is at least $1$, there exists a sequence of moves to reduce the minimum \spine{} length of the configuration.
\end{lemma}

To reduce the minimum \spine{} length of the configuration, we execute \spinecombs{} in a specific order. Pick a \spine{} of minimum length and denote it as $S_0$. We apply \spinecombs{} in a counterclockwise order, from $S_0$ to $S_1$, followed by $S_1$ to $S_2$, and so on. When applying comb a comb operation from a \sourcespine{} $S_i$ to a \targetspine{} $S_{i+1}$, as always, for ease of analysis, we will treat $S_i$ and $S_{i+1}$ as the \spines{} going in the up-left and down-left directions from the food particle respectively.

\begin{definition}[\SpineComb{}]
\label{proc:spinecomb}
Let $r$ denote the length of the \sourcespine{} $S_0$. Let $r_t$ denote the distance of the furthest out particle on the \sourcespine{} from the food particle (hence the particles of distances $r+1,\dots,r_t$ are the tail particles).

A \spinecomb{} applies a comb on the combable sequence $((x_1,y_1),(x_2,y_2),\dots,(x_k,y_k))$, where $x_1$ vertically coincides with the leftmost particle of the configuration, $x_i = x_1-i+1$ for each $i \in \{2,3,\dots,k\}$, $x_k = r+1$, $y_i = 1$ whenever $x_i > r_t$, and $y_i = 0$ when $r+1 \leq x_i \leq r_t$. From the definition of tail particles one can easily verify that this $(x_i,y_i-1)$ is vacant for each $i \in \{1,2,\dots,k\}$. Figure~\ref{fig:comb_before_after} illustrates configurations before and after a \spinecomb{} is applied.
\end{definition}

\begin{lemma}
\label{lem:spinecomb}
Consider a \spinecomb{} from a \sourcespine{} of length $r$.
After the \spinecomb{}, the position $(r+1,1)$ will be combed.
Also, the region between (and including) the two half-lines extending up-left and down-left indefinitely from the position $(r+1,0)$ will be empty.
\end{lemma}

\begin{proof}
The last comb operation is on position $(r+1,1)$ or $(r+1,0)$. If it is the former, $(r+1,1)$ will be combed (Lemma~\ref{lem:combmakescombed}). If it is on the latter, as none of the comb operations will affect the site $(r,0)$ directly down-right of the last comb position, $(r,0)$ will remain occupied by a particle. As $(r+1,0)$ is combed, there will be no particles on the diagonal stretching down-left from $(r+1,0)$. Hence, $(r+1,1)$ is also combed.

The region described in the lemma can be divided into ``files'', diagonal lines going down-left. Consider any site $(x,y)$ in this region. If $(x,y)$ is in the lowest file of the region (on the diagonal extending down-left from $(r+1,0)$), as $(r+1,1)$ is combed, $(x,y)$ must be unoccupied. Otherwise, if $(x,y+1)$ is in the same file as some position in the combable sequence, let $(x_i,y_i)$ be the last position in the sequence in the same file as $(x,y+1)$. After $(x_i,y_i)$ is combed, $(x,y)$ must be empty. Subsequent combs will not affect $(x,y)$ by Lemma~\ref{lem:aboveunaffected}. If $(x,y+1)$ is not in the same file as any position in the combable sequence, $(x,y)$ must be empty as $(x_1,y_1)$ is vertically aligned with the leftmost particle of the configuration. Similarly by Lemma~\ref{lem:aboveunaffected}, none of the combs will move a particle into $(x,y)$. Hence, $(x,y)$ will be empty after the \spinecomb{} in all cases.
\end{proof}

After a \spinecomb{} is applied, there are two possible cases, having a gap in the line (defined next), and not having a gap in the line. If there is a gap in the line, we show that we can directly reduce the minimum \spine{} length of the configuration from here, giving us the result of Lemma~\ref{lem:reducespinelength}. Hence, we can proceed with the rest of the proof of Lemma~\ref{lem:reducespinelength} assuming that there will never be a gap in the line.

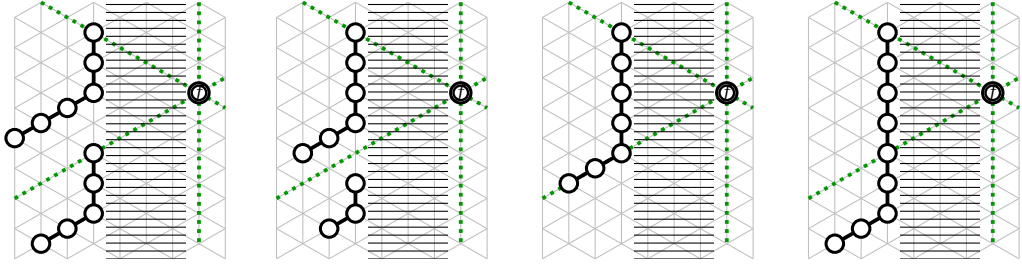
\begin{figure*}
\begin{subfigure}{.25\textwidth}
  \begin{center}
  \begin{tikzpicture}[x=0.4cm,y=0.4cm]
  \draw[lightgray] (5.19615,-9) -- (6.9282,-8);
\draw[lightgray] (0,-7) -- (6.9282,-3);
\draw[lightgray] (0,-7) -- (3.4641,-9);
\draw[lightgray] (0,-4) -- (6.06218,-0.5);
\draw[lightgray] (0,-4) -- (6.9282,-8);
\draw[lightgray] (0,-1) -- (0.866025,-0.5);
\draw[lightgray] (0,-1) -- (6.9282,-5);
\draw[lightgray] (0,-1) -- (0,-9);
\draw[lightgray] (1.73205,-1) -- (1.73205,-9);
\draw[lightgray] (3.4641,-1) -- (3.4641,-9);
\draw[lightgray] (0,-5) -- (6.9282,-1);
\draw[lightgray] (0,-5) -- (6.9282,-9);
\draw[lightgray] (0.866025,-0.5) -- (6.9282,-4);
\draw[lightgray] (0.866025,-0.5) -- (0.866025,-8.5);
\draw[lightgray] (0,-8) -- (6.9282,-4);
\draw[lightgray] (0,-8) -- (1.73205,-9);
\draw[lightgray] (5.19615,-1) -- (5.19615,-9);
\draw[lightgray] (0,-2) -- (2.59808,-0.5);
\draw[lightgray] (0,-2) -- (6.9282,-6);
\draw[lightgray] (2.59808,-0.5) -- (6.9282,-3);
\draw[lightgray] (2.59808,-0.5) -- (2.59808,-8.5);
\draw[lightgray] (4.33013,-0.5) -- (6.9282,-2);
\draw[lightgray] (4.33013,-0.5) -- (4.33013,-8.5);
\draw[lightgray] (6.06218,-0.5) -- (6.9282,-1);
\draw[lightgray] (6.06218,-0.5) -- (6.06218,-8.5);
\draw[lightgray] (0,-6) -- (6.9282,-2);
\draw[lightgray] (0,-6) -- (5.19615,-9);
\draw[lightgray] (0,-9) -- (6.9282,-5);
\draw[lightgray] (0,-3) -- (4.33013,-0.5);
\draw[lightgray] (0,-3) -- (6.9282,-7);
\draw[lightgray] (3.4641,-9) -- (6.9282,-7);
\draw[lightgray] (1.73205,-9) -- (6.9282,-6);
\draw[lightgray] (6.9282,-1) -- (6.9282,-9);
\draw[black!40!green, dotted, line width=0.5mm] (0,-7) -- (6.9282,-3);
\draw[black!40!green, dotted, line width=0.5mm] (0.866025,-0.5) -- (6.9282,-4);
\draw[black!40!green, dotted, line width=0.5mm] (6.06218,-0.5) -- (6.06218,-8.5);
\draw[black, line width=0.4mm, fill=white] (0,-5) circle (0.288);
\draw[black, line width=0.4mm, fill=white] (0.866025,-4.5) circle (0.288);
\draw[black, line width=0.4mm, fill=white] (0.866025,-8.5) circle (0.288);
\draw[black, line width=0.4mm, fill=white] (1.73205,-4) circle (0.288);
\draw[black, line width=0.4mm, fill=white] (1.73205,-8) circle (0.288);
\draw[black, line width=0.4mm, fill=white] (2.59808,-1.5) circle (0.288);
\draw[black, line width=0.4mm, fill=white] (2.59808,-2.5) circle (0.288);
\draw[black, line width=0.4mm, fill=white] (2.59808,-3.5) circle (0.288);
\draw[black, line width=0.4mm, fill=white] (2.59808,-5.5) circle (0.288);
\draw[black, line width=0.4mm, fill=white] (2.59808,-6.5) circle (0.288);
\draw[black, line width=0.4mm, fill=white] (2.59808,-7.5) circle (0.288);
\draw[black, line width=0.4mm, fill=white] (6.06218,-3.5) circle (0.336);
\draw[black, line width=0.32mm] (6.06218,-3.5) circle (0.24);
\node[align=left] at (6.06218,-3.5) {\scriptsize $f$};
\draw[black, line width=0.5mm] (2.59808,-2.2) -- (2.59808,-1.8);
\draw[black, line width=0.5mm] (2.59808,-3.2) -- (2.59808,-2.8);
\draw[black, line width=0.5mm] (1.99186,-7.85) -- (2.33827,-7.65);
\draw[black, line width=0.5mm] (1.12583,-4.35) -- (1.47224,-4.15);
\draw[black, line width=0.5mm] (1.12583,-8.35) -- (1.47224,-8.15);
\draw[black, line width=0.5mm] (0.259808,-4.85) -- (0.606218,-4.65);
\draw[black, line width=0.5mm] (2.59808,-6.2) -- (2.59808,-5.8);
\draw[black, line width=0.5mm] (2.59808,-7.2) -- (2.59808,-6.8);
\draw[black, line width=0.5mm] (1.99186,-3.85) -- (2.33827,-3.65);
\path [pattern=horizontal lines, pattern color=black] (3.03109,-0.5) -- (5.62917,-0.5) -- (5.62917,-9) -- (3.03109,-9) -- cycle;
  \end{tikzpicture}
  \end{center}
  \caption{Gap between \spines{}.}
  \label{fig:gap_yes1}
\end{subfigure}%
\begin{subfigure}{.25\textwidth}
  \begin{center}
  \begin{tikzpicture}[x=0.4cm,y=0.4cm]
  \draw[lightgray] (5.19615,-9) -- (6.9282,-8);
\draw[lightgray] (0,-7) -- (6.9282,-3);
\draw[lightgray] (0,-7) -- (3.4641,-9);
\draw[lightgray] (0,-4) -- (6.06218,-0.5);
\draw[lightgray] (0,-4) -- (6.9282,-8);
\draw[lightgray] (0,-1) -- (0.866025,-0.5);
\draw[lightgray] (0,-1) -- (6.9282,-5);
\draw[lightgray] (0,-1) -- (0,-9);
\draw[lightgray] (1.73205,-1) -- (1.73205,-9);
\draw[lightgray] (3.4641,-1) -- (3.4641,-9);
\draw[lightgray] (0,-5) -- (6.9282,-1);
\draw[lightgray] (0,-5) -- (6.9282,-9);
\draw[lightgray] (0.866025,-0.5) -- (6.9282,-4);
\draw[lightgray] (0.866025,-0.5) -- (0.866025,-8.5);
\draw[lightgray] (0,-8) -- (6.9282,-4);
\draw[lightgray] (0,-8) -- (1.73205,-9);
\draw[lightgray] (5.19615,-1) -- (5.19615,-9);
\draw[lightgray] (0,-2) -- (2.59808,-0.5);
\draw[lightgray] (0,-2) -- (6.9282,-6);
\draw[lightgray] (2.59808,-0.5) -- (6.9282,-3);
\draw[lightgray] (2.59808,-0.5) -- (2.59808,-8.5);
\draw[lightgray] (4.33013,-0.5) -- (6.9282,-2);
\draw[lightgray] (4.33013,-0.5) -- (4.33013,-8.5);
\draw[lightgray] (6.06218,-0.5) -- (6.9282,-1);
\draw[lightgray] (6.06218,-0.5) -- (6.06218,-8.5);
\draw[lightgray] (0,-6) -- (6.9282,-2);
\draw[lightgray] (0,-6) -- (5.19615,-9);
\draw[lightgray] (0,-9) -- (6.9282,-5);
\draw[lightgray] (0,-3) -- (4.33013,-0.5);
\draw[lightgray] (0,-3) -- (6.9282,-7);
\draw[lightgray] (3.4641,-9) -- (6.9282,-7);
\draw[lightgray] (1.73205,-9) -- (6.9282,-6);
\draw[lightgray] (6.9282,-1) -- (6.9282,-9);
\draw[black!40!green, dotted, line width=0.5mm] (0,-7) -- (6.9282,-3);
\draw[black!40!green, dotted, line width=0.5mm] (0.866025,-0.5) -- (6.9282,-4);
\draw[black!40!green, dotted, line width=0.5mm] (6.06218,-0.5) -- (6.06218,-8.5);
\draw[black, line width=0.4mm, fill=white] (0.866025,-5.5) circle (0.288);
\draw[black, line width=0.4mm, fill=white] (1.73205,-5) circle (0.288);
\draw[black, line width=0.4mm, fill=white] (1.73205,-8) circle (0.288);
\draw[black, line width=0.4mm, fill=white] (2.59808,-1.5) circle (0.288);
\draw[black, line width=0.4mm, fill=white] (2.59808,-2.5) circle (0.288);
\draw[black, line width=0.4mm, fill=white] (2.59808,-3.5) circle (0.288);
\draw[black, line width=0.4mm, fill=white] (2.59808,-4.5) circle (0.288);
\draw[black, line width=0.4mm, fill=white] (2.59808,-6.5) circle (0.288);
\draw[black, line width=0.4mm, fill=white] (2.59808,-7.5) circle (0.288);
\draw[black, line width=0.4mm, fill=white] (6.06218,-3.5) circle (0.336);
\draw[black, line width=0.32mm] (6.06218,-3.5) circle (0.24);
\node[align=left] at (6.06218,-3.5) {\scriptsize $f$};
\draw[black, line width=0.5mm] (2.59808,-2.2) -- (2.59808,-1.8);
\draw[black, line width=0.5mm] (2.59808,-3.2) -- (2.59808,-2.8);
\draw[black, line width=0.5mm] (1.99186,-7.85) -- (2.33827,-7.65);
\draw[black, line width=0.5mm] (1.12583,-5.35) -- (1.47224,-5.15);
\draw[black, line width=0.5mm] (2.59808,-7.2) -- (2.59808,-6.8);
\draw[black, line width=0.5mm] (1.99186,-4.85) -- (2.33827,-4.65);
\draw[black, line width=0.5mm] (2.59808,-4.2) -- (2.59808,-3.8);
\path [pattern=horizontal lines, pattern color=black] (3.03109,-0.5) -- (5.62917,-0.5) -- (5.62917,-9) -- (3.03109,-9) -- cycle;
  \end{tikzpicture}
  \end{center}
  \caption{Gap on \targetspine{}.}
  \label{fig:gap_yes2}
\end{subfigure}%
\begin{subfigure}{.25\textwidth}
  \begin{center}
  \begin{tikzpicture}[x=0.4cm,y=0.4cm]
  \draw[lightgray] (5.19615,-9) -- (6.9282,-8);
\draw[lightgray] (0,-7) -- (6.9282,-3);
\draw[lightgray] (0,-7) -- (3.4641,-9);
\draw[lightgray] (0,-4) -- (6.06218,-0.5);
\draw[lightgray] (0,-4) -- (6.9282,-8);
\draw[lightgray] (0,-1) -- (0.866025,-0.5);
\draw[lightgray] (0,-1) -- (6.9282,-5);
\draw[lightgray] (0,-1) -- (0,-9);
\draw[lightgray] (1.73205,-1) -- (1.73205,-9);
\draw[lightgray] (3.4641,-1) -- (3.4641,-9);
\draw[lightgray] (0,-5) -- (6.9282,-1);
\draw[lightgray] (0,-5) -- (6.9282,-9);
\draw[lightgray] (0.866025,-0.5) -- (6.9282,-4);
\draw[lightgray] (0.866025,-0.5) -- (0.866025,-8.5);
\draw[lightgray] (0,-8) -- (6.9282,-4);
\draw[lightgray] (0,-8) -- (1.73205,-9);
\draw[lightgray] (5.19615,-1) -- (5.19615,-9);
\draw[lightgray] (0,-2) -- (2.59808,-0.5);
\draw[lightgray] (0,-2) -- (6.9282,-6);
\draw[lightgray] (2.59808,-0.5) -- (6.9282,-3);
\draw[lightgray] (2.59808,-0.5) -- (2.59808,-8.5);
\draw[lightgray] (4.33013,-0.5) -- (6.9282,-2);
\draw[lightgray] (4.33013,-0.5) -- (4.33013,-8.5);
\draw[lightgray] (6.06218,-0.5) -- (6.9282,-1);
\draw[lightgray] (6.06218,-0.5) -- (6.06218,-8.5);
\draw[lightgray] (0,-6) -- (6.9282,-2);
\draw[lightgray] (0,-6) -- (5.19615,-9);
\draw[lightgray] (0,-9) -- (6.9282,-5);
\draw[lightgray] (0,-3) -- (4.33013,-0.5);
\draw[lightgray] (0,-3) -- (6.9282,-7);
\draw[lightgray] (3.4641,-9) -- (6.9282,-7);
\draw[lightgray] (1.73205,-9) -- (6.9282,-6);
\draw[lightgray] (6.9282,-1) -- (6.9282,-9);
\draw[black!40!green, dotted, line width=0.5mm] (0,-7) -- (6.9282,-3);
\draw[black!40!green, dotted, line width=0.5mm] (0.866025,-0.5) -- (6.9282,-4);
\draw[black!40!green, dotted, line width=0.5mm] (6.06218,-0.5) -- (6.06218,-8.5);
\draw[black, line width=0.4mm, fill=white] (0.866025,-6.5) circle (0.288);
\draw[black, line width=0.4mm, fill=white] (1.73205,-6) circle (0.288);
\draw[black, line width=0.4mm, fill=white] (2.59808,-1.5) circle (0.288);
\draw[black, line width=0.4mm, fill=white] (2.59808,-2.5) circle (0.288);
\draw[black, line width=0.4mm, fill=white] (2.59808,-3.5) circle (0.288);
\draw[black, line width=0.4mm, fill=white] (2.59808,-4.5) circle (0.288);
\draw[black, line width=0.4mm, fill=white] (2.59808,-5.5) circle (0.288);
\draw[black, line width=0.4mm, fill=white] (6.06218,-3.5) circle (0.336);
\draw[black, line width=0.32mm] (6.06218,-3.5) circle (0.24);
\node[align=left] at (6.06218,-3.5) {\scriptsize $f$};
\draw[black, line width=0.5mm] (2.59808,-2.2) -- (2.59808,-1.8);
\draw[black, line width=0.5mm] (1.99186,-5.85) -- (2.33827,-5.65);
\draw[black, line width=0.5mm] (2.59808,-3.2) -- (2.59808,-2.8);
\draw[black, line width=0.5mm] (1.12583,-6.35) -- (1.47224,-6.15);
\draw[black, line width=0.5mm] (2.59808,-4.2) -- (2.59808,-3.8);
\draw[black, line width=0.5mm] (2.59808,-5.2) -- (2.59808,-4.8);
\path [pattern=horizontal lines, pattern color=black] (3.03109,-0.5) -- (5.62917,-0.5) -- (5.62917,-9) -- (3.03109,-9) -- cycle;
  \end{tikzpicture}
  \end{center}
  \caption{No gap.}
  \label{fig:gap_no1}
\end{subfigure}%
\begin{subfigure}{.25\textwidth}
  \begin{center}
  \begin{tikzpicture}[x=0.4cm,y=0.4cm]
  \draw[lightgray] (5.19615,-9) -- (6.9282,-8);
\draw[lightgray] (0,-7) -- (6.9282,-3);
\draw[lightgray] (0,-7) -- (3.4641,-9);
\draw[lightgray] (0,-4) -- (6.06218,-0.5);
\draw[lightgray] (0,-4) -- (6.9282,-8);
\draw[lightgray] (0,-1) -- (0.866025,-0.5);
\draw[lightgray] (0,-1) -- (6.9282,-5);
\draw[lightgray] (0,-1) -- (0,-9);
\draw[lightgray] (1.73205,-1) -- (1.73205,-9);
\draw[lightgray] (3.4641,-1) -- (3.4641,-9);
\draw[lightgray] (0,-5) -- (6.9282,-1);
\draw[lightgray] (0,-5) -- (6.9282,-9);
\draw[lightgray] (0.866025,-0.5) -- (6.9282,-4);
\draw[lightgray] (0.866025,-0.5) -- (0.866025,-8.5);
\draw[lightgray] (0,-8) -- (6.9282,-4);
\draw[lightgray] (0,-8) -- (1.73205,-9);
\draw[lightgray] (5.19615,-1) -- (5.19615,-9);
\draw[lightgray] (0,-2) -- (2.59808,-0.5);
\draw[lightgray] (0,-2) -- (6.9282,-6);
\draw[lightgray] (2.59808,-0.5) -- (6.9282,-3);
\draw[lightgray] (2.59808,-0.5) -- (2.59808,-8.5);
\draw[lightgray] (4.33013,-0.5) -- (6.9282,-2);
\draw[lightgray] (4.33013,-0.5) -- (4.33013,-8.5);
\draw[lightgray] (6.06218,-0.5) -- (6.9282,-1);
\draw[lightgray] (6.06218,-0.5) -- (6.06218,-8.5);
\draw[lightgray] (0,-6) -- (6.9282,-2);
\draw[lightgray] (0,-6) -- (5.19615,-9);
\draw[lightgray] (0,-9) -- (6.9282,-5);
\draw[lightgray] (0,-3) -- (4.33013,-0.5);
\draw[lightgray] (0,-3) -- (6.9282,-7);
\draw[lightgray] (3.4641,-9) -- (6.9282,-7);
\draw[lightgray] (1.73205,-9) -- (6.9282,-6);
\draw[lightgray] (6.9282,-1) -- (6.9282,-9);
\draw[black!40!green, dotted, line width=0.5mm] (0,-7) -- (6.9282,-3);
\draw[black!40!green, dotted, line width=0.5mm] (0.866025,-0.5) -- (6.9282,-4);
\draw[black!40!green, dotted, line width=0.5mm] (6.06218,-0.5) -- (6.06218,-8.5);
\draw[black, line width=0.4mm, fill=white] (0.866025,-8.5) circle (0.288);
\draw[black, line width=0.4mm, fill=white] (1.73205,-8) circle (0.288);
\draw[black, line width=0.4mm, fill=white] (2.59808,-1.5) circle (0.288);
\draw[black, line width=0.4mm, fill=white] (2.59808,-2.5) circle (0.288);
\draw[black, line width=0.4mm, fill=white] (2.59808,-3.5) circle (0.288);
\draw[black, line width=0.4mm, fill=white] (2.59808,-4.5) circle (0.288);
\draw[black, line width=0.4mm, fill=white] (2.59808,-5.5) circle (0.288);
\draw[black, line width=0.4mm, fill=white] (2.59808,-6.5) circle (0.288);
\draw[black, line width=0.4mm, fill=white] (2.59808,-7.5) circle (0.288);
\draw[black, line width=0.4mm, fill=white] (6.06218,-3.5) circle (0.336);
\draw[black, line width=0.32mm] (6.06218,-3.5) circle (0.24);
\node[align=left] at (6.06218,-3.5) {\scriptsize $f$};
\draw[black, line width=0.5mm] (2.59808,-2.2) -- (2.59808,-1.8);
\draw[black, line width=0.5mm] (2.59808,-3.2) -- (2.59808,-2.8);
\draw[black, line width=0.5mm] (1.99186,-7.85) -- (2.33827,-7.65);
\draw[black, line width=0.5mm] (1.12583,-8.35) -- (1.47224,-8.15);
\draw[black, line width=0.5mm] (2.59808,-6.2) -- (2.59808,-5.8);
\draw[black, line width=0.5mm] (2.59808,-7.2) -- (2.59808,-6.8);
\draw[black, line width=0.5mm] (2.59808,-4.2) -- (2.59808,-3.8);
\draw[black, line width=0.5mm] (2.59808,-5.2) -- (2.59808,-4.8);
\path [pattern=horizontal lines, pattern color=black] (3.03109,-0.5) -- (5.62917,-0.5) -- (5.62917,-9) -- (3.03109,-9) -- cycle;
  \end{tikzpicture}
  \end{center}
  \caption{No gap.}
  \label{fig:gap_no2}
\end{subfigure}%
\caption{Illustration of the line between the \spines{} going up-left and down-left from the food particle. Figures~\ref{fig:gap_yes1} and \ref{fig:gap_yes2} have gaps in the line (Definition~\ref{defn:gapintheline}), while Figures~\ref{fig:gap_no1} and \ref{fig:gap_no2} do not. Observe that in the cases with no gap, the length of the \targetspine{} matches that of the \sourcespine{}.}
\end{figure*}

\begin{definition}[Gap in the line]
\label{defn:gapintheline}
Let $r$ be the length of the \sourcespine{} $S_i$. Consider the vertical line segment of sites from the location of the \anchorparticle{} of $S_i$ down to the site on the \targetspine{} $S_{i+1}$ of distance $r$ from the center, including the two \spine{} location endpoints. If there is a site on this line segment that is unoccupied by particles, we say that there is a \emph{gap} in the line from the \sourcespine{} to the \targetspine{}.
\end{definition}

\begin{lemma}[Reducing minimum \spine{} length using a gap]
\label{lem:reducingusinggap}
After a \spinecomb{} is applied from a \sourcespine{} of minimum length, if there is a gap in the line from the \sourcespine{} to the \targetspine{}, there exists a sequence of moves to reduce the minimum \spine{} length of the configuration. 
\end{lemma}

\begin{proof}
Let $r$ denote the length of the \sourcespine{} $S_i$. Suppose that there is an unoccupied site $(r,d)$ on this line segment.
If the unoccupied site on the line segment is on the \sourcespine{} $S_i$ (which actually never happens), as every site on the \spine{} of distance greater than $r$ will be unoccupied by Lemma~\ref{lem:spinecomb}, the new minimum \spine{} length would be at most $r-1$. If the unoccupied site is on the \targetspine{} $S_{i+1}$, as $(r+1,1)$ is combed, every site of the \targetspine{} of distance greater than this unoccupied site would also be unoccupied. Hence the minimum \spine{} length would also have decreased to at most $r-1$ in this case.

If the unoccupied site $(r,d)$ lies strictly between the \sourcespine{} and the \targetspine{}, we apply one more comb on position $(r,d+1)$. Position $(r,d+1)$ is combable as $(r,d)$ is empty, and $(r+1,1)$ being combed ensures that all sites on the half-line extending down-left from $(r,d)$ are also empty, so $(r+1,d+1)$ is also combed. By Lemma~\ref{lem:combingspinelength}, combing $(r,d+1)$ results in the \targetspine{} having length at most $r-1$.
\end{proof}

Thus, from now on we may assume that whenever a \spinecomb{} is executed from a \sourcespine{} of minimum length, there will be no gaps in the line between the \sourcespine{} and the \targetspine{}.
In addition, we assume that the comb operations do not cause any other \spine{} (in particular the \spine{} going downwards) to end up with a \spine{} length below the current minimum spine length, as in this case we have already achieved the result of Lemma~\ref{lem:reducespinelength}.
The following Lemma shows that such a \spinecomb{} ``pushes'' all of the particles of distance greater than the minimum \spine{} length towards the \targetspine{} ore beyond.

\begin{lemma}[Resulting configuration assuming no gap exists]
\label{lem:nogapresult}
Suppose that a \spinecomb{} is executed from a \sourcespine{} $S_i$ (of minimal length $r$) to a \targetspine{} $S_{i+1}$, and assume that there are no gaps in the line between $S_i$ and $S_{i+1}$.
In the resulting configuration, there will be no particles of distance greater than $r$ strictly between the \sourcespine{} and \targetspine{}, or on the \sourcespine{} itself. Furthermore, the lengths of both the \sourcespine{} and \targetspine{} will now be exactly $r$.
\end{lemma}

\begin{proof}
After the \spinecomb{}, position $(r+1,1)$ will be combed, and the region between the down-left and up-left diagonals extending from $(r+1,0)$ as described in Lemma~\ref{lem:spinecomb} will be empty. As there is no gap in the line from $S_i$ to $S_{i+1}$, the only lines extending left and down in the residual region of $(r+1,1)$ will be on the \targetspine{} or below, giving us the first part of this Lemma.

The length of the \sourcespine{} is $r$ as position $(r,0)$ is occupied while no position on the \sourcespine{} beyond that is. For the length of the \targetspine{}, the position on the \targetspine{} of distance $r$ from the food particle is occupied and is not a tail particle, and by Lemma~\ref{lem:reducespinelength}, $(r+1,1)$ being combed implies that the \targetspine{} has length at most $r$.
\end{proof}

As a \spinecomb{} sets the length of the \targetspine{} $S_{i+1}$ to be the same as that of the \sourcespine{} $S_i$, which has minimum length, we can continue executing \spinecombs{} in a counterclockwise fashion, from $S_{i+1}$ to $S_{i+2}$, followed by $S_{i+2}$ to $S_{i+3}$, and so on. We show that after seven of these \spinecombs{} which do not create gaps, we will reach a type of configuration we will call a \emph{hexagon with a tail}. Figure~\ref{fig:hexagon_tail} illustrates examples of these ``hexagon with a tail'' configurations, though one should note that it is not necessary for all sites on the outer hexagon to be filled.

\begin{figure*}[t]
\begin{subfigure}{.33\textwidth}
  \begin{center}
  \begin{tikzpicture}[x=0.35cm,y=0.35cm]
  \draw[lightgray] (10.3923,-1) -- (10.3923,-10);
\draw[lightgray] (0,-7) -- (11.2583,-0.5);
\draw[lightgray] (0,-7) -- (6.06218,-10.5);
\draw[lightgray] (0,-10) -- (12.1244,-3);
\draw[lightgray] (0,-10) -- (0.866025,-10.5);
\draw[lightgray] (7.79423,-0.5) -- (12.1244,-3);
\draw[lightgray] (7.79423,-0.5) -- (7.79423,-10.5);
\draw[lightgray] (12.1244,-1) -- (12.1244,-10);
\draw[lightgray] (0,-4) -- (6.06218,-0.5);
\draw[lightgray] (0,-4) -- (11.2583,-10.5);
\draw[lightgray] (9.52628,-0.5) -- (12.1244,-2);
\draw[lightgray] (9.52628,-0.5) -- (9.52628,-10.5);
\draw[lightgray] (11.2583,-10.5) -- (12.1244,-10);
\draw[lightgray] (0,-1) -- (0.866025,-0.5);
\draw[lightgray] (0,-1) -- (12.1244,-8);
\draw[lightgray] (0,-1) -- (0,-10);
\draw[lightgray] (11.2583,-0.5) -- (12.1244,-1);
\draw[lightgray] (11.2583,-0.5) -- (11.2583,-10.5);
\draw[lightgray] (1.73205,-1) -- (1.73205,-10);
\draw[lightgray] (9.52628,-10.5) -- (12.1244,-9);
\draw[lightgray] (0.866025,-10.5) -- (12.1244,-4);
\draw[lightgray] (3.4641,-1) -- (3.4641,-10);
\draw[lightgray] (0,-5) -- (7.79423,-0.5);
\draw[lightgray] (0,-5) -- (9.52628,-10.5);
\draw[lightgray] (0.866025,-0.5) -- (12.1244,-7);
\draw[lightgray] (0.866025,-0.5) -- (0.866025,-10.5);
\draw[lightgray] (0,-8) -- (12.1244,-1);
\draw[lightgray] (0,-8) -- (4.33013,-10.5);
\draw[lightgray] (4.33013,-10.5) -- (12.1244,-6);
\draw[lightgray] (5.19615,-1) -- (5.19615,-10);
\draw[lightgray] (0,-2) -- (2.59808,-0.5);
\draw[lightgray] (0,-2) -- (12.1244,-9);
\draw[lightgray] (2.59808,-0.5) -- (12.1244,-6);
\draw[lightgray] (2.59808,-0.5) -- (2.59808,-10.5);
\draw[lightgray] (2.59808,-10.5) -- (12.1244,-5);
\draw[lightgray] (4.33013,-0.5) -- (12.1244,-5);
\draw[lightgray] (4.33013,-0.5) -- (4.33013,-10.5);
\draw[lightgray] (6.06218,-0.5) -- (12.1244,-4);
\draw[lightgray] (6.06218,-0.5) -- (6.06218,-10.5);
\draw[lightgray] (0,-6) -- (9.52628,-0.5);
\draw[lightgray] (0,-6) -- (7.79423,-10.5);
\draw[lightgray] (0,-9) -- (12.1244,-2);
\draw[lightgray] (0,-9) -- (2.59808,-10.5);
\draw[lightgray] (0,-3) -- (4.33013,-0.5);
\draw[lightgray] (0,-3) -- (12.1244,-10);
\draw[lightgray] (6.9282,-1) -- (6.9282,-10);
\draw[lightgray] (6.06218,-10.5) -- (12.1244,-7);
\draw[lightgray] (7.79423,-10.5) -- (12.1244,-8);
\draw[lightgray] (8.66025,-1) -- (8.66025,-10);
\draw[black!40!green, dotted, line width=0.5mm] (0,-10) -- (12.1244,-3);
\draw[black!40!green, dotted, line width=0.5mm] (0,-1) -- (12.1244,-8);
\draw[black!40!green, dotted, line width=0.5mm] (7.79423,-0.5) -- (7.79423,-10.5);
\draw[black, line width=0.4mm, fill=white] (0.866025,-1.5) circle (0.288);
\draw[black, line width=0.4mm, fill=white] (1.73205,-2) circle (0.288);
\draw[black, line width=0.4mm, fill=white] (2.59808,-2.5) circle (0.288);
\draw[black, line width=0.4mm, fill=white] (3.4641,-3) circle (0.288);
\draw[black, line width=0.4mm, fill=white] (4.33013,-3.5) circle (0.288);
\node[align=left] at (4.50333,-2.8) {\footnotesize $p_1$};
\draw[black, line width=0.4mm, fill=white] (4.33013,-4.5) circle (0.288);
\draw[black, line width=0.4mm, fill=white] (4.33013,-5.5) circle (0.288);
\draw[black, line width=0.4mm, fill=white] (4.33013,-6.5) circle (0.288);
\draw[black, line width=0.4mm, fill=white] (4.33013,-7.5) circle (0.288);
\draw[black, line width=0.4mm, fill=white] (5.19615,-3) circle (0.288);
\draw[black, line width=0.4mm, fill=white] (5.19615,-4) circle (0.288);
\draw[black, line width=0.4mm, fill=white] (5.19615,-8) circle (0.288);
\draw[black, line width=0.4mm, fill=white] (6.06218,-2.5) circle (0.288);
\draw[black, line width=0.4mm, fill=white] (6.06218,-4.5) circle (0.288);
\draw[black, line width=0.4mm, fill=white] (6.06218,-8.5) circle (0.288);
\draw[black, line width=0.4mm, fill=white] (6.9282,-2) circle (0.288);
\draw[black, line width=0.4mm, fill=white] (6.9282,-4) circle (0.288);
\draw[black, line width=0.4mm, fill=white] (6.9282,-6) circle (0.288);
\draw[black, line width=0.4mm, fill=white] (6.9282,-7) circle (0.288);
\draw[black, line width=0.4mm, fill=white] (6.9282,-9) circle (0.288);
\draw[black, line width=0.4mm, fill=white] (7.79423,-1.5) circle (0.288);
\draw[black, line width=0.4mm, fill=white] (7.79423,-4.5) circle (0.288);
\draw[black, line width=0.4mm, fill=white] (7.79423,-5.5) circle (0.336);
\draw[black, line width=0.32mm] (7.79423,-5.5) circle (0.24);
\node[align=left] at (7.79423,-5.5) {\scriptsize $f$};
\draw[black, line width=0.4mm, fill=white] (7.79423,-6.5) circle (0.288);
\draw[black, line width=0.4mm, fill=white] (7.79423,-9.5) circle (0.288);
\draw[black, line width=0.4mm, fill=white] (8.66025,-2) circle (0.288);
\draw[black, line width=0.4mm, fill=white] (8.66025,-7) circle (0.288);
\draw[black, line width=0.4mm, fill=white] (8.66025,-9) circle (0.288);
\draw[black, line width=0.4mm, fill=white] (9.52628,-2.5) circle (0.288);
\draw[black, line width=0.4mm, fill=white] (9.52628,-8.5) circle (0.288);
\draw[black, line width=0.4mm, fill=white] (10.3923,-3) circle (0.288);
\draw[black, line width=0.4mm, fill=white] (10.3923,-8) circle (0.288);
\draw[black, line width=0.4mm, fill=white] (11.2583,-3.5) circle (0.288);
\draw[black, line width=0.4mm, fill=white] (11.2583,-4.5) circle (0.288);
\draw[black, line width=0.4mm, fill=white] (11.2583,-5.5) circle (0.288);
\draw[black, line width=0.4mm, fill=white] (11.2583,-6.5) circle (0.288);
\draw[black, line width=0.4mm, fill=white] (11.2583,-7.5) circle (0.288);
\draw[black, line width=0.5mm] (8.05404,-1.65) -- (8.40045,-1.85);
\draw[black, line width=0.5mm] (4.33013,-6.2) -- (4.33013,-5.8);
\draw[black, line width=0.5mm] (9.78609,-2.65) -- (10.1325,-2.85);
\draw[black, line width=0.5mm] (11.2583,-4.2) -- (11.2583,-3.8);
\draw[black, line width=0.5mm] (7.18801,-9.15) -- (7.53442,-9.35);
\draw[black, line width=0.5mm] (6.32199,-2.35) -- (6.6684,-2.15);
\draw[black, line width=0.5mm] (4.33013,-5.2) -- (4.33013,-4.8);
\draw[black, line width=0.5mm] (8.92006,-2.15) -- (9.26647,-2.35);
\draw[black, line width=0.5mm] (5.45596,-2.85) -- (5.80237,-2.65);
\draw[black, line width=0.5mm] (11.2583,-7.2) -- (11.2583,-6.8);
\draw[black, line width=0.5mm] (1.99186,-2.15) -- (2.33827,-2.35);
\draw[black, line width=0.5mm] (7.18801,-5.85) -- (7.53442,-5.65);
\draw[black, line width=0.5mm] (7.18801,-6.15) -- (7.53442,-6.35);
\draw[black, line width=0.5mm] (1.12583,-1.65) -- (1.47224,-1.85);
\draw[black, line width=0.5mm] (8.92006,-8.85) -- (9.26647,-8.65);
\draw[black, line width=0.5mm] (10.6521,-3.15) -- (10.9985,-3.35);
\draw[black, line width=0.5mm] (10.6521,-7.85) -- (10.9985,-7.65);
\draw[black, line width=0.5mm] (5.19615,-3.7) -- (5.19615,-3.3);
\draw[black, line width=0.5mm] (5.45596,-4.15) -- (5.80237,-4.35);
\draw[black, line width=0.5mm] (2.85788,-2.65) -- (3.20429,-2.85);
\draw[black, line width=0.5mm] (4.33013,-4.2) -- (4.33013,-3.8);
\draw[black, line width=0.5mm] (4.58993,-4.35) -- (4.93634,-4.15);
\draw[black, line width=0.5mm] (7.18801,-1.85) -- (7.53442,-1.65);
\draw[black, line width=0.5mm] (11.2583,-6.2) -- (11.2583,-5.8);
\draw[black, line width=0.5mm] (6.9282,-6.7) -- (6.9282,-6.3);
\draw[black, line width=0.5mm] (7.18801,-6.85) -- (7.53442,-6.65);
\draw[black, line width=0.5mm] (7.79423,-6.2) -- (7.79423,-5.8);
\draw[black, line width=0.5mm] (8.05404,-6.65) -- (8.40045,-6.85);
\draw[black, line width=0.5mm] (4.58993,-3.35) -- (4.93634,-3.15);
\draw[black, line width=0.5mm] (4.58993,-3.65) -- (4.93634,-3.85);
\draw[black, line width=0.5mm] (5.45596,-8.15) -- (5.80237,-8.35);
\draw[black, line width=0.5mm] (9.78609,-8.35) -- (10.1325,-8.15);
\draw[black, line width=0.5mm] (4.33013,-7.2) -- (4.33013,-6.8);
\draw[black, line width=0.5mm] (4.58993,-7.65) -- (4.93634,-7.85);
\draw[black, line width=0.5mm] (8.05404,-9.35) -- (8.40045,-9.15);
\draw[black, line width=0.5mm] (11.2583,-5.2) -- (11.2583,-4.8);
\draw[black, line width=0.5mm] (6.32199,-4.35) -- (6.6684,-4.15);
\draw[black, line width=0.5mm] (3.72391,-3.15) -- (4.07032,-3.35);
\draw[black, line width=0.5mm] (7.79423,-5.2) -- (7.79423,-4.8);
\draw[black, line width=0.5mm] (6.32199,-8.65) -- (6.6684,-8.85);
\draw[black, line width=0.5mm] (7.18801,-4.15) -- (7.53442,-4.35);
\draw[black!20!red, line width=0.36mm, ] (0.866025,-1.5) circle (0.408);
\draw[black!20!red,-{Stealth[length=1.6mm,width=2.5mm]},line width=0.7mm] (0.597558,-1.655) -- (0,-2);
\draw[black!20!red, line width=0.36mm, ] (1.73205,-2) circle (0.408);
\draw[black!20!red,-{Stealth[length=1.6mm,width=2.5mm]},line width=0.7mm] (1.46358,-2.155) -- (0.866025,-2.5);
\draw[black!20!red, line width=0.36mm, ] (2.59808,-2.5) circle (0.408);
\draw[black!20!red,-{Stealth[length=1.6mm,width=2.5mm]},line width=0.7mm] (2.32961,-2.655) -- (1.73205,-3);
\draw[black!20!red, line width=0.36mm, ] (3.4641,-3) circle (0.408);
\draw[black!20!red,-{Stealth[length=1.6mm,width=2.5mm]},line width=0.7mm] (3.19563,-3.155) -- (2.59808,-3.5);
\draw[black!20!red, line width=0.36mm, ] (4.33013,-3.5) circle (0.408);
\draw[black!20!red,-{Stealth[length=1.6mm,width=2.5mm]},line width=0.7mm] (4.06166,-3.655) -- (3.4641,-4);
  \end{tikzpicture}
  \end{center}
  \caption{Moving a corner outward.}
  \label{fig:hexagon_tail1}
\end{subfigure}%
\begin{subfigure}{.33\textwidth}
  \begin{center}
  \begin{tikzpicture}[x=0.35cm,y=0.35cm]
  \draw[lightgray] (10.3923,-1) -- (10.3923,-10);
\draw[lightgray] (0,-7) -- (11.2583,-0.5);
\draw[lightgray] (0,-7) -- (6.06218,-10.5);
\draw[lightgray] (0,-10) -- (12.1244,-3);
\draw[lightgray] (0,-10) -- (0.866025,-10.5);
\draw[lightgray] (7.79423,-0.5) -- (12.1244,-3);
\draw[lightgray] (7.79423,-0.5) -- (7.79423,-10.5);
\draw[lightgray] (12.1244,-1) -- (12.1244,-10);
\draw[lightgray] (0,-4) -- (6.06218,-0.5);
\draw[lightgray] (0,-4) -- (11.2583,-10.5);
\draw[lightgray] (9.52628,-0.5) -- (12.1244,-2);
\draw[lightgray] (9.52628,-0.5) -- (9.52628,-10.5);
\draw[lightgray] (11.2583,-10.5) -- (12.1244,-10);
\draw[lightgray] (0,-1) -- (0.866025,-0.5);
\draw[lightgray] (0,-1) -- (12.1244,-8);
\draw[lightgray] (0,-1) -- (0,-10);
\draw[lightgray] (11.2583,-0.5) -- (12.1244,-1);
\draw[lightgray] (11.2583,-0.5) -- (11.2583,-10.5);
\draw[lightgray] (1.73205,-1) -- (1.73205,-10);
\draw[lightgray] (9.52628,-10.5) -- (12.1244,-9);
\draw[lightgray] (0.866025,-10.5) -- (12.1244,-4);
\draw[lightgray] (3.4641,-1) -- (3.4641,-10);
\draw[lightgray] (0,-5) -- (7.79423,-0.5);
\draw[lightgray] (0,-5) -- (9.52628,-10.5);
\draw[lightgray] (0.866025,-0.5) -- (12.1244,-7);
\draw[lightgray] (0.866025,-0.5) -- (0.866025,-10.5);
\draw[lightgray] (0,-8) -- (12.1244,-1);
\draw[lightgray] (0,-8) -- (4.33013,-10.5);
\draw[lightgray] (4.33013,-10.5) -- (12.1244,-6);
\draw[lightgray] (5.19615,-1) -- (5.19615,-10);
\draw[lightgray] (0,-2) -- (2.59808,-0.5);
\draw[lightgray] (0,-2) -- (12.1244,-9);
\draw[lightgray] (2.59808,-0.5) -- (12.1244,-6);
\draw[lightgray] (2.59808,-0.5) -- (2.59808,-10.5);
\draw[lightgray] (2.59808,-10.5) -- (12.1244,-5);
\draw[lightgray] (4.33013,-0.5) -- (12.1244,-5);
\draw[lightgray] (4.33013,-0.5) -- (4.33013,-10.5);
\draw[lightgray] (6.06218,-0.5) -- (12.1244,-4);
\draw[lightgray] (6.06218,-0.5) -- (6.06218,-10.5);
\draw[lightgray] (0,-6) -- (9.52628,-0.5);
\draw[lightgray] (0,-6) -- (7.79423,-10.5);
\draw[lightgray] (0,-9) -- (12.1244,-2);
\draw[lightgray] (0,-9) -- (2.59808,-10.5);
\draw[lightgray] (0,-3) -- (4.33013,-0.5);
\draw[lightgray] (0,-3) -- (12.1244,-10);
\draw[lightgray] (6.9282,-1) -- (6.9282,-10);
\draw[lightgray] (6.06218,-10.5) -- (12.1244,-7);
\draw[lightgray] (7.79423,-10.5) -- (12.1244,-8);
\draw[lightgray] (8.66025,-1) -- (8.66025,-10);
\draw[black!40!green, dotted, line width=0.5mm] (0,-10) -- (12.1244,-3);
\draw[black!40!green, dotted, line width=0.5mm] (0,-1) -- (12.1244,-8);
\draw[black!40!green, dotted, line width=0.5mm] (7.79423,-0.5) -- (7.79423,-10.5);
\draw[black, line width=0.4mm, fill=white] (0.866025,-1.5) circle (0.288);
\draw[black, line width=0.4mm, fill=white] (1.73205,-2) circle (0.288);
\draw[black, line width=0.4mm, fill=white] (2.59808,-2.5) circle (0.288);
\draw[black, line width=0.4mm, fill=white] (3.4641,-3) circle (0.288);
\draw[black, line width=0.4mm, fill=white] (4.33013,-3.5) circle (0.288);
\draw[black, line width=0.4mm, fill=white] (4.33013,-4.5) circle (0.288);
\draw[black, line width=0.4mm, fill=white] (4.33013,-5.5) circle (0.288);
\node[align=left] at (3.5074,-5.5) {\footnotesize $p_2$};
\draw[black, line width=0.4mm, fill=white] (4.33013,-6.5) circle (0.288);
\draw[black, line width=0.4mm, fill=white] (4.33013,-7.5) circle (0.288);
\draw[black, line width=0.4mm, fill=white] (5.19615,-3) circle (0.288);
\draw[black, line width=0.4mm, fill=white] (5.19615,-5) circle (0.288);
\draw[black, line width=0.4mm, fill=white] (5.19615,-8) circle (0.288);
\draw[black, line width=0.4mm, fill=white] (6.06218,-2.5) circle (0.288);
\draw[black, line width=0.4mm, fill=white] (6.06218,-5.5) circle (0.288);
\draw[black, line width=0.4mm, fill=white] (6.06218,-8.5) circle (0.288);
\draw[black, line width=0.4mm, fill=white] (6.9282,-2) circle (0.288);
\draw[black, line width=0.4mm, fill=white] (6.9282,-5) circle (0.288);
\draw[black, line width=0.4mm, fill=white] (6.9282,-9) circle (0.288);
\draw[black, line width=0.4mm, fill=white] (7.79423,-1.5) circle (0.288);
\draw[black, line width=0.4mm, fill=white] (7.79423,-3.5) circle (0.288);
\draw[black, line width=0.4mm, fill=white] (7.79423,-4.5) circle (0.288);
\draw[black, line width=0.4mm, fill=white] (7.79423,-5.5) circle (0.336);
\draw[black, line width=0.32mm] (7.79423,-5.5) circle (0.24);
\node[align=left] at (7.79423,-5.5) {\scriptsize $f$};
\draw[black, line width=0.4mm, fill=white] (7.79423,-6.5) circle (0.288);
\draw[black, line width=0.4mm, fill=white] (7.79423,-9.5) circle (0.288);
\draw[black, line width=0.4mm, fill=white] (8.66025,-2) circle (0.288);
\draw[black, line width=0.4mm, fill=white] (8.66025,-4) circle (0.288);
\draw[black, line width=0.4mm, fill=white] (8.66025,-5) circle (0.288);
\draw[black, line width=0.4mm, fill=white] (8.66025,-9) circle (0.288);
\draw[black, line width=0.4mm, fill=white] (9.52628,-2.5) circle (0.288);
\draw[black, line width=0.4mm, fill=white] (9.52628,-8.5) circle (0.288);
\draw[black, line width=0.4mm, fill=white] (10.3923,-3) circle (0.288);
\draw[black, line width=0.4mm, fill=white] (10.3923,-8) circle (0.288);
\draw[black, line width=0.4mm, fill=white] (11.2583,-3.5) circle (0.288);
\draw[black, line width=0.4mm, fill=white] (11.2583,-4.5) circle (0.288);
\draw[black, line width=0.4mm, fill=white] (11.2583,-5.5) circle (0.288);
\draw[black, line width=0.4mm, fill=white] (11.2583,-6.5) circle (0.288);
\draw[black, line width=0.4mm, fill=white] (11.2583,-7.5) circle (0.288);
\draw[black, line width=0.5mm] (8.05404,-1.65) -- (8.40045,-1.85);
\draw[black, line width=0.5mm] (4.33013,-6.2) -- (4.33013,-5.8);
\draw[black, line width=0.5mm] (9.78609,-2.65) -- (10.1325,-2.85);
\draw[black, line width=0.5mm] (11.2583,-4.2) -- (11.2583,-3.8);
\draw[black, line width=0.5mm] (7.18801,-9.15) -- (7.53442,-9.35);
\draw[black, line width=0.5mm] (7.79423,-4.2) -- (7.79423,-3.8);
\draw[black, line width=0.5mm] (8.05404,-4.35) -- (8.40045,-4.15);
\draw[black, line width=0.5mm] (8.05404,-4.65) -- (8.40045,-4.85);
\draw[black, line width=0.5mm] (7.18801,-4.85) -- (7.53442,-4.65);
\draw[black, line width=0.5mm] (7.18801,-5.15) -- (7.53442,-5.35);
\draw[black, line width=0.5mm] (6.32199,-2.35) -- (6.6684,-2.15);
\draw[black, line width=0.5mm] (4.33013,-5.2) -- (4.33013,-4.8);
\draw[black, line width=0.5mm] (4.58993,-5.35) -- (4.93634,-5.15);
\draw[black, line width=0.5mm] (5.45596,-2.85) -- (5.80237,-2.65);
\draw[black, line width=0.5mm] (11.2583,-7.2) -- (11.2583,-6.8);
\draw[black, line width=0.5mm] (1.99186,-2.15) -- (2.33827,-2.35);
\draw[black, line width=0.5mm] (1.12583,-1.65) -- (1.47224,-1.85);
\draw[black, line width=0.5mm] (8.92006,-8.85) -- (9.26647,-8.65);
\draw[black, line width=0.5mm] (10.6521,-3.15) -- (10.9985,-3.35);
\draw[black, line width=0.5mm] (10.6521,-7.85) -- (10.9985,-7.65);
\draw[black, line width=0.5mm] (2.85788,-2.65) -- (3.20429,-2.85);
\draw[black, line width=0.5mm] (4.33013,-4.2) -- (4.33013,-3.8);
\draw[black, line width=0.5mm] (4.58993,-4.65) -- (4.93634,-4.85);
\draw[black, line width=0.5mm] (7.18801,-1.85) -- (7.53442,-1.65);
\draw[black, line width=0.5mm] (8.05404,-3.65) -- (8.40045,-3.85);
\draw[black, line width=0.5mm] (11.2583,-6.2) -- (11.2583,-5.8);
\draw[black, line width=0.5mm] (6.32199,-5.35) -- (6.6684,-5.15);
\draw[black, line width=0.5mm] (7.79423,-6.2) -- (7.79423,-5.8);
\draw[black, line width=0.5mm] (5.45596,-5.15) -- (5.80237,-5.35);
\draw[black, line width=0.5mm] (4.58993,-3.35) -- (4.93634,-3.15);
\draw[black, line width=0.5mm] (8.66025,-4.7) -- (8.66025,-4.3);
\draw[black, line width=0.5mm] (5.45596,-8.15) -- (5.80237,-8.35);
\draw[black, line width=0.5mm] (9.78609,-8.35) -- (10.1325,-8.15);
\draw[black, line width=0.5mm] (4.33013,-7.2) -- (4.33013,-6.8);
\draw[black, line width=0.5mm] (4.58993,-7.65) -- (4.93634,-7.85);
\draw[black, line width=0.5mm] (8.05404,-9.35) -- (8.40045,-9.15);
\draw[black, line width=0.5mm] (11.2583,-5.2) -- (11.2583,-4.8);
\draw[black, line width=0.5mm] (3.72391,-3.15) -- (4.07032,-3.35);
\draw[black, line width=0.5mm] (7.79423,-5.2) -- (7.79423,-4.8);
\draw[black, line width=0.5mm] (8.05404,-5.35) -- (8.40045,-5.15);
\draw[black, line width=0.5mm] (6.32199,-8.65) -- (6.6684,-8.85);
\draw[black, line width=0.5mm] (8.92006,-2.15) -- (9.26647,-2.35);
\draw[black!20!red, line width=0.36mm, ] (4.33013,-5.5) circle (0.408);
\draw[black!20!red,-{Stealth[length=1.6mm,width=2.5mm]},line width=0.7mm] (4.59859,-5.655) -- (5.19615,-6);
  \end{tikzpicture}
  \end{center}
  \caption{Moving a side particle inward.}
  \label{fig:hexagon_tail2}
\end{subfigure}%
\begin{subfigure}{.33\textwidth}
  \begin{center}
  \begin{tikzpicture}[x=0.35cm,y=0.35cm]
  \draw[lightgray] (10.3923,-1) -- (10.3923,-10);
\draw[lightgray] (0,-7) -- (11.2583,-0.5);
\draw[lightgray] (0,-7) -- (6.06218,-10.5);
\draw[lightgray] (0,-10) -- (12.1244,-3);
\draw[lightgray] (0,-10) -- (0.866025,-10.5);
\draw[lightgray] (7.79423,-0.5) -- (12.1244,-3);
\draw[lightgray] (7.79423,-0.5) -- (7.79423,-10.5);
\draw[lightgray] (12.1244,-1) -- (12.1244,-10);
\draw[lightgray] (0,-4) -- (6.06218,-0.5);
\draw[lightgray] (0,-4) -- (11.2583,-10.5);
\draw[lightgray] (9.52628,-0.5) -- (12.1244,-2);
\draw[lightgray] (9.52628,-0.5) -- (9.52628,-10.5);
\draw[lightgray] (11.2583,-10.5) -- (12.1244,-10);
\draw[lightgray] (0,-1) -- (0.866025,-0.5);
\draw[lightgray] (0,-1) -- (12.1244,-8);
\draw[lightgray] (0,-1) -- (0,-10);
\draw[lightgray] (11.2583,-0.5) -- (12.1244,-1);
\draw[lightgray] (11.2583,-0.5) -- (11.2583,-10.5);
\draw[lightgray] (1.73205,-1) -- (1.73205,-10);
\draw[lightgray] (9.52628,-10.5) -- (12.1244,-9);
\draw[lightgray] (0.866025,-10.5) -- (12.1244,-4);
\draw[lightgray] (3.4641,-1) -- (3.4641,-10);
\draw[lightgray] (0,-5) -- (7.79423,-0.5);
\draw[lightgray] (0,-5) -- (9.52628,-10.5);
\draw[lightgray] (0.866025,-0.5) -- (12.1244,-7);
\draw[lightgray] (0.866025,-0.5) -- (0.866025,-10.5);
\draw[lightgray] (0,-8) -- (12.1244,-1);
\draw[lightgray] (0,-8) -- (4.33013,-10.5);
\draw[lightgray] (4.33013,-10.5) -- (12.1244,-6);
\draw[lightgray] (5.19615,-1) -- (5.19615,-10);
\draw[lightgray] (0,-2) -- (2.59808,-0.5);
\draw[lightgray] (0,-2) -- (12.1244,-9);
\draw[lightgray] (2.59808,-0.5) -- (12.1244,-6);
\draw[lightgray] (2.59808,-0.5) -- (2.59808,-10.5);
\draw[lightgray] (2.59808,-10.5) -- (12.1244,-5);
\draw[lightgray] (4.33013,-0.5) -- (12.1244,-5);
\draw[lightgray] (4.33013,-0.5) -- (4.33013,-10.5);
\draw[lightgray] (6.06218,-0.5) -- (12.1244,-4);
\draw[lightgray] (6.06218,-0.5) -- (6.06218,-10.5);
\draw[lightgray] (0,-6) -- (9.52628,-0.5);
\draw[lightgray] (0,-6) -- (7.79423,-10.5);
\draw[lightgray] (0,-9) -- (12.1244,-2);
\draw[lightgray] (0,-9) -- (2.59808,-10.5);
\draw[lightgray] (0,-3) -- (4.33013,-0.5);
\draw[lightgray] (0,-3) -- (12.1244,-10);
\draw[lightgray] (6.9282,-1) -- (6.9282,-10);
\draw[lightgray] (6.06218,-10.5) -- (12.1244,-7);
\draw[lightgray] (7.79423,-10.5) -- (12.1244,-8);
\draw[lightgray] (8.66025,-1) -- (8.66025,-10);
\draw[black!40!green, dotted, line width=0.5mm] (0,-10) -- (12.1244,-3);
\draw[black!40!green, dotted, line width=0.5mm] (0,-1) -- (12.1244,-8);
\draw[black!40!green, dotted, line width=0.5mm] (7.79423,-0.5) -- (7.79423,-10.5);
\draw[black, line width=0.4mm, fill=white] (0.866025,-1.5) circle (0.288);
\draw[black, line width=0.4mm, fill=white] (1.73205,-2) circle (0.288);
\draw[black, line width=0.4mm, fill=white] (2.59808,-2.5) circle (0.288);
\draw[black, line width=0.4mm, fill=white] (3.4641,-3) circle (0.288);
\draw[black, line width=0.4mm, fill=white] (4.33013,-3.5) circle (0.288);
\draw[black, line width=0.4mm, fill=white] (4.33013,-4.5) circle (0.288);
\draw[black, line width=0.4mm, fill=white] (4.33013,-5.5) circle (0.288);
\node[align=left] at (3.72391,-5.95) {\footnotesize $p_3$};
\draw[black, line width=0.4mm, fill=white] (4.33013,-6.5) circle (0.288);
\draw[black, line width=0.4mm, fill=white] (4.33013,-7.5) circle (0.288);
\draw[black, line width=0.4mm, fill=white] (5.19615,-3) circle (0.288);
\draw[black, line width=0.4mm, fill=white] (5.19615,-5) circle (0.288);
\draw[black, line width=0.4mm, fill=white] (5.19615,-6) circle (0.288);
\draw[black, line width=0.4mm, fill=white] (5.19615,-8) circle (0.288);
\draw[black, line width=0.4mm, fill=white] (6.06218,-2.5) circle (0.288);
\draw[black, line width=0.4mm, fill=white] (6.06218,-4.5) circle (0.288);
\draw[black, line width=0.4mm, fill=white] (6.06218,-8.5) circle (0.288);
\draw[black, line width=0.4mm, fill=white] (6.9282,-2) circle (0.288);
\draw[black, line width=0.4mm, fill=white] (6.9282,-5) circle (0.288);
\draw[black, line width=0.4mm, fill=white] (6.9282,-9) circle (0.288);
\draw[black, line width=0.4mm, fill=white] (7.79423,-1.5) circle (0.288);
\draw[black, line width=0.4mm, fill=white] (7.79423,-3.5) circle (0.288);
\draw[black, line width=0.4mm, fill=white] (7.79423,-4.5) circle (0.288);
\draw[black, line width=0.4mm, fill=white] (7.79423,-5.5) circle (0.336);
\draw[black, line width=0.32mm] (7.79423,-5.5) circle (0.24);
\node[align=left] at (7.79423,-5.5) {\scriptsize $f$};
\draw[black, line width=0.4mm, fill=white] (7.79423,-6.5) circle (0.288);
\draw[black, line width=0.4mm, fill=white] (7.79423,-9.5) circle (0.288);
\draw[black, line width=0.4mm, fill=white] (8.66025,-2) circle (0.288);
\draw[black, line width=0.4mm, fill=white] (8.66025,-4) circle (0.288);
\draw[black, line width=0.4mm, fill=white] (8.66025,-5) circle (0.288);
\draw[black, line width=0.4mm, fill=white] (8.66025,-9) circle (0.288);
\draw[black, line width=0.4mm, fill=white] (9.52628,-2.5) circle (0.288);
\draw[black, line width=0.4mm, fill=white] (9.52628,-8.5) circle (0.288);
\draw[black, line width=0.4mm, fill=white] (10.3923,-3) circle (0.288);
\draw[black, line width=0.4mm, fill=white] (10.3923,-8) circle (0.288);
\draw[black, line width=0.4mm, fill=white] (11.2583,-3.5) circle (0.288);
\draw[black, line width=0.4mm, fill=white] (11.2583,-4.5) circle (0.288);
\draw[black, line width=0.4mm, fill=white] (11.2583,-5.5) circle (0.288);
\draw[black, line width=0.4mm, fill=white] (11.2583,-6.5) circle (0.288);
\draw[black, line width=0.4mm, fill=white] (11.2583,-7.5) circle (0.288);
\draw[black, line width=0.5mm] (8.05404,-1.65) -- (8.40045,-1.85);
\draw[black, line width=0.5mm] (5.19615,-5.7) -- (5.19615,-5.3);
\draw[black, line width=0.5mm] (4.33013,-6.2) -- (4.33013,-5.8);
\draw[black, line width=0.5mm] (4.58993,-6.35) -- (4.93634,-6.15);
\draw[black, line width=0.5mm] (9.78609,-2.65) -- (10.1325,-2.85);
\draw[black, line width=0.5mm] (11.2583,-4.2) -- (11.2583,-3.8);
\draw[black, line width=0.5mm] (7.18801,-9.15) -- (7.53442,-9.35);
\draw[black, line width=0.5mm] (7.79423,-4.2) -- (7.79423,-3.8);
\draw[black, line width=0.5mm] (8.05404,-4.35) -- (8.40045,-4.15);
\draw[black, line width=0.5mm] (8.05404,-4.65) -- (8.40045,-4.85);
\draw[black, line width=0.5mm] (7.18801,-4.85) -- (7.53442,-4.65);
\draw[black, line width=0.5mm] (7.18801,-5.15) -- (7.53442,-5.35);
\draw[black, line width=0.5mm] (6.32199,-2.35) -- (6.6684,-2.15);
\draw[black, line width=0.5mm] (4.33013,-5.2) -- (4.33013,-4.8);
\draw[black, line width=0.5mm] (4.58993,-5.35) -- (4.93634,-5.15);
\draw[black, line width=0.5mm] (4.58993,-5.65) -- (4.93634,-5.85);
\draw[black, line width=0.5mm] (5.45596,-2.85) -- (5.80237,-2.65);
\draw[black, line width=0.5mm] (11.2583,-7.2) -- (11.2583,-6.8);
\draw[black, line width=0.5mm] (1.99186,-2.15) -- (2.33827,-2.35);
\draw[black, line width=0.5mm] (1.12583,-1.65) -- (1.47224,-1.85);
\draw[black, line width=0.5mm] (8.92006,-8.85) -- (9.26647,-8.65);
\draw[black, line width=0.5mm] (10.6521,-3.15) -- (10.9985,-3.35);
\draw[black, line width=0.5mm] (10.6521,-7.85) -- (10.9985,-7.65);
\draw[black, line width=0.5mm] (2.85788,-2.65) -- (3.20429,-2.85);
\draw[black, line width=0.5mm] (4.33013,-4.2) -- (4.33013,-3.8);
\draw[black, line width=0.5mm] (4.58993,-4.65) -- (4.93634,-4.85);
\draw[black, line width=0.5mm] (7.18801,-1.85) -- (7.53442,-1.65);
\draw[black, line width=0.5mm] (8.05404,-3.65) -- (8.40045,-3.85);
\draw[black, line width=0.5mm] (11.2583,-6.2) -- (11.2583,-5.8);
\draw[black, line width=0.5mm] (7.79423,-6.2) -- (7.79423,-5.8);
\draw[black, line width=0.5mm] (5.45596,-4.85) -- (5.80237,-4.65);
\draw[black, line width=0.5mm] (4.58993,-3.35) -- (4.93634,-3.15);
\draw[black, line width=0.5mm] (8.66025,-4.7) -- (8.66025,-4.3);
\draw[black, line width=0.5mm] (5.45596,-8.15) -- (5.80237,-8.35);
\draw[black, line width=0.5mm] (9.78609,-8.35) -- (10.1325,-8.15);
\draw[black, line width=0.5mm] (4.33013,-7.2) -- (4.33013,-6.8);
\draw[black, line width=0.5mm] (4.58993,-7.65) -- (4.93634,-7.85);
\draw[black, line width=0.5mm] (8.05404,-9.35) -- (8.40045,-9.15);
\draw[black, line width=0.5mm] (11.2583,-5.2) -- (11.2583,-4.8);
\draw[black, line width=0.5mm] (6.32199,-4.65) -- (6.6684,-4.85);
\draw[black, line width=0.5mm] (3.72391,-3.15) -- (4.07032,-3.35);
\draw[black, line width=0.5mm] (7.79423,-5.2) -- (7.79423,-4.8);
\draw[black, line width=0.5mm] (8.05404,-5.35) -- (8.40045,-5.15);
\draw[black, line width=0.5mm] (6.32199,-8.65) -- (6.6684,-8.85);
\draw[black, line width=0.5mm] (8.92006,-2.15) -- (9.26647,-2.35);
\draw[black!20!red, line width=0.36mm, ] (4.33013,-5.5) circle (0.408);
\draw[black!20!red,-{Stealth[length=1.6mm,width=2.5mm]},line width=0.7mm] (4.06166,-5.345) -- (3.4641,-5);
  \end{tikzpicture}
  \end{center}
  \caption{Moving a side particle outward.}
  \label{fig:hexagon_tail3}
\end{subfigure}%
\caption{Possible cases for reducing the minimum \spine{} length from a hexagon with a tail.}
\label{fig:hexagon_tail}
\end{figure*}
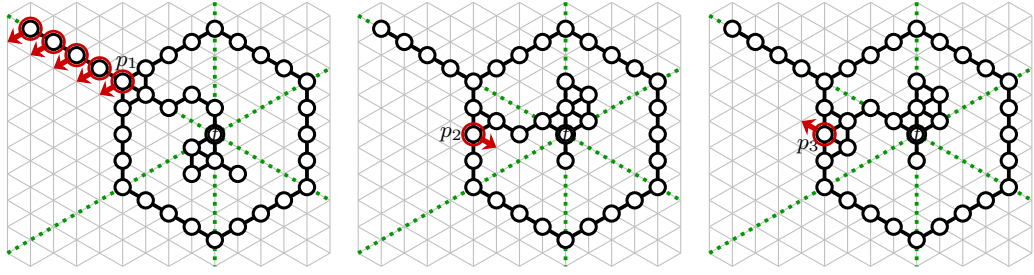

\begin{definition}[Hexagon with a Tail]
We say a configuration forms a has the ``hexagon with a tail'' arrangement of radius $r$ if:
\begin{itemize}
\item All \spines{} have length exactly $r$.
\item There are tail particles on at most one of the \spines{}.
\item Aside from these tail particles, there are no particles of distance greater than $r$ from the center.
\end{itemize}
If $r = 0$, this ``regular hexagon'' comprises of only the food particle. In other words, a hexagon with a tail of radius $r$ has all of the particles extending in a straight line from the food particle.
\end{definition}

\begin{lemma}[Reaching a Hexagon with a Tail]
\label{lem:reachinghexagon}
After seven \spinecombs{} in a counterclockwise order starting from a \spine{} of minimum length $r$, assuming that no gaps in the lines are formed and that no \spine{} ends up with length below $r$ in the process, we will end up with a hexagon with a tail arrangement of radius $r$.
\end{lemma}

\begin{proof}
We will denote the starting \spine{} as $S_0$, and name the remaining spines $S_1$ to $S_5$ in counterclockwise order. The \spinecombs{} hence go from $S_0$ to $S_1$, from $S_1$ to $S_2$ and so on, with the final (seventh) comb being from $S_0$ to $S_1$. \Spine{} $S_0$ is assumed to be a minimum length \spine{}, of length $r$.

By Lemma~\ref{lem:aboveunaffected}, a \spinecomb{} from spines $S_i$ to $S_{i+1}$ will only affect particles on spines $S_i$, $S_{i+1}$, $S_{i+2}$, and the particles between spines $S_i$ and $S_{i+1}$, between spines $S_{i+1}$ and $S_{i+2}$, and between spines $S_{i+2}$ and $S_{i+3}$. Note that this does include the particles on spine $S_{i+3}$. Hence, the first four \spinecombs{} will not affect the result of the first \spinecomb{} from $S_0$ to $S_1$. 

On the fifth \spinecomb{} from $S_4$ to $S_5$, as usual without loss of generality we take $S_4$ to be the \spine{} going up-left and $S_5$ to be the \spine{} going down-left. \Spine{} $S_0$ will thus be going downwards and \spine{} $S_1$ will be going down-right. Due to the effects of the first three combs, there will be no particle further right than the \anchorparticle{} of \spine{} $S_1$. By Lemma~\ref{lem:rightmostextent}, while the fifth \spinecomb{} may move particles onto \spine{} $S_0$ or the region between \spines{} $S_0$ and $S_1$, none of these particles in the resulting configuration will be further right than the \anchorparticle{} of \spine $S_1$.

On the sixth \spinecomb{} from $S_5$ to $S_0$, taking $S_5$ to be going up-left and $S_0$ to be going down-left, consider the position $(-r-1,0)$, which is one particle down-right of the \anchorparticle{} of the down-right \spine{} $S_2$. The region $R_{-r-1,0}$, as defined in Lemma~\ref{lem:unenterableregion}, will be empty after the fifth \spinecomb{}, due to what we have just shown to happen after the fifth \spinecomb{}. By Lemma~\ref{lem:unenterableregion}, this region will continue to be empty after the sixth \spinecomb{}.

On the seventh and final \spinecomb{} from $S_0$ to $S_1$, take $S_0$ to be going up-left and $S_1$ to be going down-left. Consider the positions $(r,0)$ and $(r,r)$, which are on the \sourcespine{} $S_0$ and \targetspine{} $S_1$ respectively, of distance $r$ from the center. As a result of the sixth \spinecomb{} with Lemma~\ref{lem:nogapresult}, all particles of distance greater than $r$ from the center must lie between (inclusive) the two diagonal lines going up-left from the positions $(r,0)$ and $(r,r)$. Now, from the position $(r,r+1)$ which lies directly below $(r,r)$ and the position $(-r-1,0)$, which lies on \spine{} $S_3$ of distance $r+1$ from the center, we consider the two regions $R_{r,r+1}$ and $R_{-r-1,0}$ as in Lemma~\ref{lem:unenterableregion}. Both of these regions are initially empty, and so will remain empty after the seventh comb. By Lemmas~\ref{lem:aboveunaffected} and \ref{lem:nogapresult}, the only place where particles of distance greater than $r$ can be are on the \targetspine{} $S_1$.

As we had assumed that no \spine{} will have ended up with length less than $r$ in the process, this means we have reached a hexagon with a tail arrangement of radius $r$.
\end{proof}

The following Lemma then concludes the proof that we can always reduce the minimum \spine{} length, provided that the current minimum \spine{} length is at least $1$.

\begin{lemma}
From a hexagon with a tail arrangement of radius $r \geq 1$, there exists a sequence of moves to reduce the minimum \spine{} length by $1$.
\end{lemma}

\begin{proof}
Consider the set $H$ of positions of distance exactly $r$ from the center. This set of positions is in the shape of a hexagon. If one of these sites is unfilled, without loss of generality assume this site $(r,d)$ is on the left side of the hexagon (it is not on a corner as all \spines{} have length $r$). The site $(r,d+1)$ is combable, which by Lemma~\ref{lem:combingspinelength} gives us a way to reduce the length of the \spine{} going down-left to at most $r-1$.

If no such gap in $H$ currently exists, we show that we can create such a gap.
If $r=1$, pick any particle on the hexagon $H$ aside from the one on the \spine{} the tail is on. This particle can be moved to a vacant spot between two \spines{}, reducing the minimum \spine{} length to $0$.

Otherwise, as the configuration is assumed to be connected, there is a path of particles from the center (food) particle to a particle on $H$. This implies that there is a particle of $v_{-2}$ distance $r-2$ from the center adjacent to a particle $v_{-1}$ of distance $r-1$ from the center. Note that if $r = 2$, $v_{-2}$ will be the food particle. If $v_{-1}$ is adjacent to a corner particle of the hexagon $H$, assuming without loss of generality that this corner is on the \spine{} going up-left, we can move this corner particle one step down-left, and if there are any tail particles attached to this corner particle, they can then subsequently be moved one-by-one one step down-left as well (Figure~\ref{fig:hexagon_tail1}). This reduces the minimum \spine{} length to at most $r-1$.

If $v_{-1}$ is not adjacent to a corner particle of $H$, we note that $v_{-1}$ and $v_{-2}$ will share a neighboring site $u_{-1}$ of distance $r-1$ from the center. The sites $u_{-1}$ and $v_{-1}$ share a neighbor particle $v_0$ on $H$. If site $u_{-1}$ is unoccupied, particle $v_0$ can be moved into site $u_{-1}$, creating a gap in the hexagon $H$ (Figure~\ref{fig:hexagon_tail2}). If $u_{-1}$ is occupied, $v_0$ can be moved in the opposite direction of $u_{-1}$, to a position $u_{+1}$ of distance $r+1$ from the center, creating a gap in the hexagon $H$ (Figure~\ref{fig:hexagon_tail3}).

In both of these cases, by reflection and rotational symmetry, without loss of generality, this newly created gap $v_0$ is on the left wall of the hexagon $H$, and if the particle was moved to $u_{+1}$, $u_{+1}$ is directly up-left of $v_0$. The site directly below $v_0$ is thus combable, and by Lemma~\ref{lem:combingspinelength}, combing this reduces the minimum \spine{} length to at most $r-1$.
\end{proof}

Finally, we show that we can reach a straight line of particles, thus showing ergodicity of the chain since the chain is reversible.
\begin{lemma}
From any connected configuration of particles with one single immobile (food) particle, there exists a sequence of valid moves to transform this configuration into a straight line of particles with the food particle at one end.
\end{lemma}

\begin{proof}
Applying Lemma~\ref{lem:reducespinelength} repeatedly allows us to arrive at a configuration with minimum \spine{} length $0$. 
Applying Lemma~\ref{lem:reachinghexagon} from here gives us a hexagon with a tail arrangement of radius $0$, which is a straight line of particles with the food particle at one end.
\end{proof}

We observe that the direction in which the final tail faces is irrelevant, as there is a simple sequence of moves to change the direction of the tail, by moving the particles one by one to the location of the new tail, starting from the particle at the very end of the initial tail. This thus allows us to conclude Lemma~\ref{lemma:irreducible}, which also implies that the Markov chain is irreducible.


\section{Simulation of the Adaptive $\alpha$-Compression Algorithm with Multiple Food Particles}
\label{apx:multiplefood}
In this appendix we showcase a simulation of the Adaptive $\alpha$-Compression algorithm with $2560$ particles in a $80\times 80$ triangular lattice with periodic boundary conditions. Six food particles are placed and moved around to illustrate the gather and search phases. This simulation is shown as a sequence of images in chronological order, spanning Figures~\ref{fig:simulation1} and~\ref{fig:simulation2}.
\begin{figure}[!ht]
\begin{subfigure}[b]{.5\linewidth}
  \begin{center}
  \includegraphics[width=\linewidth]{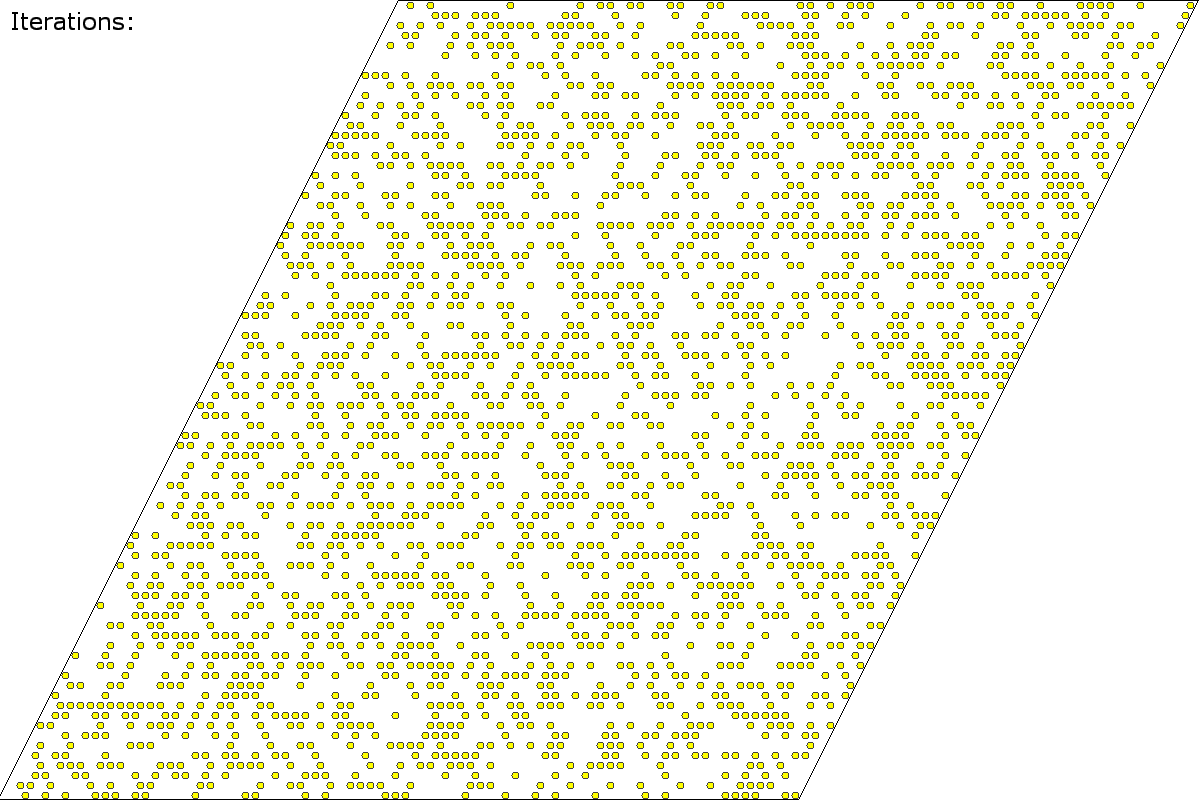}
  \end{center}
  \caption{All particles in dispersion state.}
\end{subfigure}%
\begin{subfigure}[b]{.5\linewidth}
  \begin{center}
  \includegraphics[width=\linewidth]{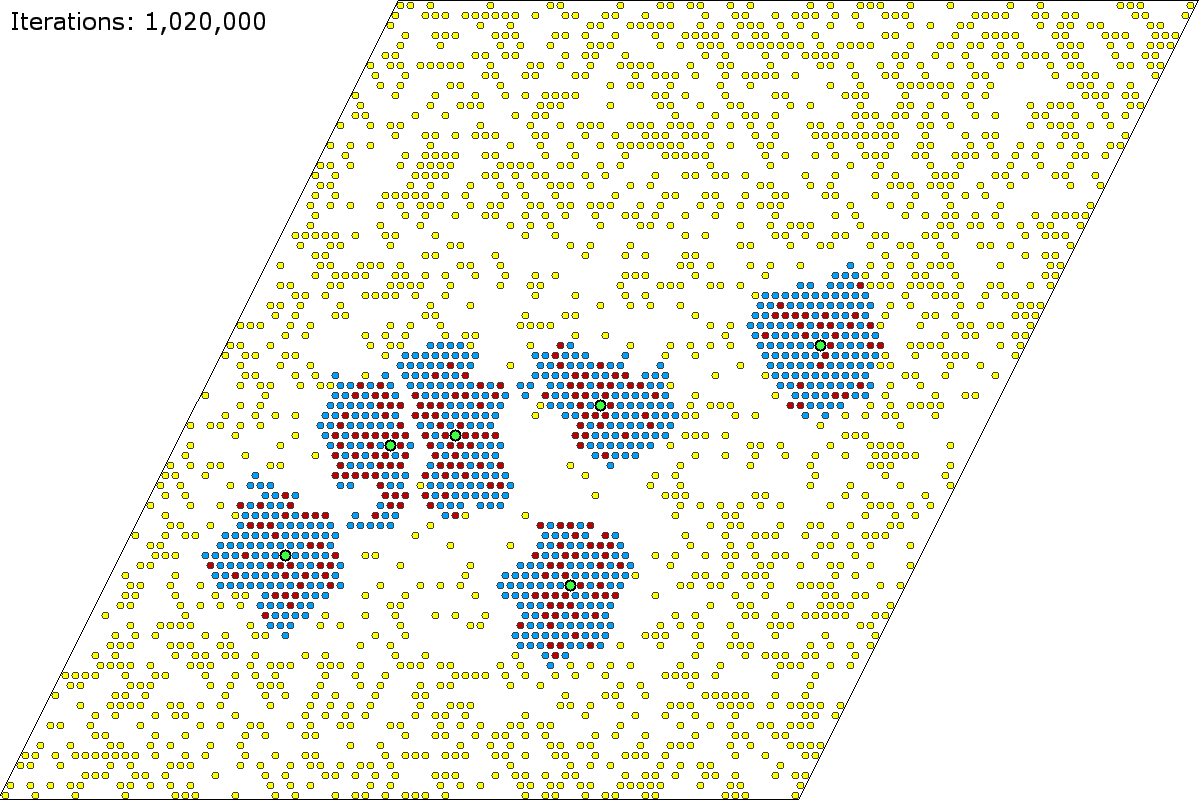}
  \end{center}
  \caption{Six food particles placed.}
\end{subfigure}\\
~\\
\begin{subfigure}[b]{.5\linewidth}
  \begin{center}
  \includegraphics[width=\linewidth]{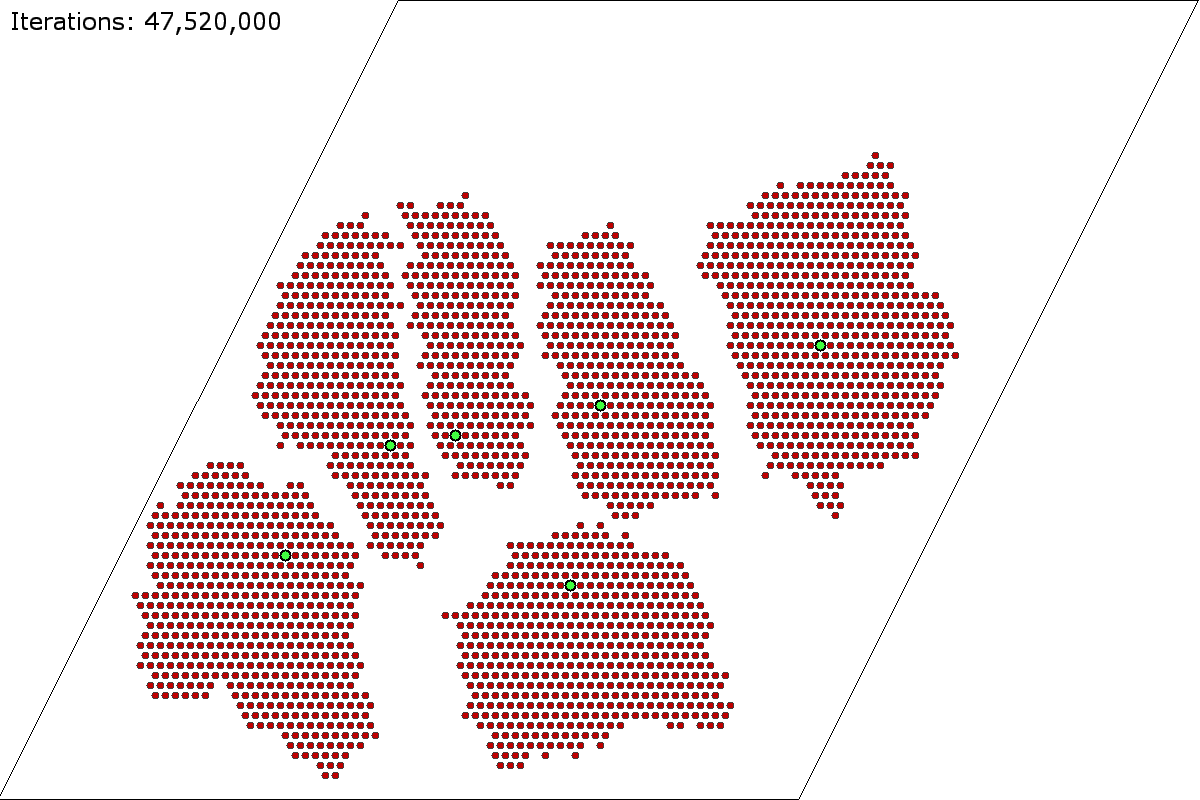}
  \end{center}
  \caption{All particles in compression state.}
\end{subfigure}%
\begin{subfigure}[b]{.5\linewidth}
  \begin{center}
  \includegraphics[width=\linewidth]{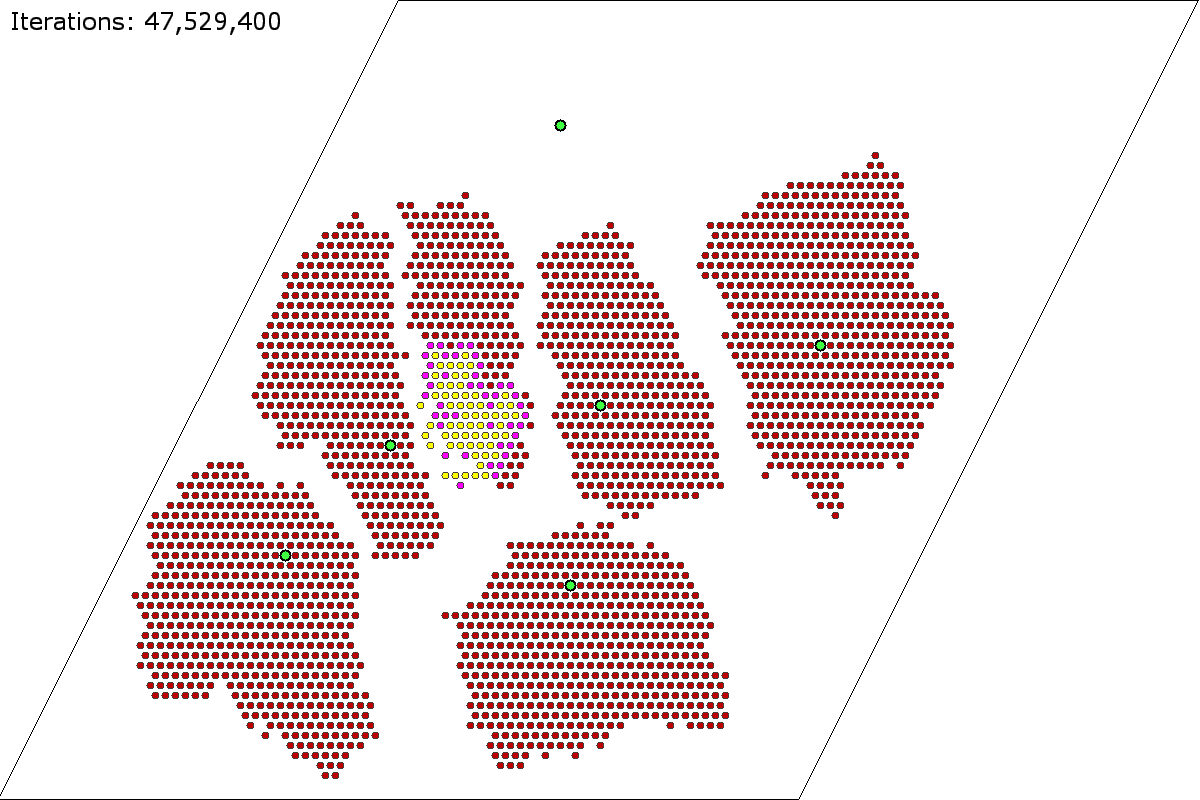}
  \end{center}
  \caption{One food particle moved, dispersion starts.}
\end{subfigure}\\
~\\
\begin{subfigure}[b]{.5\linewidth}
  \begin{center}
  \includegraphics[width=\linewidth]{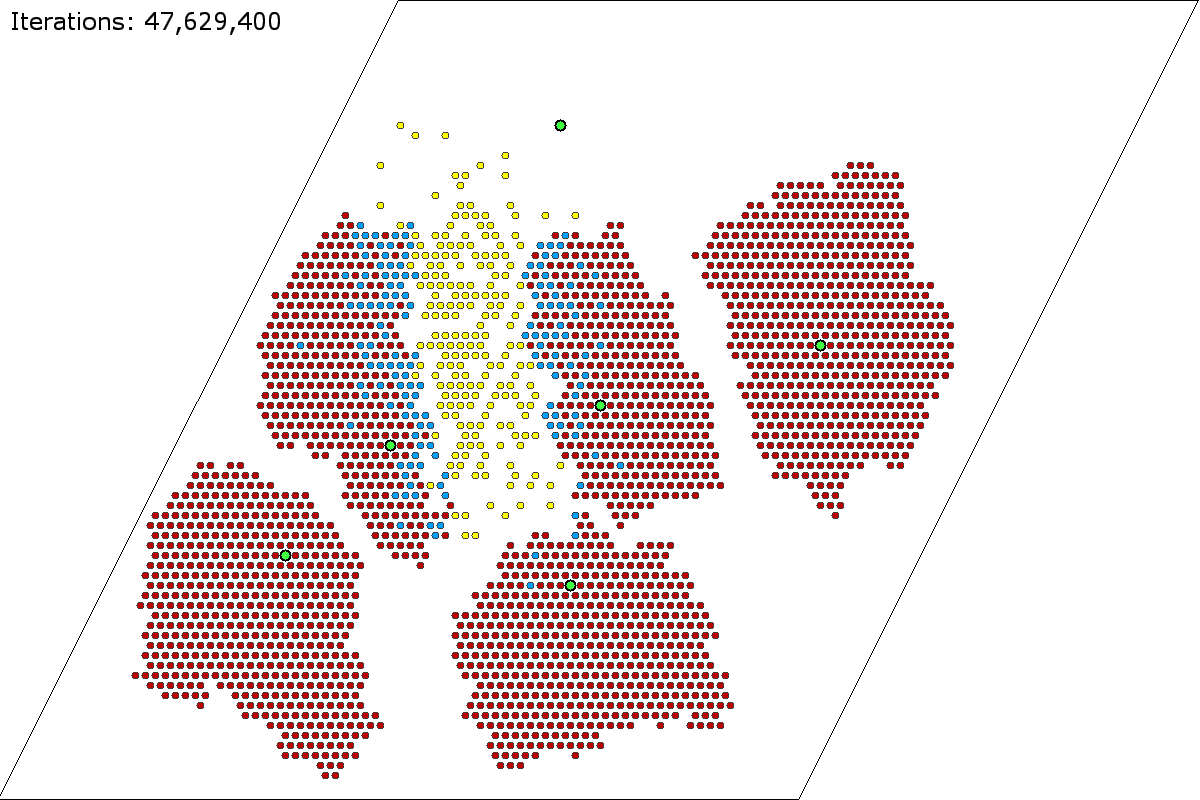}
  \end{center}
  \caption{Cluster fully dispersed.}
\end{subfigure}%
\begin{subfigure}[b]{.5\linewidth}
  \begin{center}
  \includegraphics[width=\linewidth]{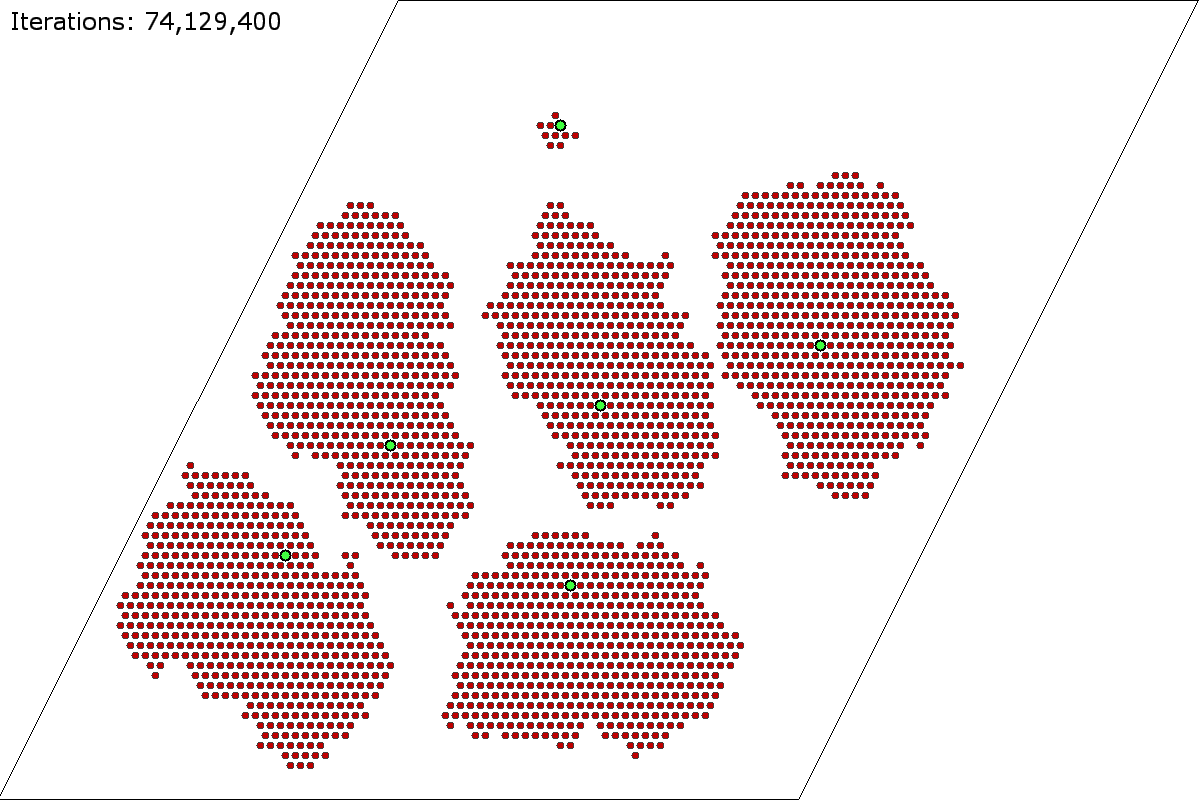}
  \end{center}
  \caption{All particles have rejoined other clusters.}
\end{subfigure}
\caption{Simulation of Adaptive $\alpha$-Compression with multiple food particles. The images are in chronological order and the simulation continues in Figure~\ref{fig:simulation2}. Dispersion state particles are in yellow, compression state particles are in red with a growth ticket and in blue without, and dispersion token state particles are in purple. The food particle is drawn slightly larger and is in green.}
\label{fig:simulation1}
\end{figure}
\begin{figure}[!ht]
\begin{subfigure}[b]{.5\linewidth}
  \begin{center}
  \includegraphics[width=\linewidth]{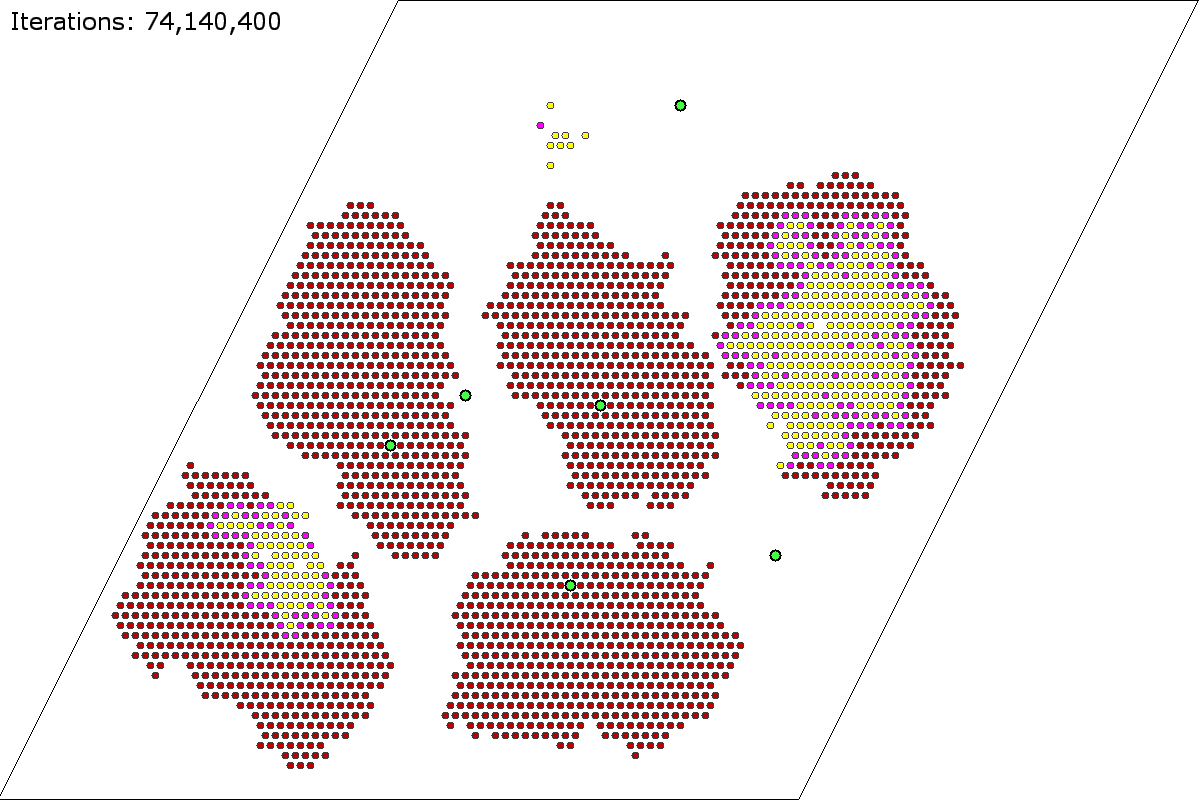}
  \end{center}
  \caption{Three food particles moved, dispersion starts.}
\end{subfigure}%
\begin{subfigure}[b]{.5\linewidth}
  \begin{center}
  \includegraphics[width=\linewidth]{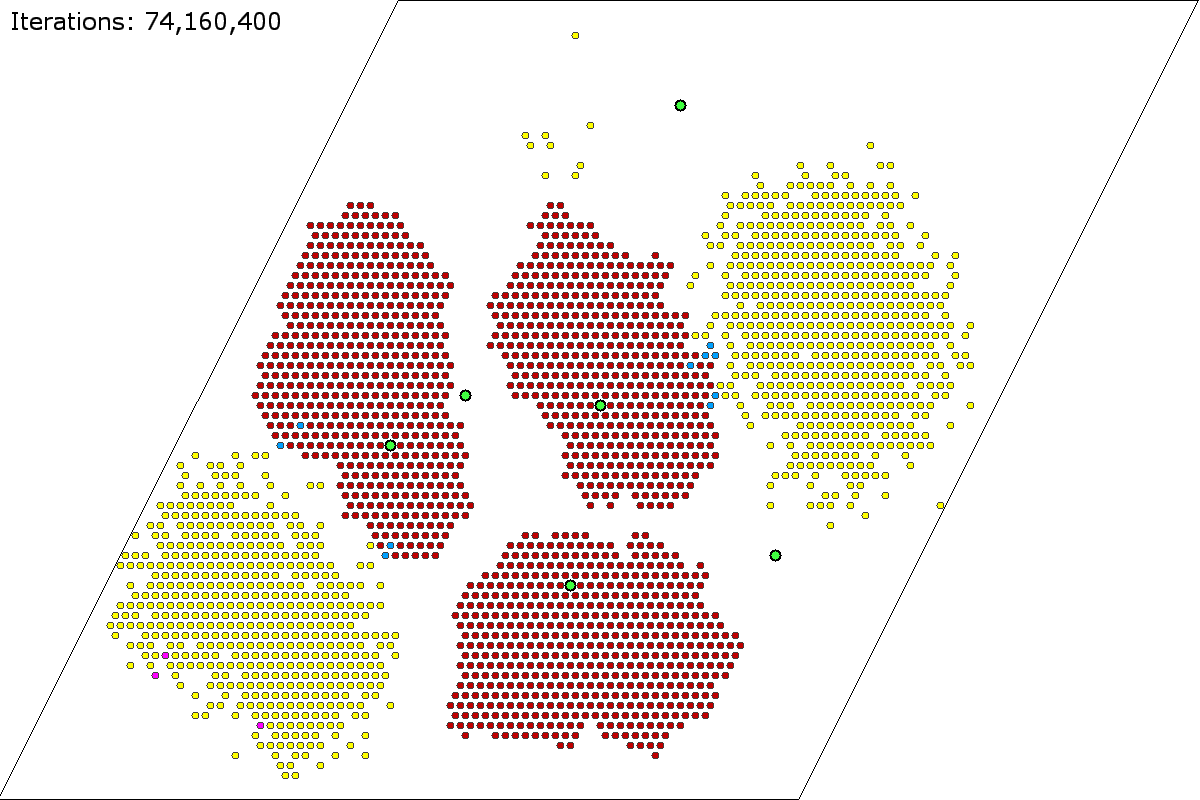}
  \end{center}
  \caption{Three clusters fully dispersed.}
\end{subfigure}\\
~\\
\begin{subfigure}[b]{.5\linewidth}
  \begin{center}
  \includegraphics[width=\linewidth]{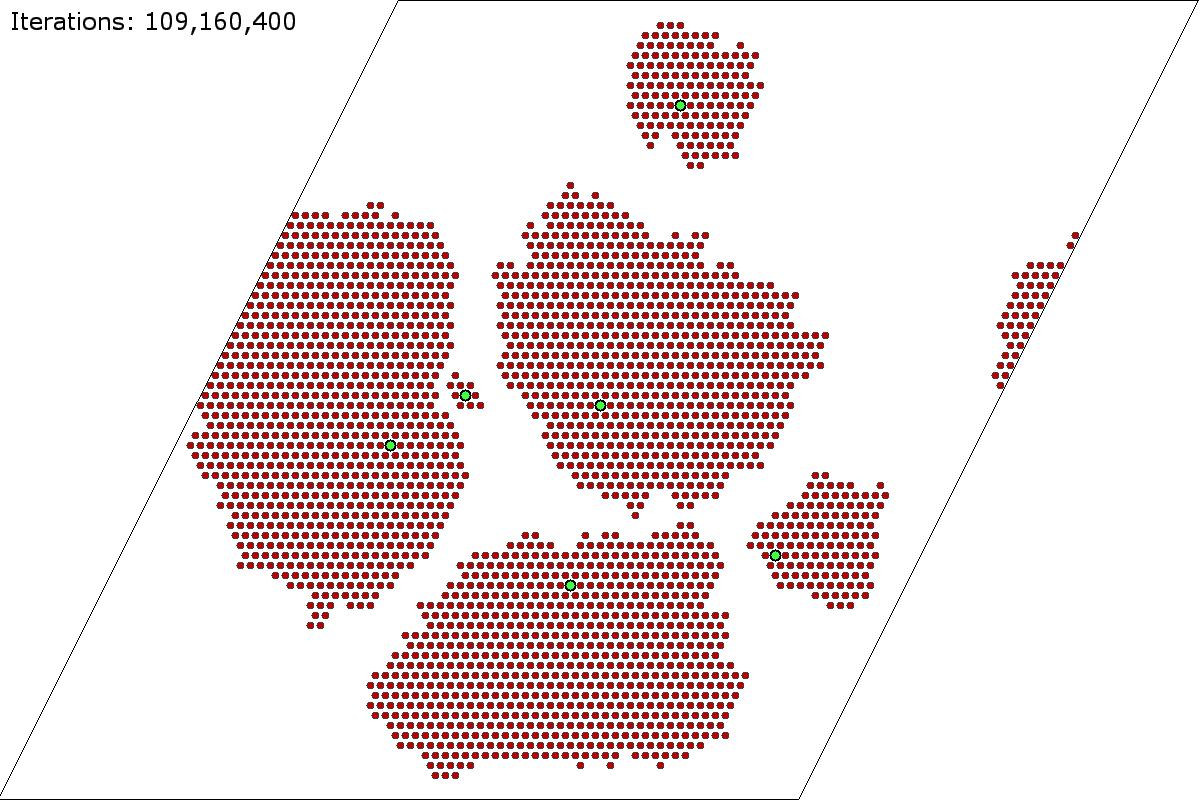}
  \end{center}
  \caption{All particles have rejoined other clusters.}
\end{subfigure}%
\begin{subfigure}[b]{.5\linewidth}
  \begin{center}
  \includegraphics[width=\linewidth]{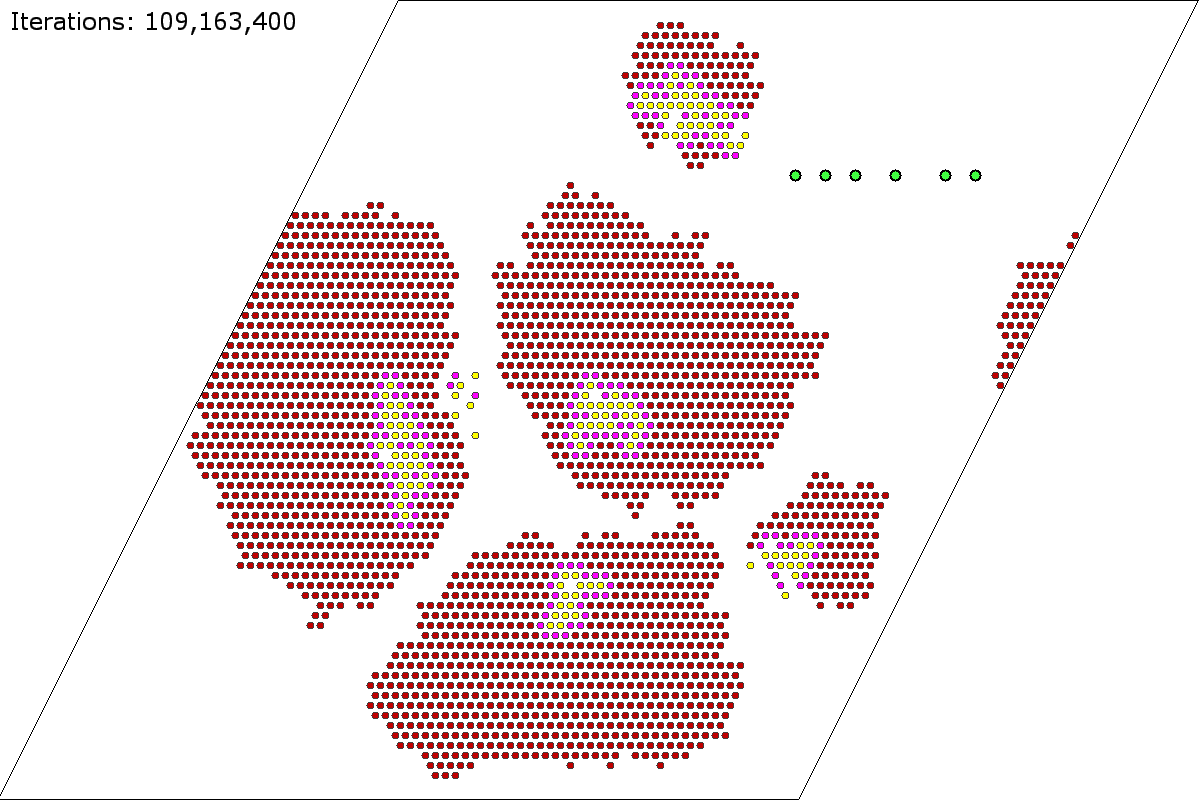}
  \end{center}
  \caption{All food particles moved, dispersion starts.}
\end{subfigure}\\
~\\
\begin{subfigure}[b]{.5\linewidth}
  \begin{center}
  \includegraphics[width=\linewidth]{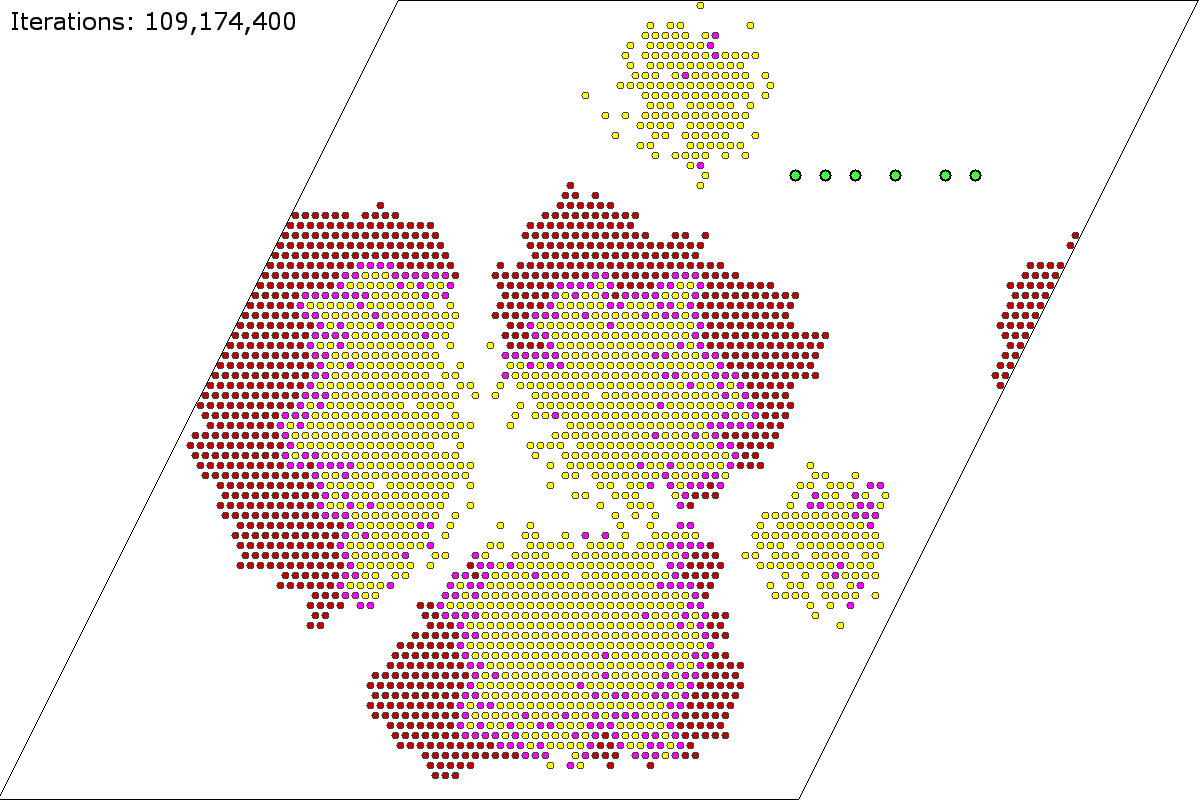}
  \end{center}
  \caption{Dispersion in progress.}
\end{subfigure}%
\begin{subfigure}[b]{.5\linewidth}
  \begin{center}
  \includegraphics[width=\linewidth]{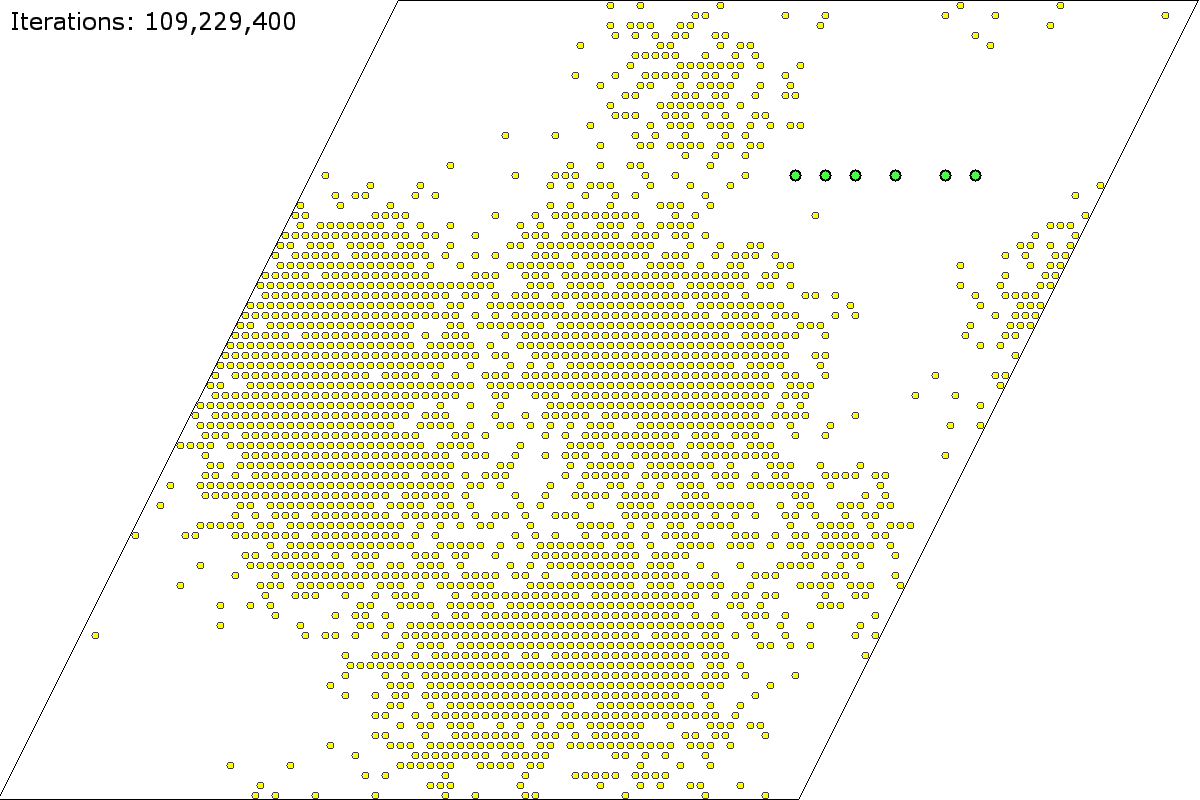}
  \end{center}
  \caption{All particles fully dispersed.}
\end{subfigure}
\caption{Continuation of the simulation of the Adaptive $\alpha$-Compression algorithm in Figure~\ref{fig:simulation1}.}
\label{fig:simulation2}
\end{figure}

\end{appendices}

\end{document}